\documentclass[sigconf, nonacm]{acmart}
\pdfoutput=1

\def\BibTeX{{\rm B\kern-.05em{\sc i\kern-.025em b}\kern-.08em
    T\kern-.1667em\lower.7ex\hbox{E}\kern-.125emX}}

\usepackage{booktabs} 
\usepackage[hide]{chato-notes}
\usepackage[inline]{enumitem}
\usepackage{balance}
\usepackage[review]{switch-version} 

\setcopyright{rightsretained}

\acmDOI{}

\acmISBN{978-3-89318-086-8}

\acmConference[EDBT 2022]{25th International Conference on Extending Database Technology (EDBT)}{29th March-1st April, 2022}{Edinburgh, UK} 
\acmYear{2022}

\newcommand{\hide}[1]{\iffalse{#1}\fi}

\usepackage{graphicx}
\usepackage{balance}  
\usepackage{xspace}
\usepackage{booktabs}
\usepackage{multirow}
\usepackage{algorithm}
\usepackage{algorithmicx}
\usepackage[noend]{algpseudocode}
\usepackage{tablefootnote}
\usepackage{hyperref}
\usepackage{float}

\usepackage{tikz}
\usetikzlibrary{arrows}
\usetikzlibrary{decorations.markings}
\usetikzlibrary{arrows.meta}
\usepackage{filecontents}
\usepackage{xcolor}
\definecolor{darkorchid}{rgb}{0.6,0.2,0.8}
\usetikzlibrary{calc,arrows,shapes}

\makeatletter
\newcommand{\StatexIndent}[1][3]{%
  \setlength\@tempdima{\algorithmicindent}%
  \Statex\hskip\dimexpr#1\@tempdima\relax}
\algdef{S}[WHILE]{WhileNoDo}[1]{\algorithmicwhile\ #1}%
\algdef{S}[IF]{IfNoDo}[1]{\algorithmicif\ #1}%

\makeatother

\newtheorem{problem}{Problem}

\newcommand{\sosp}{{\sf\small SOSP}\xspace} 
\newcommand{\mosp}{{\sf\small MOSP}\xspace} 
\newcommand{\lrdp}{{\sf\small LRDP}\xspace} 
\newcommand{\budp}{{\sf\small BUDP}\xspace} 
\newcommand{\ouralgorithm}{{\sf\small PEANUT}\xspace} 
\newcommand{\ouralgorithmplus}{{\sf\small PEANUT}\ensuremath{+}\xspace} 

\newcommand{\indsep}{{\sf\small INDSEP}\xspace} 
\newcommand{\qtm}[1]{{VE-{#1}}\xspace}
\newcommand{\supplink}{https://tinyurl.com/579nfvr7}

\newcommand{\variables}{\ensuremath{X}\xspace}
\newcommand{\freq}{\ensuremath{f}\xspace}
\newcommand{\querylog}{\ensuremath{\mathcal{Q}}\xspace}
\newcommand{\vertices}{{\ensuremath{V}}\xspace}
\newcommand{\spot}{\ensuremath{S}\xspace}
\newcommand{\spotvertices}[1]{\ensuremath{\vertices(#1)}\xspace}
\newcommand{\cutS}{\ensuremath{cut(S)}\xspace}
\newcommand{\jtroot}{\ensuremath{r}\xspace}
\newcommand{\treespot}{\ensuremath{T_S}\xspace}

\newcommand{\tkp}{{\sf TKP}\xspace}
\newcommand{\tkpprofit}{\ensuremath{p}\xspace}
\newcommand{\tkpweight}{\ensuremath{w}\xspace}

\newcommand{\knapsack}{\ensuremath{{k_{0}}}\xspace}
\newcommand{\troot}{\ensuremath{t}\xspace}
\newcommand{\targetprofit}{\ensuremath{{p_{0}}}\xspace}
\newcommand{\targetbenefit}{\ensuremath{b}\xspace}
\newcommand{\musep}{\ensuremath{\mu_\mathit{sep}}\xspace}
\newcommand{\mures}{\ensuremath{\mu_\mathit{res}}\xspace}
\newcommand{\alphasep}{\ensuremath{\alpha_\mathit{sep}}\xspace}
\newcommand{\alphares}{\ensuremath{\alpha_\mathit{res}}\xspace}
\newcommand{\kp}{{\sf KP}\xspace}

\newcommand{\forwardmove}{\ensuremath{\textsc{ForwardMove}}\xspace}
\newcommand{\backwardmove}{\ensuremath{\textsc{BackwardMove}}\xspace}
\newcommand{\childs}{\ensuremath{\mathit{Children}}\xspace}
\newcommand{\currentchilds}{\ensuremath{\mathit{currChildren}}\xspace}
\newcommand{\child}{\ensuremath{\mathit{child}}\xspace}
\newcommand{\jtleafs}{\ensuremath{\mathit{leaves}}\xspace}
\newcommand{\jtree}{\ensuremath{T}\xspace}
\newcommand{\parent}{\ensuremath{\pi}\xspace}
\newcommand{\nodebenefit}[1]{\revisioncol{\ensuremath{b^\querylog({#1})}\xspace}}
\newcommand{\nodecost}[1]{\ensuremath{c({#1})}\xspace}
\newcommand{\bopt}[1]{\ensuremath{b^{*}\!({#1})}\xspace}
\newcommand{\tovisit}{\ensuremath{\mathit{toVisit}}\xspace}
\newcommand{\tovisitpop}{\ensuremath{\mathit{toVisit.pop}}\xspace}
\newcommand{\pop}{\ensuremath{\mathit{pop}}\xspace}
\newcommand{\tovisitsize}{\ensuremath{\mathit{toVisit.size}}\xspace}
\newcommand{\currentpath}{\ensuremath{\mathit{currPath}}\xspace}
\newcommand{\currentpathsize}{\ensuremath{\mathit{currPath.size}}\xspace}

\newcommand{\jtpath}{\ensuremath{\mathit{path}}\xspace}
\newcommand{\treepath}{\ensuremath{\mathit{path}_{\tree}}\xspace}
\newcommand{\spotset}{\ensuremath{\mathcal{S}}\xspace}

\newcommand{\NP}{\ensuremath{\mathbf{NP}}}
\newcommand{\NPhard}{{{\NP}-hard}}

\newcommand{\bayesianNet}{\ensuremath{\mathcal{N}}\xspace}

\newcommand{\query}{\ensuremath{q}\xspace}

\newcommand{\tree}{\ensuremath{T}\xspace}

\newcommand{\nodes}{\ensuremath{V}\xspace}
\newcommand{\edges}{\ensuremath{E}\xspace}

\newcommand{\steinertree}{\ensuremath{\tree_\query}\xspace}

\newcommand{\revisioncol}[1]{{#1}} 

\newcommand{\budget}{\ensuremath{K}\xspace}

\newcommand{\cardbudget}{\ensuremath{k}\xspace}

\newcommand{\netprob}[2]{\ensuremath{\mathit{Pr}_{#1}\!\left({#2}\right)}\xspace}

\newcommand{\isuseful}[3]{\ensuremath{\delta_{#1}\!\left({#2}\right)}\xspace}

\newcommand{\benefit}[1]{\ensuremath{B\!\left({#1}\right)}\xspace}

\newcommand{\spara}[1]{\vspace{1mm}\noindent{\bf{#1}}}

\newcommand{\squishlist}{\begin{list}{$\bullet$}
  {\setlength{\itemsep}{0pt}
    \setlength{\parsep}{3pt}
    \setlength{\topsep}{3pt}
    \setlength{\partopsep}{0pt}
    \setlength{\leftmargin}{1.5em}
    \setlength{\labelwidth}{1em}
    \setlength{\labelsep}{0.5em}}}
\newcommand{\squishend}{
\end{list}}

\begin{document}
\title{Workload-Aware Materialization of Junction Trees}

\author{Martino Ciaperoni}
\affiliation{%
	\institution{Aalto University}
    \city{Espoo}
	\country{Finland}
}
\email{martino.ciaperoni@aalto.fi}

\author{Cigdem Aslay}
\affiliation{%
	\institution{Aarhus University}
	\city{Aarhus}
	\country{Denmark}
}
\email{cigdem@cs.au.dk}

\author{Aristides Gionis}
\affiliation{%
	\institution{KTH Royal Institute of Technology}
	\city{Stockholm}
	\country{Sweden}
}
\email{argioni@kth.se}

\author{Michael Mathioudakis}
\affiliation{%
	\institution{University of Helsinki}
	\city{Helsinki}
	\country{Finland}
}
\email{michael.mathioudakis@helsinki.fi}


\begin{abstract}
Bayesian networks are popular probabilistic models that capture the conditional dependencies among a set of variables. 
Inference in Bayesian networks is a fundamental task for answering probabilistic queries over a subset of variables in the data.  
However, exact inference in Bayesian networks is \NP-hard, 
which has prompted the development of many practical inference methods. 

In this paper, we focus on improving the performance of the junction-tree algorithm, a well-known method for exact inference in Bayesian networks.
In particular, we seek to leverage information in the workload of probabilistic queries 
to obtain an optimal workload-aware materialization of junction trees, 
with the aim to accelerate the processing of inference queries. 
We devise an optimal pseudo-polynomial algorithm to tackle this problem 
and discuss approximation schemes. 
Compared to state-of-the-art approaches for efficient processing of inference queries via junction trees, 
our methods are the first to exploit the information provided in query workloads.
Our experimentation on several real-world Bayesian networks confirms the effectiveness of our techniques 
in speeding-up query processing. 
\end{abstract}

\maketitle

\section{Introduction}
\label{sec:introduction}

Bayesian networks are probabilistic graphical models 
that represent a set of variables and their conditional dependencies
via a directed acyclic graph. 
They are powerful models for answering probabilistic queries over variables in the data, 
and making predictions about the likelihood of subsets of variables when other variables are observed. 
Inference in Bayesian networks is a fundamental task with applications in a variety of domains, including  machine learning~\citep{bishop2006pattern} and probabilistic database management~\citep{deshpande2009graphical}.
However, the problem of exact inference in Bayesian networks is \NPhard~\citep{pearl2014probabilistic}.
This challenge has prompted a large volume of literature that aims to develop practical inference algorithms, 
both exact~\citep{chavira2005compiling, chavira2007compiling, dechter1999bucket, darwiche2003differential, lauritzen1988local, poon2011sum, zhang1994simple} and approximate~\citep{henrion1988propagating}. 


 
\begin{figure}[t]
	\begin{tabular}{c}
		\hspace{-6mm} \tikzset{
	treenode/.style = {shape=circle,  draw, align=center, top color=white},
	root/.style     = {treenode, font=\ttfamily\normalsize, bottom color=blue!20},
	clique/.style      = {treenode, font=\ttfamily\normalsize,  bottom color=blue!20},
	separator/.style    = {rectangle,draw,font=\ttfamily\normalsize,  top color=white!20,  bottom color=white!20}
}
\begin{tikzpicture}[
scale = 0.7, 
every node/.style={scale=0.01},
grow=right,
sibling distance=150mm,
level distance = 200mm,
every node/.style = {shape=rectangle, rounded corners,
	draw, align=center,
	top color=white},
every node/.style  = {font=\normalsize},
decoration={
	markings,
	mark=at position 1 with {\arrow[scale=2,black]{stealth}};
},
sloped]

    \node[shape=circle,draw, align=center,top color=white, font=\ttfamily\normalsize, bottom color=blue!20] (l) at (1, -1){l};
    \node[shape=circle,draw, align=center,top color=white, font=\ttfamily\normalsize, bottom color=blue!20] (i) at 
    (-0.6,0) {i};
    \node[shape=circle,draw, align=center,top color=white, font=\ttfamily\normalsize, bottom color=blue!20] (g) at (1,1) {g};
    \node[shape=circle,draw, align=center,top color=white, font=\ttfamily\normalsize, bottom color=blue!20] (h) at (2.8,1) {h};
    \node[shape=circle,draw, align=center,top color=white, font=\ttfamily\normalsize, bottom color=blue!20] (e) at (2.8,-1) {e};
    \node[shape=circle,draw, align=center,top color=white, font=\ttfamily\normalsize, bottom color=blue!20] (f) at (4.6,-1) {f};
    \node[shape=circle,draw, align=center,top color=white, font=\ttfamily\normalsize, bottom color=blue!20] (d) at (6.4,-1) {d};
    \node[shape=circle,draw, align=center,top color=white, font=\ttfamily\normalsize, bottom color=blue!20] (c) at (4.6,1) {c};
    \node[shape=circle,draw, align=center,top color=white, font=\ttfamily\normalsize, bottom color=blue!20] (b) at (6.4,1) {b};
    \node[shape=circle,draw, align=center,top color=white, font=\ttfamily\normalsize, bottom color=blue!20] (a) at (8,0) {a};

    \draw[-{Latex[length=1.5mm,width=1.8mm]}](i) edge node[right] {} (g);
    \draw[-{Latex[length=1.5mm,width=1.8mm]}](b) edge node[right] {} (a);
    \draw[-{Latex[length=1.5mm,width=1.8mm]}](a) edge node[right] {} (d);
    \draw[-{Latex[length=1.5mm,width=1.8mm]}](l) edge node[right] {} (g);
    \draw[-{Latex[length=1.5mm,width=1.8mm]}](i) edge node[right] {} (l);
    \draw[-{Latex[length=1.5mm,width=1.8mm]}](g) edge node[right] {} (e);
    \draw[-{Latex[length=1.5mm,width=1.8mm]}](g) edge node[right] {} (h);
    \draw[-{Latex[length=1.5mm,width=1.8mm]}](h) edge node[right] {} (e); 
    \draw[-{Latex[length=1.5mm,width=1.8mm]}](e) edge node[right] {} (f); 
    \draw[-{Latex[length=1.5mm,width=1.8mm]}](e) edge node[right] {} (c); 
    \draw[-{Latex[length=1.5mm,width=1.8mm]}](c) edge node[right] {} (b); 
    \draw[-{Latex[length=1.5mm,width=1.8mm]}](b) edge node[right] {} (d); 
\end{tikzpicture}  \\
		\textbf{(a)} 
	      \vspace*{0.15cm}\\
		 \hspace{-6mm} \tikzset{
  treenode/.style = {shape=circle,
                     draw, align=center,
                     top color=white},
  basel/.style  = {rectangle,draw, font=\ttfamily\normalsize,fill=red!20}        
  root/.style     = {treenode, font=\ttfamily\normalsize, bottom color=blue!20},
  clique/.style      = {treenode, font=\ttfamily\normalsize,  bottom color=blue!20},
  query_clique/.style      = {treenode, font=\ttfamily\normalsize,  bottom color=red!20},
  separator/.style    = {rectangle,draw,font=\ttfamily\normalsize,  bottom color=white!20}
}


\begin{tikzpicture}[
scale = 0.45, 
every node/.style={scale=0.1},
grow = left,
sibling distance=9em,
level distance = 7em,
  every node/.style = {shape=rectangle,
    draw, align=center,
    top color=white}]

  \node [clique]{bc}
    child { node [separator]  {b} 
        child { node [clique] (q_1) {abd} }}
    child { node [separator] {c} 
     child { node [clique] {ce} 
        child { node [separator] {e}
        child { node [clique] {ef} }}
        child { node [separator] {e} 
        child { node  [query_clique] {egh} 
        child { node [separator] {g} 
        child { node [clique] {gil}}}}}}};

{\tiny 9}\end{tikzpicture}

  \\
		\textbf{(b)} 
	\end{tabular}
\caption{
	 \label{fig:diagram} Example: \textbf{(a)} a simple Bayesian network, and \textbf{(b)} corresponding junction tree. An in-clique query $\query = \{{\sf g,h}\}$ can be answered via marginalization from the joint probability distribution associated with the clique node in red.}\vspace*{-0.45cm}
\end{figure}
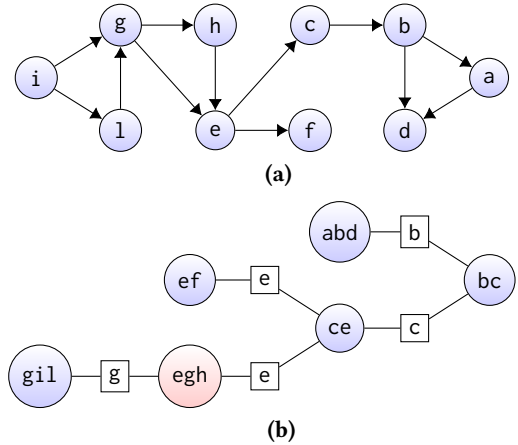


A state-of-the-art method for exact inference in Bayes\-ian networks
is the junction-tree algorithm~\citep{lauritzen1988local}, 
which enables simultaneous execution of a large class of inference queries. 
%
Specifically, the main idea behind the junction-tree algorithm is to convert the Bayesian network into a tree, 
called junction tree,  
and pre\-compute a collection of 
joint probability distributions for selected subsets of variables, 
referred to as \emph{cliques} of the tree.
This pre\-computation allows to answer inference queries 
that involve variables captured by a clique, as demonstrated in Figure~\ref{fig:diagram}. 
However, it does not allow for direct answer of queries that are not captured by a single clique.
Such \emph{out-of-clique} inference queries are instead answered via a message-passing algorithm
over the junction tree, which often demand a large amount of computations.
For instance, in a Bayesian network fully specified by approximately $10^3$ parameters, it can take more than a minute to complete the message passing procedure for some queries in our experiments. 	
To reduce the computational burden associated with processing out-of-clique queries, 
\citet{kanagal2009indexing} proposed to 
materialize suitable joint probability distributions, 
each corresponding to a partition of the junction tree.
In this paper, we follow the approach of \citet{kanagal2009indexing}, and take it one step further by considering a query workload. As in the work of Kanagal and Deshpande, we materialize joint probability distributions corresponding to partitions of the junction tree.
Departing from their approach, however, we choose the distributions to materialize in a \emph{workload-aware} manner, and not only based on the structure of the junction tree.
Leveraging the information that is available in a query workload
allows our method to be optimized for specific applications and be fine-tuned for settings of interest:  
in principle, for different query workloads, our method will choose different distributions to materialize, leading to optimal workload processing time. 
To the best of our knowledge, we are the first to address the task of identifying the \emph{optimal} 
\emph{workload-aware} materialization to speed-up the processing of arbitrary inference queries over junction trees.

More concretely, we make the following contributions.
\squishlist
\item We define the problem of \emph{workload-aware materialization} of junction trees 
as a novel optimization problem (Section~\ref{sec:setting}). 
The general problem of optimally materializing multiple shortcut potentials given a space budget (\mosp)
uses\revisioncol{,} as a subtask\revisioncol{,} the problem of materializing a single optimal shortcut potential (\sosp). 
Thus, we define and solve this simpler problem first as a special case, 
before discussing the general case. 
\item We prove that both problems are {\NPhard} (Section~\ref{sec:complexity}). 
\item 
We propose \ouralgorithm (Sections~\ref{sec:single_sp_alg}--\ref{section:peanut}), a method to optimize query processing based on junction trees through materialization. 
The {offline} component of \ouralgorithm utilizes pseudo-polynomial dynamic-programming algorithms to find the optimal materialization for a given workload and under a budget constraint. 
The {online} component of \ouralgorithm identifies the materialized distributions to exploit during the processing of a query, translating to reduced message-passing cost, and hence, improved query-answering time. 
\item
In addition to the pseudo-polynomial algorithms,
we provide  
an approximate strongly-polynomial algorithm (Section~\ref{sec:approximation_alg}) 
to improve the efficiency of the {offline} component of \ouralgorithm.
\item
Our optimal workload-aware materialization algorithm selects 
\emph{disjoint} shortcut potentials.
This constraint makes the problem tractable, but in practice it is not necessary.
Furthermore, it leads to solutions that do not fully utilize the available space budget.
To overcome this limitation, 
we also propose \ouralgorithmplus (Section~\ref{section:overlap}), which, 
based on a simple but effective greedy heuristic, permits overlapping shortcut potentials to be materialized,
and leads to better utilization of the available budget. 
\ouralgorithmplus is the best performing method in practice, and the method of~choice.

\item We evaluate \ouralgorithm and \ouralgorithmplus empirically over commonly-used benchmarks (Section~\ref{sec:experiments}) 
and show that enhancing a junction tree with the proposed materialization\revisioncol{,} for large enough budget\revisioncol{,} can save on average between $20\%$ and $40\%$ of query-processing cost, 
while typically offering an improvement of $2$ orders of magnitude over the previous work of \citet{kanagal2009indexing}. 
\squishend

\section{Related Work}
Recent years have witnessed an increasing amount of uncertain and correlated data 
generated in a variety of application areas. 
Uncertain data are collected into probabilistic databases, 
which can be efficiently represented as probabilistic graphical models \citep{deshpande2009graphical}.
Thus, the problem of querying large probabilistic databases 
can be formulated as an inference problem in probabilistic graphical models. 
\revisioncol{Although the junction tree algorithm and our materialization can be used for inference on different graphical models, such as Markov random fields, in this work we restrict our attention to Bayesian networks, which are also useful in a variety of tasks beyond querying probabilistic databases.}
For instance, \citet{getoor2001selectivity} 
resort to Bayesian networks to provide selectivity estimates 
for typical relational queries. 
\revisioncol{
		Lately, significant progress has been made towards exploring the connection between data management and graphical models. Khamis et al.~\cite{abo2016faq,khamis2017juggling} propose \emph{FAQ}, a unifying framework for a class of problems sharing the same algebraic structure, which encompasses the processing of the queries considered in our work. 
		Khamis et al.\ develop a dynamic programming algorithm for the general \emph{FAQ} problem and introduce a notion of \emph{FAQ}-width to characterize its complexity.
		In 2019, Schleich et al.~\cite{schleich2019layered} propose \emph{LMFAO}, an in-memory optimized execution engine for multiple queries over a relational database  
		modelled as a junction tree. 
		On the one hand, their work is orthogonal to ours since it does not consider materialization of additional probability distributions not directly captured by the junction tree structure. However, \emph{LMFAO} also identifies, for each query, the direction of the message passing which leads to the lowest computational cost. A similar optimization will be investigated for \ouralgorithm in future work. 
}
The conceptually simplest algorithm for exact inference over Bayesian networks
is variable elimination \citep{zhang1994simple, zhang1996exploiting}.
In our recent work, we have developed workload-aware materialization techniques for the variable-elimination inference method \cite{aslay2021workload}. 
Execution of variable elimination typically involves the computation of marginal-distribution tables, 
which have to be re-computed every time inference is performed.
The junction-tree algorithm \citep{jensen1996introduction, lauritzen1988local}, closely intertwined with variable elimination,
attempts to turn this observation into its advantage by pre\-computing and materializing 
several distributions for different subsets of variables.
The advantage of such pre\-computation is that, 
if an inference query involves only variables within a materialized distribution, as it is the case for all single-variable queries,
then the query can be answered directly via marginalization from that distribution.
It should be noted that, although 
 the junction-tree algorithm allows to perform inference on arbitrary Bayesian networks, 
if no restriction is posed to the tree structure, 
inference may become infeasible. 
In particular, the feasibility of the method depends on the junction-tree \emph{treewidth}, 
which is defined as the maximum number of variables in a materialized distribution minus 1. 
For inference queries that cannot be directly answered from the tree through marginalization, 
a message-passing algorithm needs to be performed, which may be extremely computationally expensive. 
To alleviate this issue, \citet{kanagal2009indexing} 
propose a disk-based hierarchical index structure for the junction tree. 
They also find that\revisioncol{,} by pre\-computing and materializing additional joint probability distributions\revisioncol{,}
it is possible to prune a considerable amount of computations at query time. 
Our approach also relies on extending the materialization of the junction tree, 
but in a workload-aware manner. 

\ReviewOnly{\revisioncol{Finally, note that} o}\FullOnly{O}ur approach is based on a dynamic-programming optimization framework. 
The general idea of framing the selection of a materialization as an optimization problem to be tackled by means of algorithmic techniques was first introduced by \citet{chaudhuri1995optimizing}.

\begin{figure*}[t]
	\begin{tabular}{ccc}
		 \hspace{-10mm}  \tikzset{
  treenode/.style = {shape=circle, 
                     draw, align=center,
                     top color=white},
  root/.style     = {treenode, font=\ttfamily\small, bottom color=blue!20},
  clique/.style      = {treenode, font=\ttfamily\small,  bottom color=blue!20},
   query_clique/.style      = {treenode, font=\ttfamily\small,  bottom color=red!20},
  separator/.style    = {rectangle,draw,font=\ttfamily\small,  bottom color=white!20}, 
  strom/.style={line width=1pt, shorten >=2pt,shorten <=2pt},
         gueterstrom/.style={strom,->},
         geldstrom/.style={strom,->}
}


\begin{tikzpicture}[
scale = 0.56, 
grow = left,
sibling distance=6.6em,
level distance = 5.9em,
  every node/.style = {shape=rectangle,
    draw, align=center,
    top color=white}]
    \node (e) [query_clique]{bc}
    child { node [separator] {b} 
        child { node [clique] {abd} }}
    child { node [separator] {c} 
     child { node (d) [clique] {ce} 
        child { node [separator] {e}
        child { node (c) [query_clique] {ef} }}
        child { node [separator] {e} 
        child { node (b) [clique] {egh} 
        child { node [separator] {g} 
        child { node (a) [query_clique] {gil}}}}}}}; 
        
    \draw[-{Stealth[length=1.8mm,width=2.2mm]}] (a) edge[bend right=45,font=\ttfamily\footnotesize	] node[below] {$m_1(g,i)$} (b);
    \draw[-{Stealth[length=1.8mm,width=2.2mm]}] (c) edge[bend left=40,font=\ttfamily\footnotesize	] node[above] {$m_3(e,f)$} (d);
    \draw[-{Stealth[length=1.8mm,width=2.2mm]}] (b) edge[bend right=65,font=\ttfamily\footnotesize	] node[below] {$m_2(e,i)$} (d);
    \draw[-{Stealth[length=1.8mm,width=2.2mm]}] (d) edge[bend right=70,font=\ttfamily\footnotesize	] node[below] {$m_4(c,f,i)$} (e);
   
\end{tikzpicture}

 &
		  \hspace{-12mm} \tikzset{
  treenode/.style = {shape=circle,
                     draw, align=center,
                     top color=white},
  root/.style     = {treenode, font=\ttfamily\small, bottom color=blue!20},
  clique/.style      = {treenode, font=\ttfamily\small,  bottom color=blue!20},
  sp/.style     = {treenode, font=\ttfamily\small, bottom color=black!20},
  separatorsp/.style     = {rectangle,draw, font=\ttfamily\small, bottom color=black!20},
  query_clique/.style      = {treenode, font=\ttfamily\small,  bottom color=red!20},
  separator/.style    = {rectangle,draw,font=\ttfamily\small,  bottom color=white!20}, 
  strom/.style={line width=1pt, shorten >=2pt,shorten <=2pt},
         gueterstrom/.style={strom,->},
         geldstrom/.style={strom,->}
}


\begin{tikzpicture}[
scale = 0.56, 
grow = left,
sibling distance=6.6em,
level distance = 5.9em,
  every node/.style = {shape=rectangle,
    draw, align=center,
    top color=white}]
    \node (e) [query_clique]{bc}
    child { node [separator] {b} 
        child { node [clique] {abd} }}
    child { node [separator] {c} 
     child { node (d) [sp] {ce} 
        child { node [separator] {e}
        child { node (c) [query_clique] {ef} }}
        child { node [separatorsp] {e} 
        child { node (b) [sp] {egh} 
        child { node [separator] {g} 
        child { node (a) [query_clique] {gil}}}}}}}; 
        
   
\end{tikzpicture}

 &
		   \hspace{-2mm} 
\tikzset{
  treenode/.style = {shape=circle, 
                     draw, align=center,
                     top color=white},
  root/.style     = {treenode, font=\ttfamily\small, bottom color=blue!20},
  sp/.style     = {treenode, font=\ttfamily\small, bottom color=black!20},
  clique/.style      = {treenode, font=\ttfamily\small,  bottom color=blue!20},
  query_clique/.style      = {treenode, font=\ttfamily\small,  bottom color=red!20},
  separator/.style    = {rectangle,draw,font=\ttfamily\small,  bottom color=white!20}, 
  strom/.style={line width=1pt, shorten >=2pt,shorten <=2pt},
 gueterstrom/.style={strom,->},
 geldstrom/.style={strom,->}
}


\begin{tikzpicture}[
scale = 0.56, 
grow = left,
sibling distance=6.6em,
level distance = 5.9em,
  every node/.style = {shape=rectangle,
    draw, align=center,
    top color=white}]]
  \node (e) [query_clique]{bc}
    child { node [separator] {b} 
        child { node [clique] {abd} }}
    child { node [separator] {c} 
    child { node (s) [sp] {ecg} 
    child { node [separator] {e} 
    child { node (d) [query_clique] {ef} }}
    child { node [separator] {g} 
    child { node (a) [query_clique] {gil} }}}};

    \draw[-{Stealth[length=1.8mm,width=2.2mm]}] (a) edge[bend right=65,font=\ttfamily\footnotesize] node[below] {$m_1(g,i)$} (s);
    \draw[-{Stealth[length=1.8mm,width=2.2mm]}] (d) edge[bend left=45,font=\ttfamily\footnotesize] node[above] {$m_2(e,f)$} (s);
    \draw[-{Stealth[length=1.8mm,width=2.2mm]}] (s) edge[bend right=70,font=\ttfamily\footnotesize] node[below] {$m_4(c,f,i)$} (e);
 
\end{tikzpicture}

 \\ 
		\textbf{(a)} &  \textbf{(b)} &  \textbf{(c)} 
	\end{tabular}
	\caption{\label{fig:diagram2} 
	Example: message passing to answer out-of-clique query $q = \{{\sf b,i,f}\}$, 
	In {(a)}, no materialization is used. 
	In {(b)}, the subtree \treespot associated with a shortcut potential \spot is 
	coloured in grey. 
	In {(c)}, materialization of \spot is used.  }
\end{figure*}
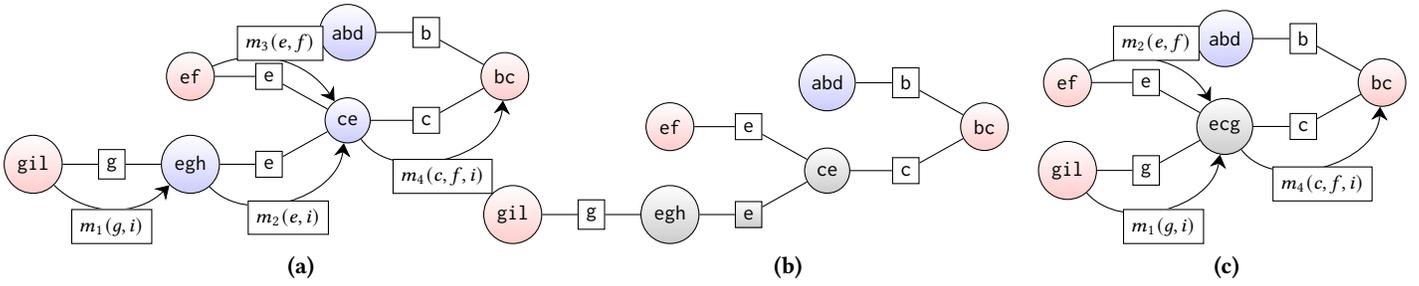

\section{Problem formulation}
\label{sec:setting}

We develop optimal 
workload-aware materialization for the junction-tree algorithm, 
which is used for exact inference on Bayesian networks.

In this section, we formally define the optimization problems we study.
To make the paper self-contained and introduce the necessary notation, 
we begin with a brief overview of the junction-tree algorithm in Section~\ref{sec:background}.
For a more elaborate presentation of the junction-tree algorithm, 
we refer the reader to the classic papers of \citet{jensen1996introduction}
and \citet{lauritzen1988local}.
After describing the junction-tree algorithm, 
we proceed with the formal definition of the optimization problems in Section~\ref{sec:problem}

\subsection{Background on junction trees}
\label{sec:background}

The junction-tree algorithm performs exact inference on Bayesian networks based on the \emph{junction tree} data structure.  
%
In what follows, we provide a brief description of  Bayesian networks, 
the junction-tree structure, and finally the junction-tree inference algorithm.

\spara{Bayesian networks.} 
Bayesian networks are probabilistic graphical models that represent the joint distribution of a set of variables.
More formally, a Bayesian network \bayesianNet is a directed acyclic graph, where nodes represent variables and directed edges 
represent dependencies of variables on their parents. 
%
%
A Bayesian network allows to express the joint probability of all variables as a product of a finite number of factors, so that each factor corresponds to the conditional probability of a variable given the value of its parents.

In our presentation\revisioncol{,} we consider discrete Bayesian networks, 
i.e., all variables are assumed categorical. 
In practice, numerical variables can be handled as categorical via discretization 
into categorical intervals of sufficient precision. 

\spara{Junction trees.} 
Given a Bayesian network \bayesianNet over a set of variables~\variables, a junction tree $\tree=(\vertices,\edges)$ is built from the Bayesian network in five steps: 
(1) moralization, (2) triangulation, (3) clique-graph formation, (4) junction-tree extraction, and 
(5) junction-tree calibration. 
The result of this procedure is the \emph{junction tree} of the Bayesian network \bayesianNet. The nodes of \tree are referred to as \emph{clique nodes} and the edges as \emph{separators}. 
%
Each \emph{clique node} is assigned a set of Bayesian network factors, and their product is used to initialize the \emph{clique potential}, which is a table mapping each configuration of the associated variables to a non-negative real value. 
%
%
In what follows, we will refer to the variables associated with a clique node $u$ or a separator $(u,v)$ 
as the \emph{scope} of the clique node or the separator, and 
we will denote them with $\variables_u$ and $\variables_{u,v}$, respectively.
In the last step of the junction tree construction, 
the junction tree \tree is \emph{calibrated} via the \emph{Hugin} algorithm \citep{cowell2006probabilistic}. 
In short, this means that the clique potentials of the junction tree are normalized to coincide with the joint distributions of the related variables. 
%
By construction, 
a junction tree satisfies the \emph{running-intersection property}, 
which states that
if $a \in \variables_u$ and $a \in \variables_v$ for two nodes $u$ and $v$ in \tree, 
then $a \in \variables_w$ for all nodes $w$ that lie on the path from $u$ to $v$,
which we denote by $\treepath(u,v)$. 
This property is crucial for inference (described later in this section).
Moreover, one node $r\in \nodes$ of the junction tree \tree 
is selected and marked as \emph{pivot} (root), and is the node towards which all messages are sent. 
The choice of the pivot node affects the order and the size of the set of messages sent to process an inference query.  
Selecting a pivot node that is optimal for all possible queries is not possible. 
For our purposes, we consider an arbitrary node to be the pivot. For future work, however, it would be interesting to study the problem of finding the materialization that is optimal across all pivot selections. 
%
The junction-tree algorithm is used to answer inference queries on the Bayesian network 
using the junction-tree structure defined above.
An inference query $\query\subseteq\variables$ is defined by a subset of variables in \variables and asks for the joint or conditional probability distribution of the variables appearing in \query. 
In what follows, we focus only on joint-probability queries, as conditional probabilities can be obtained from the answer to joint-probability queries. 
Inference queries can be separated into two types, \emph{in-clique} and \emph{out-of-clique} queries.
In-clique queries correspond to cases where all the variables of the query \query 
are associated with the same clique node of the junction tree.
The junction-tree algorithm answers in-clique queries by directly marginalizing 
the joint probability of the variables corresponding to the clique node.
For example, in Figure~\ref{fig:diagram}, the query $\query = \{{\sf g,h}\}$ is answered by marginalizing the calibrated potential of the clique node ${\sf egh}$.
Out-of-clique queries correspond to cases where not all variables are associated with the same clique node.
For out-of-clique queries, a \emph{message-passing} procedure is invoked
over a Steiner tree $\steinertree$, the smallest tree that connects all the clique nodes containing the query variables.
%
When performing message passing, we always consider the direction of the messages induced by the pivot \jtroot of~\tree. As a consequence, the pivot (root) node $\jtroot_\query$ of $\steinertree$ (i.e., the clique node towards which all messages are sent) is defined to be  the clique node closest to the pivot \jtroot of the junction tree and therefore, $\jtroot_\query = \jtroot$ if \jtroot is contained in the Steiner tree $\steinertree$.  
%
%
Messages are sent towards the pivot $\jtroot_\query$ in discrete steps from one node to the next, starting from the leaves of $\steinertree$. 
%
The message sent from a given clique node $u$ to the next is the product of its own potential with the messages it receives from other nodes.
Once the pivot $\jtroot_\query$ has received messages from all its neighbours in $\steinertree$,
it computes the answer to the query~\query by summing out all the non-query variables.%
\footnote{%
In a practical implementation, an equivalent but more efficient approach is to sum-out variables as soon as possible --- i.e., marginalization is performed for each message before it is sent to the next clique node, to compute the marginal distribution only of the variables belonging to the query and the separator over which the message is sent, since all the other variables are redundant and would later be summed out.}
An example of the described message-passing procedure is given in Figure~\ref{fig:diagram2}(a). The query variables $\query = \{{\sf b,i,f}\}$ are not all included in the same clique. Thus, the Steiner tree $\steinertree$ that contains the cliques  having all the query variables within their scopes (highlighted in red), is extracted. 
The pivot of both the junction tree and the Steiner tree is ${\sf bc}$\revisioncol{,}
and hence the leaves ${\sf gil}$ and ${\sf ef}$ begin the message passing. 
After message $m_4({\sf c,f,i})$ is sent to the pivot, 
we have a potential corresponding to the joint probability distribution of variables $\{{\sf b,c,f,i}\}$ 
from which we sum out variable ${\sf c}$ to obtain the query answer, 
namely the joint probability distribution of variables ${\sf b,i}$ and ${\sf f}$.  





\subsection{Problem statement} 
\label{sec:problem}

At a high level, our objective is to minimize the expected running time for queries 
drawn from the same distribution as of queries in a given workload.
The approach we take is to materialize \emph{shortcut potentials}, i.e., 
additional data structures introduced by \citet{kanagal2009indexing}. 
Depending on the query, shortcut potentials may be utilized by the junction-tree algorithm to skip part 
of the computations.
Naturally, this approach introduces a trade-off between the benefit in reduced running time, on one hand, 
and the cost of storing the materialized data structures, on the other.
In what follows, we define shortcut potentials, we quantify the cost and benefit of their materialization, and 
finally, we formally define the optimization problems we consider.

\spara{Shortcut potentials.} 
To speed-up out-of-clique query processing in junction trees, \citet{kanagal2009indexing} propose to store certain joint probability distributions,  
the so-called \emph{shortcut potentials}, 
which allow to substantially reduce the amount of computations. 
%
A shortcut potential \spot is identified by a subtree $\treespot \subseteq \tree$ of the junction tree with vertices \spotvertices{S} and is defined as the joint distribution of all variables in the scope of those separators that ``cut'' the subtree~\treespot from the junction tree \tree. 
This set of separators is henceforth denoted as \cutS. 
Moreover, as for any other subtree, the pivot (root)~$\jtroot_S$ of \treespot is defined to be the clique node closest to the pivot~\jtroot of the junction tree \tree --- 
and therefore, $\jtroot_S = \jtroot$ if \jtroot is contained in~\spotvertices{S}.
For example, in Figure~\ref{fig:diagram2}(b), the highlighted shortcut potential is identified by the subtree consisting of clique nodes \texttt{egh} and \texttt{ce}; and defined as the joint distribution of the separator variables \texttt{c}, \texttt{e}, \texttt{g}.
Shortcut potentials allow to reduce the size of the Steiner tree $\steinertree$ extracted to compute the answer of \query, thus, decreasing the associated message-passing cost. 
For example, consider the processing of query $q = \{\texttt{b},\texttt{i},\texttt{f}\}$ on the junction tree in Figure~\ref{fig:diagram2}(a). By utilizing the shortcut potential \spot that corresponds to the highlighted subtree \treespot in Figure~\ref{fig:diagram2}(b), we perform message passing on a smaller but valid Steiner tree and avoid the computation of $m_2(\texttt{e},\texttt{i})$, as shown in Figure~\ref{fig:diagram2}(c).

\spara{Cost of a shortcut potential.}
The cost or weight of a shortcut potential captures the amount of storage space it takes when materialized.  
Formally, we define the cost as the size $\mu(S)$  of the probability-distribution table \spot. This quantity is given by the product of the cardinalities of all the variables in the shortcut-potential scope $X_S$.
The cost is upper bounded by the product of the cardinalities of the variables in the separators $(u,v) \in cut(S)$. 

\spara{Benefit of a shortcut potential.}
Intuitively, the notion of benefit is meant to capture the expected savings in running time achieved by the materialization of a shortcut potential.
For a shortcut potential to have a positive benefit, it has to be \emph{useful}  
(i.e., to be possible to utilize it by the junction-tree algorithm) 
for some queries in the query workload.
Therefore, before we quantify the notion of \emph{benefit}, let us define the notion of \emph{usefulness}. 
In plain words, a shortcut potential \spot is useful for a query when replacing the nodes of \spotvertices{\spot} with the shortcut potential \spot allows the junction-tree algorithm to compute the same result for query \query but more efficiently.
The precise conditions under which this is possible depend on whether or not the shortcut potential subtree includes the pivot of the Steiner tree, i.e.,  $\jtroot_\query \in \spotvertices{S}$. 
On the one hand, if $\jtroot_\query \in \spotvertices{S}$,we need at least a path from the leaves of \steinertree to \jtroot that passes through a separator in \cutS.
On the other hand, if $\jtroot_\query \notin \spotvertices{S}$, we need at least two separators of \cutS in the same path between any leaf of \steinertree and \jtroot. In both cases, moreover, no query variables should be left out of \steinertree when $V(S)$ is replaced with \spot. 
Formally, we have the following definition.
\begin{definition}[Usefulness] 
\label{definition:usefulness}
Consider a query \query with associated Steiner tree $\steinertree$ and a shortcut potential \spot. 
We say that \spot is \emph{useful} for \query if \emph{one} of the following conditions holds: 
\begin{enumerate}[label=(\roman*)]
\item  
$\jtroot_\query \notin \spotvertices{S}$ and there are at least two separators 
$\{(u,v), (w,z)\} \in \cutS$ in the path between any leaf of $\steinertree$ and $\jtroot_\query$; or
\item  
$\jtroot_\query \in \spotvertices{S}$ and there is at least one separator $(u,v) \in \cutS$ in the path between any leaf of $\steinertree$ and $\jtroot_\query$;
\end{enumerate}
and in addition replacing $V(S)$ in $\steinertree$ with \spot results in a Steiner tree that still contains all the query variables while leading to a lower message-passing cost.
We define: 
\[ \isuseful{\spot}{\query} 
&= 
\left  \{
\begin{array}{ll}
1 & \text{if \spot is useful for \query}, \\
0 & \text{otherwise}.
\end{array}
\right.
\]
\end{definition}

To quantify the benefit of materializing a shortcut potential\revisioncol{,}
we propose a definition of \emph{shortcut-potential benefit}, 
which relies on the definition of usefulness. 
Intuitively, the benefit of a shortcut potential for a query 
should reflect the amount of message-passing operations in $T_S$ 
that are now avoided thanks to the shortcut potential $S$. 
Such amount can be computed based on the size of the clique potentials and the cardinality of the query variables. This motivates the following definition. 
\begin{definition}
The benefit of a shortcut potential \spot with respect to a query \query is defined as
\begin{align*}
B(S, \query) =  \isuseful{\spot}{\query}\
					\!\sum_{v \in V(S)} \mu(v) 
					\!\!\prod_{w \in \variables_{T_v} \cap \query}\!\! \alpha(w),
\end{align*}
\end{definition}

In the above definition, $\mu(v)$ denotes the size of the clique potential of a clique node $v$,
$\alpha(w)$ denotes the cardinality of a variable~$w$, and 
$X_{T_v}$ is union of the scopes of all the cliques in the subtree rooted at clique $v$.
The benefit of a shortcut potential \spot with respect to a query log~\querylog should take into account not only the benefit with respect to individual queries\revisioncol{,} but also query probabilities to 
guarantee that queries that are most likely to occur will be assigned a higher weight 
in the computation of benefit. 
\begin{definition}[Benefit] 
The benefit of a shortcut potential \spot with respect to a query log \querylog is defined as 
\begin{equation*}
	B(S, \querylog) = \sum_{ \query \in \querylog }  \netprob{_\querylog}{\query} B(S, \query).
\end{equation*}
\label{def:benefit}
\end{definition}
Here, $\netprob{_\querylog}{\query}$ is the probability of query \query being drawn from the query log \querylog, which in practice can be estimated from the frequency $\freq(\query)$ of the available queries.

\spara{Materializing a single optimal shortcut potential.} 
We build on the work of \citet{kanagal2009indexing} 
to further improve the efficiency of query processing using junction trees. 
We use a simple idea: take the anticipated query workload into account. 
One may have a precise idea of the anticipated query workload --- 
for example when a large historical query log \querylog is available, 
from which relative frequencies of different queries are known, 
possibly indicating that some queries are far more likely than others. 
Alternatively, there may be uncertainty about the anticipated workload. 
Even in the extreme case in which no historical query log is available, 
one may optimize materialization for an ``uninformative'' (e.g., uniform) distribution of queries. 
In all cases, the available information about possible queries 
can guide the selection of the shortcut potentials to materialize. 
Indeed, it is expected that some regions of the junction tree will be relevant for far more queries than others --- or associated with a higher volume of computation 
(i.e., larger messages in the junction-tree algorithm). 
%
Previous work, however, is agnostic to queries. 
Therefore\revisioncol{,} we introduce here the problem of choosing\revisioncol{,} for a given $\jtroot_S$\revisioncol{,} a shortcut potential of optimal benefit with respect to a given query workload, under a user-specified budget constraint. Formally: 

\begin{problem}[Single Optimal Shortcut Potential (\sosp)]
\label{problem:single_sp}
Consider a junction tree $\tree=( \vertices,\edges )$, with pivot $r$,
a query log \querylog,
and space budget~\budget.
We are asked to find a single shortcut potential \spot, with pivot $r_S \in \vertices$,
such that the benefit  
$\benefit{S,\querylog}$
is maximized, subject to the constraint
$\mu(\spot)\le \budget$.
\end{problem}

In simple words, Problem~\ref{problem:single_sp} seeks the shortcut potential 
that would avoid the maximum volume of message-passing operations for query log \querylog, 
under the constraint that it can be materialized within a given space budget.%
\footnote{We considered also an alternative objective function that subtracts $\mu(S)$.
from the benefit $B(S, \query)$. 
This small correction in the objective accounts 
for the number of message-passing operations performed by the junction-tree algorithm 
with materialization.
However: 
($i$)~this would make the optimization problems less tractable; and 
($ii$)~the cost $\mu(S)$ is already taken into account in the budget constraint. 
For these reasons, we opted to optimize the benefit as defined in the text. 
Nevertheless, in our experiments, we do compare the exact computational cost of algorithms. 
}
The solution to Problem~\ref{problem:single_sp} is a shortcut potential rooted at $r_S$.
\ReviewOnly{\revisioncol{Clearly, our optimization problem relies on the simple assumption that the workload is stationary.}}

\spara{Materializing multiple optimal shortcut potentials.} 
In practice, it would be desirable to materialize not just one, 
but any number of shortcut potentials within a given space budget.
We consider this problem while restricting our attention to \emph{non-overlapping} shortcut potentials, 
i.e., node-disjoint subtrees \treespot of \tree.
For this problem, as shown in the next section, 
we can show that a dynamic-programming algorithm provides the exact solution. 
More formally, we consider the following problem.

\begin{problem}[Multiple Optimal Shortcut Potentials (\mosp)] 
\label{problem:k_sp}
Consider a junction tree $\tree=(\vertices,\edges)$ with pivot $r$, 
a query log \querylog, 
and space budget~\budget.
We are asked to 
find a set of node-disjoint shortcut potentials $\spotset = \{\spot_1, \spot_2, \ldots, \spot_{\cardbudget}\}$, 
for some $\cardbudget \geq 1$, 
such that the total~benefit
\begin{equation*}
\sum_{i=1}^\cardbudget  \benefit{S_i,\querylog} 
\label{eq:sosp}
\end{equation*}
is maximized, subject to the constraints~
$\sum_{i=1}^\cardbudget \mu(\spot_i) \le \budget$  ~and \\
$V(S_i) \cap V(S_j) = \emptyset$ for all  $i \neq j$ with $i,j = 1,\ldots,\cardbudget$.
\end{problem}

Note that the previous work of \citet{kanagal2009indexing} 
considers materializing a hierarchical set of shortcut potentials such that some shortcut potential subtrees are nested into others. 
\FullOnly{While our selection of shortcut potentials is more restricted than theirs due to the exclusion of overlaps, 
the workload-aware nature of our approach is sufficient to lead to superior empirical performance.}
\ReviewOnly{\revisioncol{		
While the shortcut potentials we select are optimal only under the assumption of exclusion of overlaps and stationarity of the query workload distribution, the workload-aware nature of our approach is sufficient to lead to superior empirical performance.
}}

\section{Complexity and algorithms}
\label{sec:algorithms}

In this section\ReviewOnly{\revisioncol{,}} we first prove that the problems \sosp and \mosp, 
defined in the previous section, are both \NP-hard. 
We then present \ouralgorithm, a method to optimize inference based on junction trees. 
\ouralgorithm is composed of an \emph{offline} and an \emph{online} component. 

The {offline} component tackles the \sosp and \mosp problems; the corresponding algorithms are described in Sections~\ref{sec:single_sp_alg} and~\ref{sec:k_sp_alg}, respectively. 
Both algorithms are based on dynamic programming, 
they run in \emph{pseudo-polynomial time}, and provide an optimal solution. 
Hereafter, it is convenient to consider \tree as rooted at the pivot $r \in V$. 
The algorithm for the \mosp problem, named \budp, follows a standard bottom-up traversing scheme, while  
the algorithm for the \sosp problem, named \lrdp, visits the nodes of the tree in a ``left-to-right'' order.
Both algorithms require two dynamic-programming passes, 
one to compute the benefit of the optimal solution, 
and one to obtain the set of separators that are part of the optimal solution.  
Additionally, in Section~\ref{sec:approximation_alg}\ReviewOnly{\revisioncol{,}}
we present a \emph{strongly-polynomial algorithm}
that provides a trade-off between execution time and solution quality.
The {online} component of \ouralgorithm follows the
standard junction-tree-based inference approach
and it is described in Section~\ref{section:peanut}.
Finally, in Section~\ref{section:overlap} we describe \ouralgorithmplus, 
a practical extension of \ouralgorithm. 

\subsection{Problem complexity}
\label{sec:complexity}

We now establish the complexity of the problems defined in Section~\ref{sec:problem}.
We start with the problem of selecting a single optimal shortcut potential (\sosp).

\begin{theorem}
\label{theorem:sosp-hardness}
The problem of selecting a single optimal shortcut potential
{\sf (\sosp)} is \NP-hard. 
\end{theorem}


\begin{proof} 	
	We prove the hardness of the \sosp problem 
	by a reduction from the \emph{tree-knapsack problem} (\tkp), 
	which is known to be \NP-hard~\citep{johnson1983knapsacks}. 
	In the \tkp problem we are given a tree $R =(U,A)$ with root $\troot\in U$, 
	where each node $u \in U$ is associated with a weight $\tkpweight(u)$ and a profit $\tkpprofit(u)$,
	a knapsack capacity $\knapsack$, and a target profit $\targetprofit$. 
	The problem asks whether there exists a subtree $R' = (U', A')$ with root \troot, 
	such that $\sum_{u \in U'} \tkpweight(u) \leq \knapsack$ and 
	$\sum_{u \in U'} \tkpprofit(u) \geq \targetprofit$. 
	Similarly, in the decision version of the \sosp problem 
	we are given a tree $\jtree= (V,E)$ with root $\jtroot\in V$, 
	cost and benefit functions $\mu(S)$ and $B(S, \querylog)$ 
	associated  with each shortcut-potential subtree $T(S) = (V(S), E(S))$, 
	space budget \budget, and target benefit \targetbenefit.
	The cost function $\mu(S)$ is determined solely from the function $\mu(i,j)$ 
	associated with the separators (edges) in $\cutS$. 
	The benefit function instead depends both on the size $\mu(v)$ of the clique potentials 
	for all cliques $v \in V(S)$ as well as on \querylog. 
	Thus, both functions $\mu$ and~$B$ can be specified by encodings 
	that are polynomial with the size of the tree \jtree.
	The \sosp problem asks to determine whether there exists a sub\-tree $T(S)= (V(S), E(S))$ 
	with root \jtroot, 
	with cost $\mu(S) \leq \budget$ and benefit $B(S, \querylog) \ge\targetbenefit$.
	
	We show next how an instance of \tkp can be transformed into an instance of \sosp.
	First, given $R =(U,A)$ in \tkp, 
	we take the vertices of \jtree for \sosp to be $V =  U \cup W$, 
	where $W$ is a set of clique nodes consisting of a child for each leaf in $R$ 
	and of the parent of the root $\troot\in U$.
	The root \jtroot of \jtree is set to $\parent_\troot$, the parent of the root in~$R$.
	Notice that there is a properly defined shortcut potential 
	with sub\-tree $T(S) \subseteq \jtree$ for each $R' \subseteq R$,
	i.e., the two problems have the same space of candidate solutions. 
	Let us now discuss how to set separators $(i,\parent_i)$,
	with $\parent_i$ being the parent of $i$ in the \sosp instance. 
	Once the scope of separator $(i, \parent_i)$ is specified for all $i \in V \setminus \jtroot$, 
	all $\mu(S)$ values are fixed because $\mu(S)$ is, by definition, 
	equal to the product of the cardinalities of all variables that 
	are in the union of the scopes of the separators in~\cutS. 
	We assign to the scope $\variables_{(i,\parent_i)}$ one variable $x_v$ for each clique $v$ in the path $\jtpath(\parent_i, \jtroot)$ with cardinality $\alpha(x_v) = e^{\tkpweight(v)}$.  
	The union of all cliques belonging to paths $\jtpath(\parent_i, \jtroot)$ for each $(i, \parent_i) \in \cutS$  is precisely $\spotvertices{S}$. 
	Let the budget for \sosp problem instance be $\budget = e^\knapsack$. It follows that
	\begin{align*}
	\mu(\spot) = \prod_{x \in \variables_\spot} \alpha(x) = \prod_{v \in \spotvertices{\spot}} e^{\tkpweight(v)}= e^{\sum_{v \in \spotvertices{\spot} } \tkpweight(v)  }
	\end{align*}
	implying that
	\begin{align*}
	\mu(\spot) = e^{\sum_{v \in \spotvertices{\spot} } \tkpweight(v)  } \leq e^\knapsack  = \budget
	~~\text{  if and only if  }~
	\sum_{v \in V(S)} \tkpweight(v) \leq \knapsack,
	\end{align*}
	that is, the budget constraint in the \sosp instance is satisfied 
	if and only if the knapsack constraint in the \tkp instance is satisfied for the same subtree.

	Next, 
	we discuss how to map profits in \tkp to benefits in \sosp.
	As already discussed, the benefit function is determined 
	by \querylog and $\mu(v)$, for each node $v \in V$. 
	We consider a query log \querylog composed of a single query \query, 
	such that there is a query variable in each clique $d \in W$ and no query variables in the other cliques. 
	Query \query has probability $\netprob{_\querylog}{\query} = 1$ 
	and it is useful for subtrees of \jtree that include only nodes in $U$. 
	Furthermore, we assume that all query variables $w$ in the cliques $d \in W$ have cardinality $\alpha(w)= 1$.
	It can be easily verified that the benefit of any subtree $T(S)$ is given by $\sum_{v \in V(S)} \mu(v)$. 
	Notice that $\mu(v)$ can be expressed as a product of two terms. 
	One term is the product $\musep(v)$ of the cardinality of all variables 
	that are in separators incident to clique $v$. 
	This value is already fixed by the computation of the subtree costs $\mu(S)$ previously described. 
	The other term, still unspecified, is the product $\mures(v)$ of the cardinalities of the variables 
	that belong to clique $v$ but not to the separators incident to $v$. 
	Therefore, we can set 
	$\alphares(v) = \tkpprofit(v)/\alphasep(v)$. 
	This makes the benefit of a given subtree $T(S)$ equivalent to the profit of the same subtree in \tkp, 
	that is, $B(T(S), \querylog) = \sum_{ v \in V(S) } \tkpprofit(v)$.  
	
	The \sosp problem instance can be clearly constructed in polynomial time. 
	A YES instance for \tkp implies 
	a YES instance for \sosp, 
	and the opposite is also true. 
	This concludes the proof. 
\end{proof}

A similar complexity result holds for the more general problem 
of finding multiple node-disjoint optimal shortcut potentials (\mosp).

\begin{theorem}
\label{theorem:mosp-hardness}
The problem of finding multiple node-disjoint optimal shortcut potentials
{\sf (\mosp)} is \NP-hard. 
\end{theorem}

\begin{proof}
	We can prove the hardness of the \mosp problem 
	by a reduction from the $0$--$1$ knapsack problem (\kp), 
	a well-known \NP-hard problem~\citep{toth1990knapsack}.
	In \kp we are given a set of $n$ items, 
	each with a weight $w_i$ and profit $v_i$, 
	a knapsack capacity \knapsack, 
	and a target profit \targetprofit, and 
	we ask whether there exists a subset of distinct items 
	such that the total weight does not exceed the knapsack capacity~\knapsack,
	while the total profit of the selected items is at least \targetprofit.
	
	Similarly, in the decision version of the \mosp problem 
	we are given a tree  $\jtree=(V,E)$ with root $\jtroot\in V$, 
	a cost and benefit functions $\mu(S)$ and $B(S, \querylog)$ 
	associated with each shortcut-potential subtree $\jtree(S) = (V(S), E(S)) \subseteq \jtree$, 
	space budget \budget, and target benefit \targetbenefit.
	We remark again that both functions $\mu$ and $B$ can be specified by encodings 
	that are polynomial with the size of the tree.
	The \mosp problem asks whether there exists a set of node-disjoint shortcut-potential subtrees 
	such that the associated total weight does not exceed \budget while the total benefit is at least \targetbenefit.
	
	An instance of the \kp problem can be mapped to an instance of the \mosp problem 
	in which only node-disjoint shortcut potentials have cost smaller than \budget and 
	strictly-positive benefit. 
	In more detail, 
	an instance of \mosp can be constructed from an instance of \kp as follows:
	First, we set $\budget = \knapsack$ and $\targetbenefit = \targetprofit$.
	Then we construct a tree \jtree to have root \jtroot and $n$ leaves, 
	one for each item in \kp. For all leaves, there are at least three nodes in the path between the leaf and \jtroot.  
	Within the $i$-th subtree rooted at each of the children of \jtroot, there is only one shortcut potential $\spot_i$
	with cost smaller than \budget. 
	Furthermore, the subtrees associated with all such shortcut potentials 
	consist of a single node having two incident separators. 
	
	Next we define the query log \querylog required for the \mosp instance.
	We take \querylog to consist of $n$ queries with uniform probabilities $\frac{1}{n}$ 
	so that all query variables $x \in \query$ have cardinality $\alpha(x)=1$.
	The $i$-th shortcut-potential subtree is the only useful shortcut potential for query $i$ 
	because all query variables are in the cliques adjacent to the shortcut-potential subtree. 
	For the shortcut-potential subtree 
	contained in the subtree rooted at the $i$-th child of \jtroot,
	separators are chosen so that $\mu(S_i) = w_i$.
	The benefit of the $i$-th shortcut potential is then given by 
	$\frac{1}{n} \mu(S_i) \alphares(S_i)$,
	where we can set the cardinality of the variables belonging to the clique (but not to the separators)
	so that $\alphares(S_i) = {n p_i}/{\mu(S_i)}$. 
	This implies $B(S_i, \querylog) = p_i$, for all $i\in[n]$. 
	Notice that a subtree spanning more than one of the subtrees rooted at the children of \jtroot 
	may have cost lower than \budget but it is not useful for any query and hence it has null benefit. 
	Therefore, no other shortcut-potential subtree can be part of the solution.
	The transformation we described reduces an instance of \kp  
	to an instance of \mosp, 
	having the same solutions.
	We conclude that the \mosp problem is a generalization of the \kp problem, 
	and thus, it is \NP-hard.
\end{proof}

\subsection{Materializing a single shortcut potential} 
\label{sec:single_sp_alg}

\begin{algorithm}[t]
	\caption{Left-to-right dynamic programming (\lrdp) for the single optimal shortcut potential (\sosp) problem}
	\label{alg:single_sp}
	\hspace*{-0.6cm}\textbf{Input}: junction tree \jtree, budget \budget, query log \querylog , pivot of \spot $r_S$\\
	\begin{algorithmic}[1]
		
		\For{$c \gets 0$ to \budget} 
		\State $I[r_S,c] \gets 0$, 
		\ReviewOnly{
		\revisioncol{$P[r_s,c] \gets -\infty$}
	    }
		\EndFor
		
		\For{$v \gets 1$ to $n$} 
		\State $P \gets {\forwardmove}(v, P, \budget)$ 
		
		\If{$v \in \jtleafs(\jtree)$} 
		\WhileNoDo{$(v \neq r_S)$ and}
		\StatexIndent[4]{$(\text{there is no } i \in \childs(v) \text{ s.t. } i > v)$} \algorithmicdo
		\State $I,P \gets {\backwardmove}(v, I, P, \budget)$ 
		\State $v \gets \parent_{v}$
		\EndWhile
		\EndIf
		\EndFor
		\noindent \Return $I,P$
		
		\item[]
		
		\Procedure{{\forwardmove}}{$v$, $P$, \budget}
		\For{$c \gets 0$ to $\budget$}
		\If{$c >  \nodecost{v}$}
		\State $P[v, c] \gets \nodebenefit{v}$
		\EndIf
		\EndFor
		\Return $P$
		\EndProcedure
		
		\item[]
		
		\Procedure{{\backwardmove}}{$v$, $I$, $P$, \budget}
		\For{$c \gets 0$ to $\budget$}
		\State $\currentchilds(\parent_{v}) \gets \{u \in \childs(\parent_{v}) : u \leq v\}$
		
		\State $\Gamma_{c}' = \{\Gamma:\Gamma \subseteq \currentchilds(\parent_{v}) \land  \sum_{\child \in \Gamma} \nodecost{\child} \leq c\}$ 
		\State $\bopt{\parent_v}[c] \gets \max_{\Gamma_c \in \Gamma_c'}  \sum_{\child \in \Gamma_c} P[\nodecost{\child}, \child]$
		\State $\Gamma \gets \arg \max_{\Gamma_c \in \Gamma_c'}  \sum_{\child \in \Gamma_c} P[\nodecost{\child}, \child]$
		\State  $\eta \gets \max \left\{P[v, c] , P[\parent_{v}, c] ,  \bopt{\parent_v}[c] \right\}$
		\If{$\eta = P[v, c]$}
		\For{$\child \in \currentchilds(\parent_{v})$}
		\State $I[\child, c] \gets 0$
		\EndFor
		\State $I[v, c] \gets 1$
		\EndIf
		\If{$\eta = \bopt{\parent_v}[c]$}
		\For{$\child \in \currentchilds(\parent_{v})$}
		\State $I[\child, c] \gets 0$
		\If{$\child \in \Gamma$}
		\State $I[\child, c] \gets 1$
		\EndIf
		\EndFor
		\EndIf
		\If{$\eta =  P[\parent_{v}, c]$}
		\For{$\child \in \currentchilds(\parent_{v})$}
		\State $I[\child, c] \gets 0$
		\EndFor
		\EndIf
		\State $P[\parent_{v}, c] \gets \eta$
		\EndFor
		\Return $I,P$
		\EndProcedure		
	\end{algorithmic}
\end{algorithm}

\begin{algorithm}[t]
	\caption{Reconstruct solution found by \lrdp algorithm}
	\label{alg:single_sp_reconstruct}
	\textbf{Input}: junction tree \jtree,  budget \budget, matrix $I$ and $r_S$ from Algorithm~\ref{alg:single_sp}
	\begin{algorithmic}[1]
		\For{$c \gets 0$ to $\budget$}
		\State $\tovisit \gets  \{ \childs(\jtroot_S) \}$
		\State $\currentpath \gets  \emptyset$
		\State $S[r_S,c] \gets \emptyset$
		\While{$\tovisitsize > 0$}
		\State $v \gets \tovisitpop()$ \label{op0}
		\If{$(v \in \jtleafs(\jtree)) \text{ and } (I[v,c] = 1)$}
		\State $S[r_S,c] \gets S[r_S,c] \cup (v, \parent_{v})$
		\State ${\currentpath} \gets \emptyset$
		\EndIf
		\IfNoDo{$(\currentpathsize > 0) \text{ and } (I[v,c] = 0) \text{ and}$}
		\StatexIndent[3]{$(\text{there is no }  \child \in \childs(\parent_{v}) \text{ s.t. } \child > v)$} \algorithmicdo
		\State $S[r_S,c] \gets {S}[r_S,c] \cup \{ (\parent_{v}, p_{\parent_{v}}) \}$
		\State ${\currentpath} \gets \emptyset$
		\EndIf
		\If{$I[v,c] = 1$}
		\State ${\currentpath} \gets {\currentpath} \cup \{ v \}$
		\State $\tovisit \gets \tovisit \cup \childs(v)$
		\EndIf
		\EndWhile
		\EndFor
		\noindent \Return ${S}$
	\end{algorithmic}
\end{algorithm}

Despite the hardness of the two problems, we now show that they are not strongly \NP-hard, 
and so it is possible to obtain an optimal solution in \emph{pseudo-polynomial time}. 
We start with the \sosp problem.
The algorithm we design follows the \emph{left-to-right} (or depth-first) dynamic programming approach
\citep{johnson1983knapsacks,samphaiboon2000heuristic,cho1997depth}, 
which we appropriately name \lrdp.
The procedure is illustrated in Algorithm~\ref{alg:single_sp}.
We assume that the nodes of the tree are labeled from $0$ (the pivot $r_S$) 
to $n$ in a depth-first manner and ``from left to right.''
%
%
%
Consider any node~$v$. Let $S_v$ be the shortcut potential such that $V(S_v)$ includes 
all the nodes in the path $\jtpath(\parent_v,\jtroot)$, where $\parent_v$ is the parent of $v$. 
In what follows, for simplicity of notation, 
we refer to the benefit and cost of the shortcut potential $S_v$ as $\nodebenefit{v}$ and $\nodecost{v}$, respectively.  
The \lrdp algorithm has two building blocks: 
a \emph{forward} and a \emph{backward} step. 
The forward step computes the benefit $\nodebenefit{v}$ and the cost~$\nodecost{v}$, for each node $v$,
visiting nodes in a depth-first fashion. \revisioncol{This is the only step in which the query-log is used.}
The backward step is performed every time a leaf node $u$ of the junction tree \jtree
is reached, while there are no nodes with larger label to visit in the forward step. 
Let us denote the leaf nodes of \jtree by $\jtleafs(\jtree)$. 
To perform the backward step at node $v$, we consider all cost values $c \in [\budget]$.  
For each value $c$, we find the optimal combination of the children of $\parent_{v}$ with label $v$ or lower 
(i.e., the children that have currently been visited by the algorithm) 
that maximize the benefit under the constraint that the total cost does not exceed $c$.  
Let $\Gamma$ denote the optimal combination and $\bopt{\parent_v}$ the corresponding benefit. 
To complete the backward step, we compare $\bopt{\parent_v}$ 
with $\nodebenefit{v}$ and $\nodebenefit{\parent_{v}}$.  
Based on this comparison, 
we keep track of the nodes that contribute to the optimal solution and of the optimal benefit value. 
We tabulate the results of the computation using two matrices, $P$ and $I$. 
The former stores benefit values, whereas the latter acts as an indicator and is used to reconstruct the solution.
In particular, we set $I[v,c]=1$ for each node $v$ such that the separators $(v, \parent_{v})$ 
are part of the \cutS associated with the 
optimal shortcut potential of cost~$c$.
We also set $I[u,c] = 1$ for each node $u$ in the path between $v$ and~\jtroot. 
Similarly, $P[v,c]$ stores the current optimal benefit for a shortcut potential 
when the algorithm visits node $v$. 
%
%
After the last backward step is carried out, 
the \lrdp algorithm has computed the optimal shortcut potential for all cost values $c\in[\budget]$. 
Once the optimal single shortcut-potential benefit has been computed, 
it is rather simple to reconstruct the solution. 
The reconstruction procedure is described in Algorithm~\ref{alg:single_sp_reconstruct}. 
Two auxiliary sets are used, $\currentpath$ and $\tovisit$. 
The main idea is that all nodes that identify the optimal shortcut-potential subtree 
are marked by the matrix $I$, as discussed above. 
More specifically, all nodes in a path between the root $r_S$ of $S$ and a node $i$ 
for which the separator $(i,\parent_i)$ is part of the set \cutS identifying the optimal shortcut potential 
are assigned value $1$ in the indicator matrix.
Thus, it is necessary to retrieve the bottom-most node marked in $I$ with $1$ for each such path. 
The reconstruction procedure traverses the tree top-down\ReviewOnly{\revisioncol{,}} visiting all the nodes indicated 
by the indicator matrix $I$ in a depth-first ("left-to-right") manner. 
This order is respected because the $\pop$ operation in Line~\ref{op0} 
returns the node with the smallest {DFS}~label.
The execution of Algorithm~\ref{alg:single_sp_reconstruct} 
gives the set of separators whose joint probability distribution is the optimal shortcut potential. 
Algorithms \ref{alg:single_sp} and \ref{alg:single_sp_reconstruct} 
can be executed with each internal node of the tree \tree being considered as root $r_S$, 
so as to obtain the single optimal shortcut potential $S[r_S,c]$ rooted at $r_S$ and of cost $c$ as well as its associated optimal benefit $P({S[r_S,c]})$,
for each $\jtroot_S\in V$ and each cost value $c \in [\budget]$. 
This will be used in developing the solution for the \mosp problem, 
as we will see shortly. 
Without delving into details, it is worth noting that
when executing the algorithm for all $\jtroot_S \in V$, 
it is possible to share computations 
among different roots. 

\spara{Time complexity.}
The running time of \lrdp is pseudo-polynomial; 
more precisely $\mathcal{O}(n \budget^2)$. 
The quadratic term is due to the computation of $\bopt{\parent_v}$ in the backward step, 
which is performed once for each clique node $v \neq \jtroot$. 
\ReviewOnly{\revisioncol{ Note that a more fine-grained expression of the time complexity for \lrdp would include the complexity for computing the benefit values $\nodebenefit{v}$, but in practice this additional complexity is dominated by the backward step.  }}




\subsection{Materializing multiple shortcut potentials} 
\label{sec:k_sp_alg} 
\begin{algorithm}[t]
	\caption{Bottom-up dynamic programming (\budp) for the multiple optimal shortcut potentials (\mosp) problem}
	\label{alg:k_sp}
	\textbf{Input}: junction tree $T$, budget \budget, $P$ from Algorithm~\ref{alg:single_sp}, 
	\linebreak \indent \hspace{4.3cm} $S$ from  Algorithm~\ref{alg:single_sp_reconstruct}  \\
	\begin{algorithmic}[1]
		\ForAll {$v \in V$}
		\For{$c \gets 0$ to $\budget$} 
		\State $I[v,c]  \gets 0$
		\EndFor
		\EndFor	
		\ForAll {$v \in V$}
		\If{$v \not \in \jtleafs(\jtree)$} 
		\For{$c \gets 0$ to $\budget$}  \Comment{ case (i) }
		\State $\Gamma_{c} \gets \{ \nodecost{\child} \text{ for all } \child \in \childs(v) \text{ s.t. } $ 
		\State \indent \quad \quad  $\sum_{\child \in \childs(v) }  \nodecost{\child} \leq c \} $
		\State $\phi_1 \gets \sum_{\child \in \childs(v),  \nodecost{\child} \in  \Gamma_{c}} P(S[\child, \nodecost{\child}])$
		\State $H_1[v,c] \gets \max_{\,\Gamma_{c}} \left\{  \phi_1 \right\}$
		
		\ForAll{$\child \in \childs(v)$}
		\State $W_1[\child,c] \gets \arg\max_{\nodecost{\child}} \left\{ \phi_1 \right\}$
		\EndFor
		\EndFor
		\For{$c \gets 0$ to $\budget$} \Comment{ case (ii) }
		\State $\Gamma_{c',c} \gets \{ \nodecost{d} \text{ for all } d \in D(S[v, c']) \text{ s.t.}$
		\State  \indent \quad \quad  $\sum_{d \in  D(S[v, c'])} \nodecost{d} \leq c-c'  \}  $
		\State $\phi_2 \gets P(S[v, c']) + $
		\State \indent \quad \quad $\sum_{d \in D(S[v, c']), \nodecost{d} \in  \Gamma_{c',c}} P(S[d, \nodecost{d}])$
		
		\State $H_2[v,c] \gets \max_{c', \Gamma_{c',c}} \left\{ \phi_2 \right\} $
		
		\State $W_2[v,v, c] \gets \arg \max_{c'} \left\{ \phi_2 \right\} $
		\ForAll{$d \in D(S[v, c'])$}
		\State $W_2[v, d ,c] \gets \arg \max_{\nodecost{d} \in \Gamma_{c',c}} \left\{ \phi_2 \right\}$
		\EndFor
		\EndFor
		\For{$c \gets 0$ to $\budget$} 
		\If{$H_2[v,c] > H_1[v,c]$}
		\State $I[v,c]  \gets 1$
		\State $H[v,c] \gets H_2[v,c]$
		\Else 
		$\ H[v,c] \gets H_1[v,c]$
		\EndIf
		\EndFor
		\EndIf
		\EndFor
		\noindent \Return $I, H$
	\end{algorithmic}
\end{algorithm}


\begin{algorithm}[t]
	\caption{Reconstruct solution found by \budp algorithm}
	\label{alg:k_sp_reconstruct}
	\begin{algorithmic}[1] 	
		\Procedure{Reconstruct}{$v, c$}
		\If{$I[v,c] = 1$}
		\State $c' \gets W_2[v, v,c]$
		\State $\spotset \gets \spotset \cup S[v,c']$
		\ForAll{$d \in D(S[v,c'])$}
		\State $\nodecost{d} \gets W_2[v, d, c]$
		\State\Return{\Call{Reconstruct}{$u,\nodecost{d}$}}
		\EndFor
		\Else
		\ForAll{$\child \in \childs(v)$}
		\State $\nodecost{\child} \gets W_1[\child,c]$
		\State\Return{\Call{Reconstruct}{$\child,\nodecost{\child}$}}
		\EndFor
		\EndIf
		\EndProcedure
	\end{algorithmic}
\end{algorithm}


We now show how to obtain an optimal packing for \mosp,  
i.e., \cardbudget node-disjoint shortcut potentials. 
Recall that \cardbudget is not specified in the input, 
but rather it is optimized by the algorithm. 
In other words, 
we find the optimal set of non-overlapping shortcut potentials 
leading to the largest total benefit while satisfying the budget constraint. 
\revisioncol{We stress that the shortcut potentials we retrieve are not optimal in the general sense, 
but are optimal under the constraints introduced in our problem formulation.}
As a pre-processing step, we first execute 
Algorithms~\ref{alg:single_sp} and~\ref{alg:single_sp_reconstruct}
with each node of the tree being considered a root.
As a result we compute the single optimal shortcut potentials
for each node $\jtroot_S \in V$
and each cost value $c \in [\budget]$. 
We proceed with a bottom-up dynamic-programming algorithm, 
named \budp and shown as Algorithm~\ref{alg:k_sp}.
In the  \budp algorithm, 
each node $v$ is visited when all its children $\childs(v)$ have already been visited. 
When visiting $v$, 
if it is not a leaf, 
we compute the value of the benefit of the optimal packing in the subtree rooted at~$v$, 
while considering two cases: 
\begin{enumerate*}[label=(\roman*)]
\item the solution includes a shortcut potential subtree rooted at~$v$, and
\item the solution does not include a shortcut potential subtree rooted at~$v$. 
\end{enumerate*}
In case (i)\ReviewOnly{\revisioncol{,}}
the total benefit $H_1[v,c]$ at $v$ for a given cost $c$ is simply obtained 
by combining the optimal solutions at each node $\child \in \childs(v)$. 
In case (ii)\ReviewOnly{\revisioncol{,}} the total benefit $H_2[v,c]$ at $v$ for a given cost $c$  
is the benefit of the best combination of the optimal shortcut potential $S$ rooted at $v$ 
and the optimal solutions at the nodes $s$ belonging to the set $D(S)$ of descendants of \spot 
defined as 
$D(S) = \{s : d \notin \spotvertices{S[v,c]} \text{ and } p_u \in \spotvertices{S[c,v]} \}$.
For any internal node $v$, we also use an indicator matrix $I$ to store 
which of the two cases provides the largest benefit, 
i.e., $I[v,c] = 1$ if $H_2[v,c] > H_1[v,c]$, and $I[v,c] = 0$ otherwise.
Additionally, in order to reconstruct the optimal solution, 
for the optimal packing of cost $c$ in the subtree rooted at a node $v$, 
it is necessary to store the cost allocated to each shortcut potential 
that is part of the optimal packing in that subtree.
For this purpose, we use two matrices $W_1$ and $W_2$.  
The matrix $W_1$ is $2$-dimensional, while
$W_2$ is $3$-dimensional 
because a node can be a child of only one parent 
but a descendant of multiple shortcut potentials. 
%
Having computed the benefit associated with the optimal packing of \cardbudget shortcut potentials, 
it is necessary to reconstruct the solution.
It is convenient to address this task recursively, 
as shown in Algorithm~\ref{alg:k_sp_reconstruct}, 
where the nodes of tree \jtree are traversed top-down. 
The matrices $I$, $W_1$, and $W_2$ returned by \budp 
are used in Algorithm~\ref{alg:k_sp_reconstruct}.
Every time a node $v$ is visited, 
we use the indicator matrix~$I$ 
to distinguish the two cases discussed above. 
If case (i) is optimal for the subtree rooted at $v$,  
we continue the recursion at nodes $u \in \childs(v)$ with weights indicated by matrix $W_1$. 
Otherwise, if case (ii) is optimal, 
we continue the recursion at nodes $d \in D(S)$ with the weights given by matrix $W_2$. 
For the initial invocation of \textsc{Reconstruct} 
we pass as arguments the root of \tree and the budget \budget. 
The method returns the optimal packing \spotset of node-disjoint shortcut potentials, 
i.e., the solution to the \mosp problem.

\spara{Time complexity.}
The running time of the \budp algorithm is $\mathcal{O}(n \budget^3)$. 
Each of the internal nodes of the junction tree is visited once\ReviewOnly{\revisioncol{,}} and 
the total benefits corresponding to cases (i) and (ii) are evaluated. 
The cubic term 
stems from the evaluation of the total benefit under case (ii).
Accounting for the preprocessing, 
where the \lrdp algorithm is called $\mathcal{O}(n)$ times, 
the total complexity of solving the \mosp problem is 
$\mathcal{O}(n^2\budget^2 + n \budget^3)$.

\subsection{A strongly-polynomial algorithm} 

\label{sec:approximation_alg}
In addition to the pseudo-polynomial algorithms introduced in the previous sections, 
we now present a strongly-polynomial variant for both problems \sosp and \mosp. 
This strongly-polynomial algorithm is not optimal, clearly,
but it performs very well in practice as we will see in our experimental evaluation.
The main idea is to reduce the amount of admissible values for the budget. 
Instead of considering all budget values $c\in\{0,\ldots,\budget\}$, 
we can work with an exponentially smaller set.
In particular, for some real value $\epsilon \geq 1$, 
the set $\{0,\ldots,\budget\}$ is partitioned into a smaller number of bins 
of increasing size $\lfloor{\epsilon^i}\rfloor$, where $i$ is the bin number.  
We thus form a reduced set 
$\{0, \lfloor{\epsilon}\rfloor,\lfloor{\epsilon^2}\rfloor,\ldots,\budget \}$,  
which is used by agorithms \lrdp and \budp in lieu of $\{0,\ldots,\budget\}$.
The parameter $\epsilon$ embodies the trade-off 
between solution quality and running time. 
Increasing the value of $\epsilon$ makes the dynamic-programming framework more efficient, but
it decreases the quality of the solution. 
As a rule of thumb, a value of $\epsilon$ close to $1.2$ offers a solution 
that is often found to be close to the optimal
while ensuring a substantial reduction of running time. 
A \emph{fully-polynomial time approximation scheme} ({\sf FPTAS}) 
can also be designed by rounding and scaling the benefit values. 
However, such an approach does not scale well in practice, and therefore in our experimental evaluation 
we focus solely on the strongly-polynomial algorithm discussed above, which, 
while being a heuristic, is empirically found to be remarkably effective, 
particularly when the budget available for materialization is large.

\subsection{Putting everything together}
\label{section:peanut}
The techniques discussed in this section are incorporated into a comprehensive method, called \ouralgorithm. 
Its structure is simple: \ouralgorithm is composed of an offline and an online component. 

\spara{In the offline component}, algorithm \lrdp is executed once for each internal \revisioncol{node} as a root, 
so as to obtain the optimal shortcut potentials corresponding to each combination of root and cost.
The resulting computation is then used by algorithm \budp to compute an optimal packing of 
multiple node-disjoint shortcut potentials, which is
used by the online component of the system to answer inference queries. 
\ouralgorithm can use the strongly-polynomial algorithm 
to reduce the time required to identify the materialization. 

\spara{In the online component}, inference queries are processed. 
First, given a query \query, the associated Steiner tree \steinertree is extracted. 
\ouralgorithm checks whether there are materialized shortcut potentials 
that are useful for the query \query and, 
if there are, uses them to reduce the size of \steinertree and 
consequently the message-passing cost to answer~\query.

\subsection{Relaxing the node-disjointness constraint}
\label{section:overlap}

The algorithms described in the previous sections find the optimal set of shortcut potentials under the constraints that the shortcut potential subtrees are node-disjoint. While such a constraint is crucial to make the problem tractable, it limits the  budget that can be used for materialization.  
In particular, in our experiments we observe that \ouralgorithm tends not to utilize all the available budget, 
resulting in a smaller benefit than what could be possible if the whole budget is used.
To overcome this limitation, 
we propose \ouralgorithmplus, a method that combines the ideas discussed above with a simple \emph{greedy packing heuristic}. 
In \ouralgorithmplus, we only run Algorithms~\ref{alg:single_sp} and~\ref{alg:single_sp_reconstruct} at each possible root and then, among all the retrieved shortcut potentials, we greedily select the ones with the largest ratio $r$ of benefit $\benefit{\cdot}$ to cost $\mu(\cdot)$ until the budget is filled. 

The online component of \ouralgorithmplus is more complex than the one of \ouralgorithm because overlapping shortcut potentials cannot be simultaneously used. Hence, for each query \query, we assemble a conflict graph $G_c$ in which the shortcut potentials that are useful for \query are weighted by the corresponding value of $r$ and are joined by an edge if they are overlapping. From $G_c$, we extract a maximum weighted independent set by using the greedy {\sf{\small GWMIN}} algorithm~\cite{sakai2003note}.
In practice, this extra step incurs an overhead, 
which can be expected to be negligible with respect to the execution time of the queries. 


\section{Experimental evaluation}
\label{sec:experiments}

\begin{figure*}[t]
	\begin{tabular}{ccc}
		\hspace{-2mm} \includegraphics[width=0.52\columnwidth, height = 0.33\textheight, keepaspectratio]{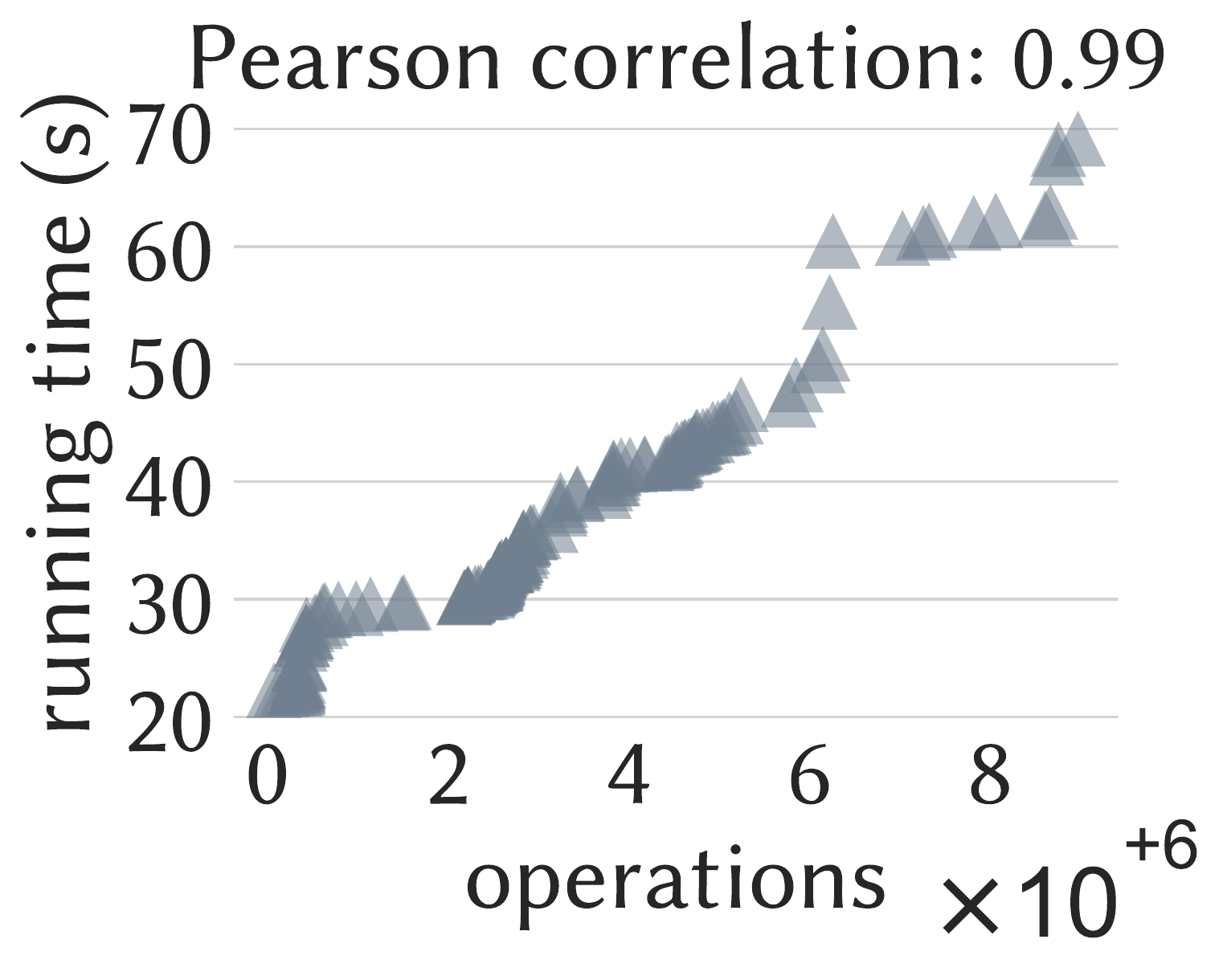}&
		\hspace{-2mm} \includegraphics[width=0.52\columnwidth, height = 0.33\textheight, keepaspectratio]{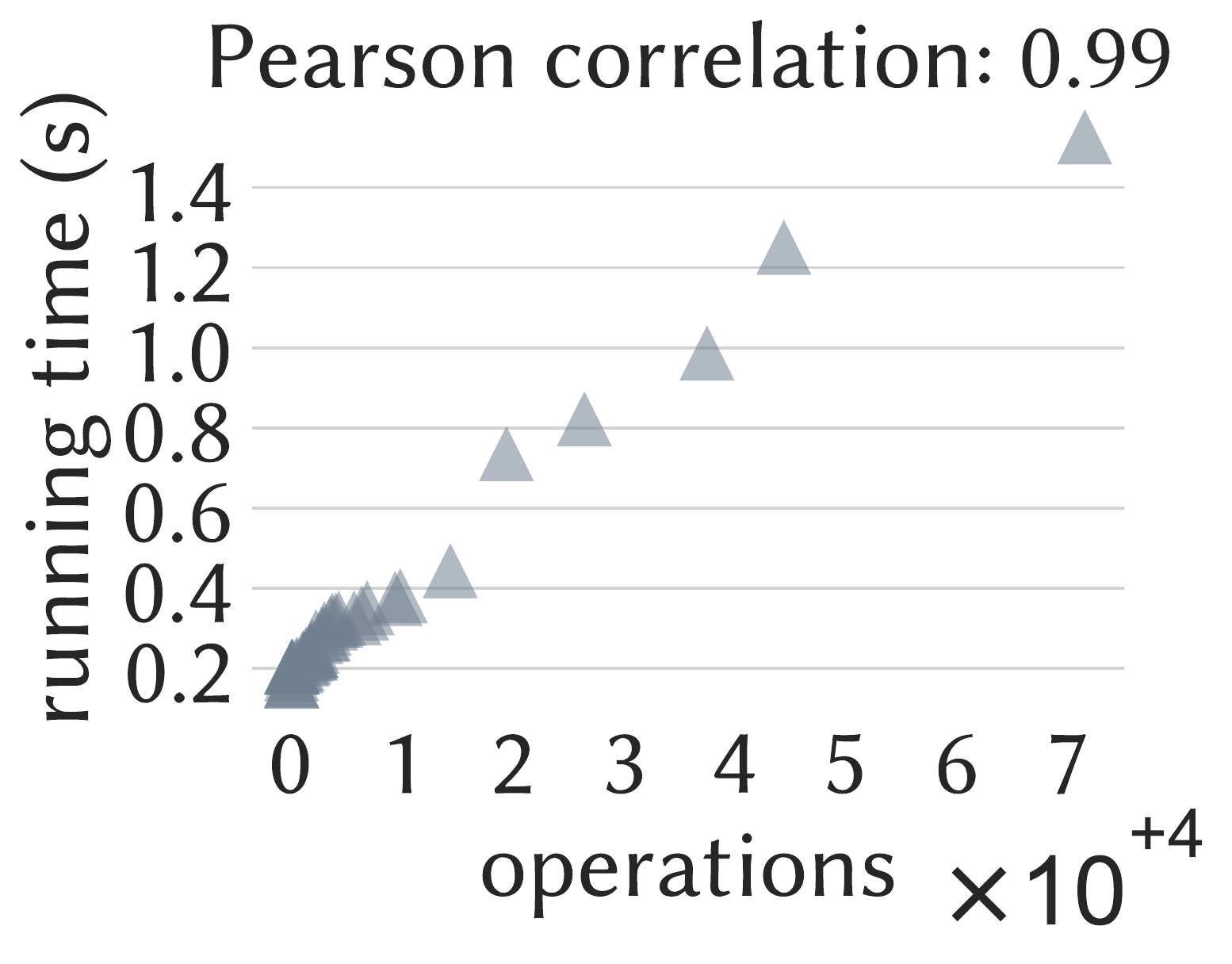}&
		\hspace{-2mm} \includegraphics[width=0.52\columnwidth, height = 0.33\textheight, keepaspectratio]{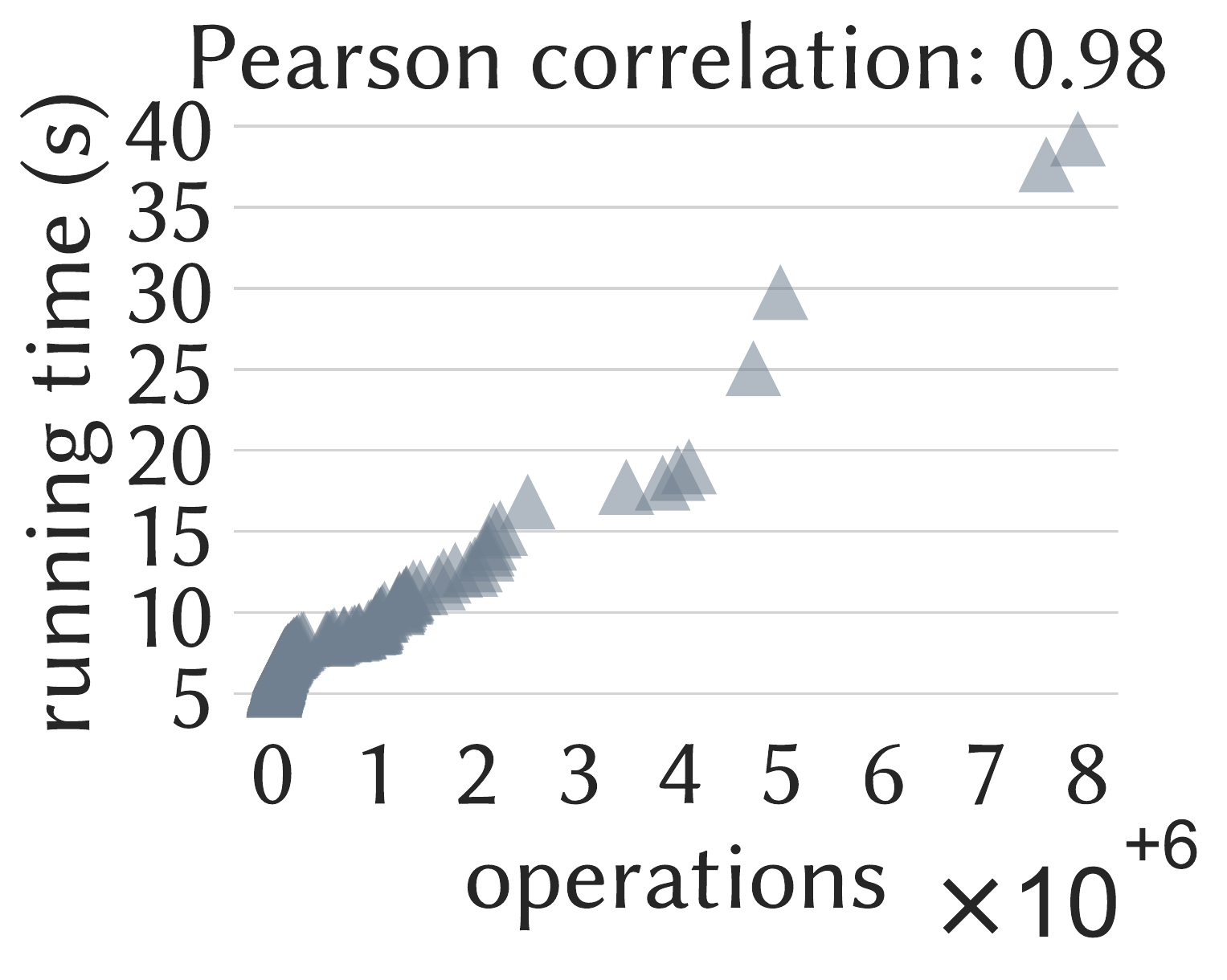}\\
		(a) \textsc{Andes} & (b) \textsc{HailFinder}   & (c) \textsc{PathFinder}  
	\end{tabular}
	\caption{\label{fig:correlation_plot} \revisioncol{ Running time against cost values for a set of queries processed with the standard junction-tree algorithm.}} 
	\centering
	\includegraphics[width=1\columnwidth, height = 0.6\textheight, keepaspectratio]{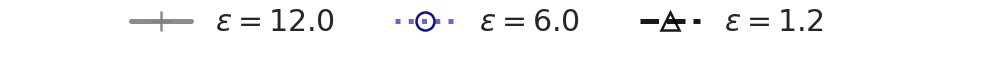}\\
	\begin{tabular}{ccc}
		\hspace{-4mm} \includegraphics[width=0.52\columnwidth, height = 0.33\textheight, keepaspectratio]{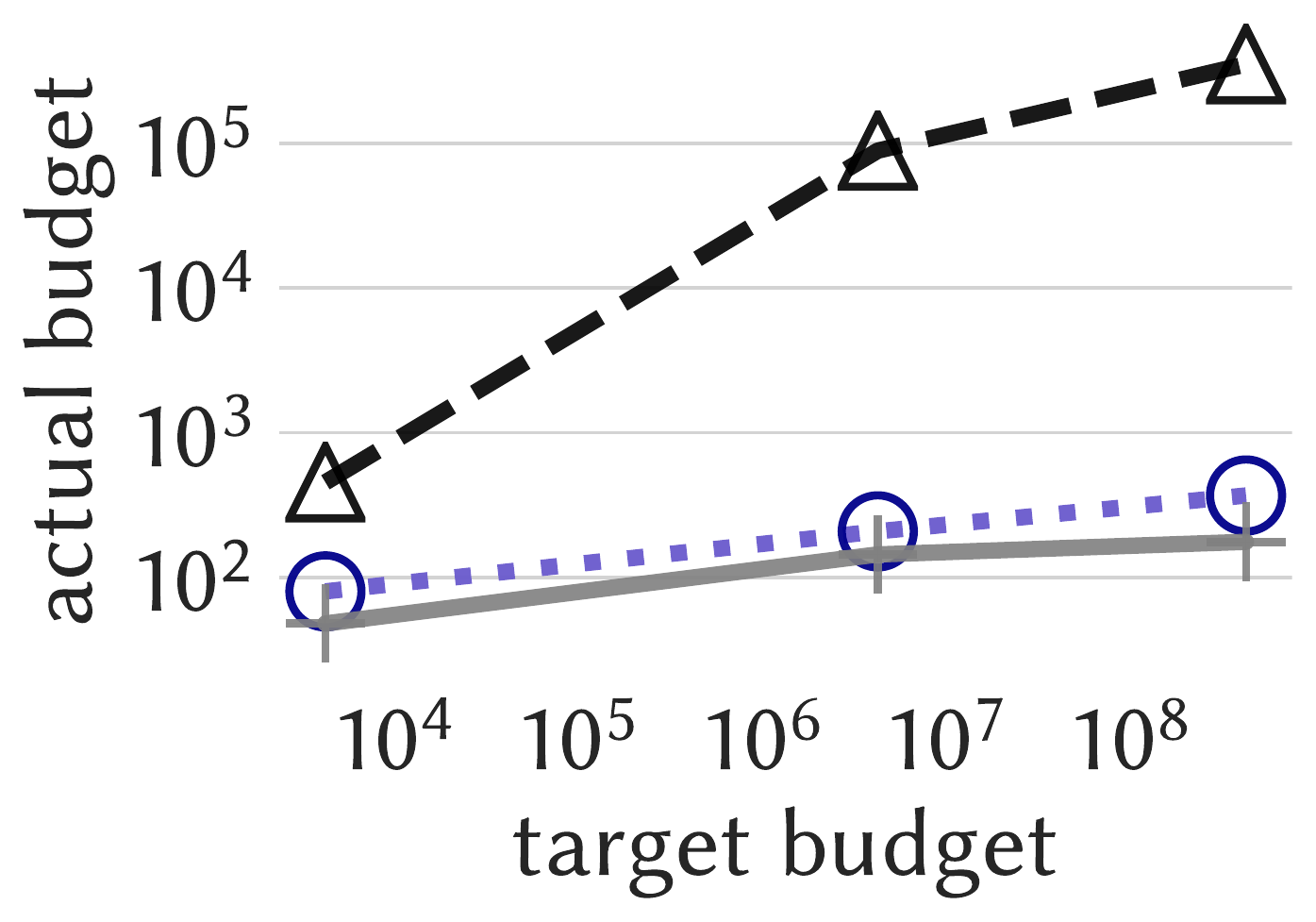}&
		\hspace{-4mm} \includegraphics[width=0.52\columnwidth, height = 0.33\textheight, keepaspectratio]{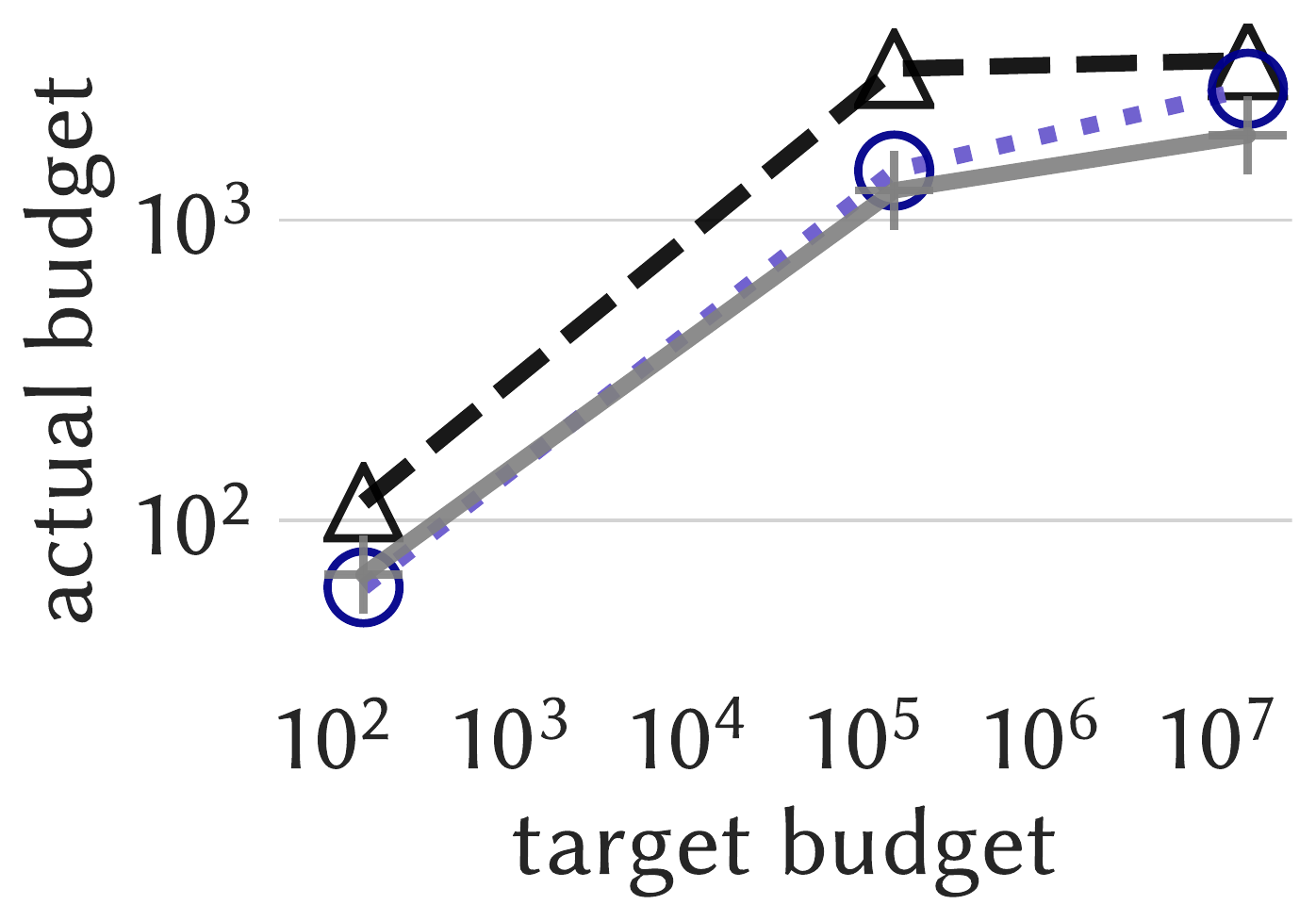}&
		\hspace{-4mm} \includegraphics[width=0.52\columnwidth, height = 0.33\textheight, keepaspectratio]{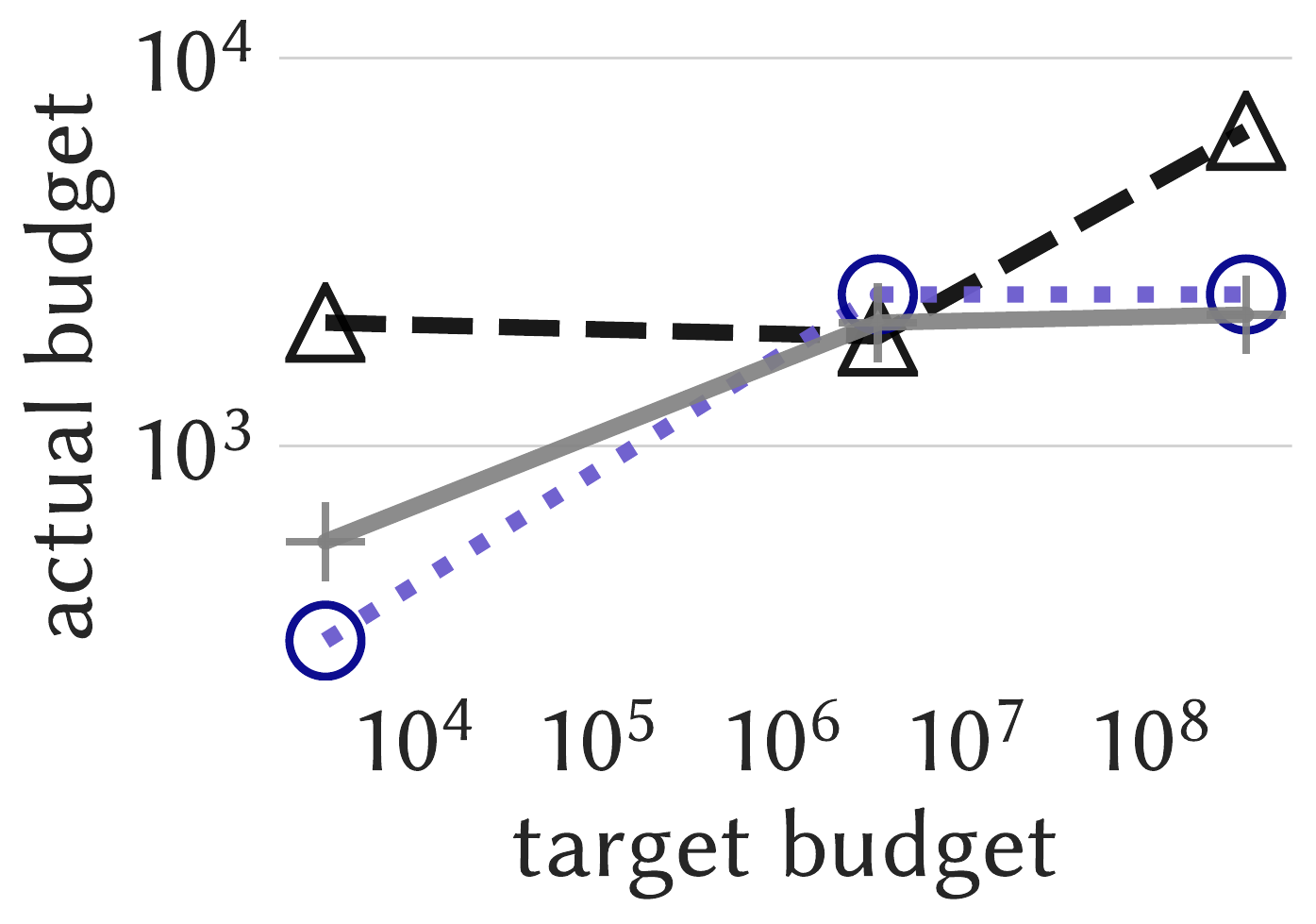} \\
		(a) \textsc{Andes} & (b) \textsc{hailFinder}   & (c) \textsc{pathFinder}  
	\end{tabular}
	\caption{\label{fig:materialized_vs_target_budget} Materialized budget against target budget in \ouralgorithm with different levels of approximation (log-log scale). }
\end{figure*}

In this section\ReviewOnly{\revisioncol{,}} we present an extensive experimental evaluation of our methods, 
using benchmark Bayesian network datasets, 
and comparing \ouralgorithm and \ouralgorithmplus
with the strong (but not workload-aware) baseline of \citet{kanagal2009indexing}, and with the recent techniques for optimal workload-aware materialization for inference based on the variable-elimination algorithm~\cite{aslay2021workload}. \revisioncol{We also study the performance of our methods in the presence of temporally-shifted workloads.}

\begin{table}[t]
\begin{small}
	\centering
	\caption{Summary statistics of Bayesian networks.}
	\label{tab:bn_statics}
	\vspace{-1mm}
	\begin{tabular}{lrrrc}
		\toprule
		\multicolumn{1}{l}{dataset} & \multicolumn{1}{r}{\# nodes} & \multicolumn{1}{r}{\# edges} & \multicolumn{1}{r}{\# parameters} & \multicolumn{1}{c}{max in-degree} \\
		\midrule
		\textsc{Child} & $20$ & $25$ & $230$ & $2$  \\
		\textsc{Hepar II} & $70$ & $123$ & $1.4$\,K & $6$  \\ 
		\textsc{Andes} & $223$ & $338$ & $1.1$\,K & $6$  \\
		\textsc{Hailfinder} & $56$ & $66$ & $2.6$\,K & $4$  \\ 
		\textsc{TPC-H} & $38$ & $39$ & $355.5$\,K & $2$  \\   
		\textsc{Munin} & $186$ & $273$ & $15.6$\,K & $3$ \\ 
		\textsc{PathFinder} & $109$ & $195$ & $72.1$\,K & $5$  \\ 
		\textsc{Barley} & $48$ & $84$ & $114$\,K & $4$  \\ 
		\bottomrule
	\end{tabular}
\end{small}
\end{table}

\begin{table}[t]
\begin{small}
	\centering
	\caption{Summary statistics of junction trees.}
	\label{tab:jt_statics}
	\vspace{-1mm}
	\begin{tabular}{lrrr}
		\toprule
		\multicolumn{1}{l}{dataset} & \multicolumn{1}{c}{\# cliques} & \multicolumn{1}{c}{diameter} & \multicolumn{1}{c}{treewidth}  \\
		\midrule
		\textsc{Child} & $17$ & $10$ & $3$  \\
		\textsc{Hepar II} & $58$ & $14$ & $6$  \\
		\textsc{Andes} & $175$ & $25$ & $17$  \\
		\textsc{Hailfinder} & $43$ & $14$ & $4$  \\
		\textsc{TPC-H} & $33$ & $16$ & $2$  \\
		\textsc{Munin} & $158$ & $23$ & $11$  \\
		\textsc{PathFinder} & $91$ & $17$ & $6$  \\
		\textsc{Barley} & $36$ & $14$ & $7$  \\
		\bottomrule
	\end{tabular}
\end{small}
\end{table}

\subsection{Experimental setting}

We describe in more detail the benchmark Bayesian network datasets we use, 
the baseline method we compare, 
and the other settings and design choices for our experimental evaluation. 



\spara{Datasets.} 
We perform experiments on $8$ real-world Bayesian network datasets.
A summary of the dataset statistics is shown in Table~\ref{tab:bn_statics}.
The number of parameters in the table refers to the 
number of probablility values required to define all probabilistic dependencies in the Bayesian network.
All datasets are available online.\footnote{See \url{https://github.com/martinociaperoni/PEANUT} for \textsc{TPC-H} and \url{www.bnlearn.com/bnrepository/} for the rest} 
\textsc{Child}~\cite{spiegelhalter1993bayesian} is a model for congenital heart-disease diagnosis 
in new born ``blue babies.'' 
\textsc{Hepar II}~\cite{onisko2003probabilistic} is a model for liver-disorder diagnosis. 
\textsc{Andes}~\cite{conati1997line} is used in an intelligent tutoring system for teaching Newtonian physics to
students. 
\textsc{Hailfinder}~\cite{abramson1996hailfinder} combines meteorological data and human expertise to forecast severe weather in Northeastern Colorado.
\textsc{TPC-H} is a Bayesian network learned from TPC-H benchmark data, following Tzoumas et al.~\cite{tzoumas2013adapting}.
\textsc{Munin}~\cite{andreassen1989munin} is a subnetwork of a Bayesian network model proposed to be exploited in electromyography. 
\textsc{PathFinder}~\cite{heckerman1990toward} assists with the diagnosis of lymph-node diseases.
\textsc{Barley}~\cite{barley} supports decision making in growing malting barley without use of pesticides. 
Table~\ref{tab:jt_statics} shows salient characteristics of the junction trees associated with the datasets. 

\spara{Baselines.} 
We compare our approach against  \indsep~\citep{kanagal2009indexing}, which relies on a \emph{tree-partitioning technique}~\cite{kundu1977linear} 
to build a hierarchical index over the junction tree. 
The index construction hinges on the disk block size, 
which provides an upper bound on the memory required by each index node. 
An index node $I$ corresponds to a connected subtree of the junction tree 
and is associated with a shortcut potential given by the joint-probability distribution 
of the variables in the separators adjacent to $I$. 
Queries are processed by a recursive algorithm\revisioncol{,} and shortcut potentials 
prune the recursion tree for some queries. 
A multi-level approximation scheme for shortcut potentials 
is also implemented to deal with the scenario 
in which the size of a shortcut potential exceeds the block~size.  
We furthermore compare \ouralgorithm and \ouralgorithmplus with the optimal materialization for variable elimination that we have developed in our previous work~\cite{aslay2021workload}. We refer to this method as \qtm{n}, where $n$ refers to the number of materialized factors. 
\ReviewOnly{\revisioncol{The time and space requirements for the construction and calibration of the junction tree (JT) and for the materialization step with the different methods under comparison}} \FullOnly{the offline materialization step in this comparison} \revisioncol{are shown in Table~\ref{table:offlineStats}.}
Note that for the \textsc{TPC-H}, \textsc{Munin} and \textsc{Barley} datasets, calibration of the tree did not terminate after a one-day-long execution (asterisk and NA entries in Table~\ref{table:offlineStats}). 
The calibration of the junction tree in principle represents a preliminary step for \indsep, \ouralgorithm and \ouralgorithmplus which take as input the calibrated junction tree. 
However, the calibration process is only necessary to obtain correct answers to inference queries, and it is not needed for comparing exact query processing costs using different materialization strategies which is the goal of our experimental evaluation. Therefore, such datasets are also considered in the experiments. In this case,   \indsep, \ouralgorithm and \ouralgorithmplus take as input the uncalibrated junction tree and, as a consequence, query answers contain erroneous probability values. Nonetheless, the computational burden associated with query processing is not affected. 

\spara{Query workloads.} 
For the purpose of our study, 
as we do not have access to historical query workloads for the benchmark datasets we use
so as to estimate query probabilities, 
we propose to sample queries using a random mechanism.

In the \emph{skewed workload}, variables are sampled from a distribution skewed towards the leaves of the tree, i.e., variables have probabilities of being queried for proportional to their distance from the pivot of the junction tree. 
We generate $N_q = 3\,000$  skewed queries, 
$2\,000$ of which are used to estimate benefits, 
and the remaining $1\,000$ are used for assessment of the materialization strategies.

In the \emph{uniform workload}, variables are sampled uniformly at random. This ensures fairness in the comparison between materialization of junction trees and \qtm{n}. 
In this case, we use \ReviewOnly{\revisioncol{the same}} $N_q = 250$ queries,  all of which contain $1$ to $5$ query variables, for both optimization and evaluation. 
\revisioncol{As shown in Section~\ref{sec:query_size}, the impact of parameter $N_q$ on materialization performance is minor.}
%
\FullOnly{We also evaluate the robustness of \ouralgorithm with respect to drifts in the query distribution, 
but due to space limitations these results are given in supplementary.\footnote{\url{\supplink}}}

\begin{table}[t]
\begin{small}
	\centering
	\caption{\revisioncol{ Offline running times (seconds) for \ouralgorithm and  \ouralgorithmplus (in parenthesis) with different approximation levels and \indsep. }}
	\label{tab:offline_running_times}
	\begin{tabular}{lrrrr}
		\toprule
		\multicolumn{1}{l}{dataset} & \multicolumn{1}{r}{$\epsilon = 1.2 $} & 
		\multicolumn{1}{r}{$\epsilon = 6 $} & \multicolumn{1}{r}{$\epsilon = 12 $}  & \multicolumn{1}{r}{\indsep}  \\
		\midrule
		\textsc{Child} & $0.028\,\revisioncol{(0.020)}$ & $0.024\,\revisioncol{(0.022)}$ & $0.016\,\revisioncol{(0.015)}$ & $0.004$ \\
		\textsc{Hepar II} & $2.88\,\revisioncol{(2.81)}$ & $1.79\,\revisioncol{(1.75)}$ & $1.30\,\revisioncol{(1.27)}$ & $0.086$ \\
		\textsc{Andes} & $4.7$\,K \revisioncol{$(4.6$}\revisioncol{\,K}\revisioncol{$)$} & $252\,\revisioncol{(238)}$ & $196\,\revisioncol{(162)}$ & $2.19$ \\ 
		\textsc{Hailfinder} & $6.71\,\revisioncol{(5.89)}$ & $0.46\,\revisioncol{(0.38)}$ & $0.32\,\revisioncol{(0.29)}$ & $0.039$ \\
		\textsc{TPC-H} & $16.93\,\revisioncol{(12.23)}$ & $0.44\,\revisioncol{(0.37)}$ & $0.24\,\revisioncol{(0.21)}$ & $0.025$ \\
		\textsc{Munin} & $8.8$\,K $\revisioncol{(8.4}$\revisioncol{\,K}$\revisioncol{)}$ & $140\,\revisioncol{(138)}$ & $99.41\,\revisioncol{(99.03)}$ & $2.53$ \\ 
		\textsc{PathFinder} & $80.42\,\revisioncol{(73.22)}$ & $2.76\,\revisioncol{(2.71)}$ & $2.18\,\revisioncol{(2.15)}$ & $0.25$  \\
		\textsc{Barley} & $86.30\,\revisioncol{(78.44)}$ & $1.67\,\revisioncol{(1.63)}$ & $1.027\,\revisioncol{(1.018)}$ & $0.030$  \\
		\bottomrule
	\end{tabular}
\end{small}
\end{table}

\spara{Cost values.} 
\ReviewOnly{\revisioncol{The standard junction-tree method is prohibitively time-consuming for several datasets; note ``NA'' entries in Table~\ref{table:offlineStats}. To circumvent this limitation, we measure the performance of a materialization using the total number of operations required to process a query workload. }}
In more detail, given a query \query, a cost value $c(v)$ is assigned 
to all nodes of the Steiner tree \steinertree associated with \query, 
which could be either cliques or shortcut potentials replacing part of the original Steiner tree \steinertree. 
The cost $c(v)$ is the number of operations required to compute the message sent to $\parent_{v}$, 
or to compute the final answer to the query if $v$ coincides with the pivot of \steinertree. 
The cost of \query is then the sum of the cost values over the nodes in \steinertree. 
We confirm empirically that the assigned cost values 
align almost perfectly with the corresponding execution times, 
as shown in Figure~\ref{fig:correlation_plot},  
where we show the running time against the cost for a set of queries, 
for three datasets. 
In Figure~\ref{fig:correlation_plot}\revisioncol{,} we additionally display the Pearson correlation coefficient.  
The extremely strong correlation clearly provides empirical support for using our cost values 
to assess the performance of a materialization. 

\spara{Experiment parameters.} 
Next we discuss the choice of parameters
for space budget \budget, and approximation~$\epsilon$.

\vspace{1mm}
\noindent
\emph{Budget.} 
First we note that the two materialization methods we evaluate, 
\ouralgorithm and \indsep, 
do not always make use of the full available budget, 
as is typically the case with such \emph{discrete-knapsack} constraints\ReviewOnly{\revisioncol{,}}\FullOnly{.}
\ReviewOnly{\revisioncol{both because the node-disjointess constraint may limit the space that can be materialized and because of the strongly-polynomial approximation, if used. }}
We thus make a distinction between the \emph{target budget}~\budget, 
i.e., the value of the budget constraint, 
and the \emph{actual budget}, 
i.e., the materialized space used by a method.
Indicatively, differences between target and actual budget are demonstrated 
in Figure~\ref{fig:materialized_vs_target_budget} for \ouralgorithm, 
for different values of $\epsilon$ and a subset of datasets.
It is evident that large values of $\epsilon$ only take up a small part of the target budget, and, 
in general, 
the smaller the value of $\epsilon$, the closer the materialization space is to the target one.  
Moreover, the target budget \budget is expressed differently for \ouralgorithm and \indsep. 
In the case of \ouralgorithm, 
the target budget is specified directly as part of the input.
The larger the target budget, the larger the materialized budget is expected to be. 
Instead, in \indsep,  
the materialized budget is determined by the block size 
and cannot be directly controlled. 
Unlike \ouralgorithm and \indsep, \ouralgorithmplus can control the actual budget, which is greedily filled. This allows to compare 
\ouralgorithmplus and \indsep at (approximately) parity budget. 
In skewed workload experiments, for \indsep we consider a large set of possible block-sizes 
$\{10, 20, 50, 100, 150,\allowbreak 500,\allowbreak 1000,\allowbreak 5 \times 10^3,\allowbreak 5 \times 10^4,\allowbreak 5 \times 10^5,\allowbreak 5\times 10^6\}$\revisioncol{,} 
and we choose the three values leading to the minimum, maximum, and median materialization space. 
For \ouralgorithm we consider three different target budgets, 
namely, $\{ b_{T} /10,\allowbreak  10\,b_{T},\allowbreak  10000\,b_{T} \}$, 
where $b_{T}$ is the total potential size of the separator in \tree. 
In uniform workload experiments, for \indsep we consider block-size  $10^3$,  and similarly for \ouralgorithm and \ouralgorithmplus we consider target budget $1000\,b_{T}$.





\vspace{1mm}
\noindent
\emph{Approximation.} 
For \ouralgorithm and \ouralgorithmplus, 
the user-specified variable $\epsilon$ that 
controls the trade-off between solution quality and running time
is an important parameter.
Setting $\epsilon = 1$ gives the optimal solution, 
but it will be time-consuming.  
We show results for $\epsilon \in \{1.2, 6, 12\}$.  
Table~\ref{tab:offline_running_times} shows, for the skewed workload, \ouralgorithm \ReviewOnly{\revisioncol{~and \ouralgorithmplus}} wall-clock time
for finding the optimal materialization and \indsep wall-clock time
to build the \indsep data structure. 
Here, the budget is fixed to $\frac{1}{10}b_{T}$ for \ouralgorithm, and similarly, 
to the smallest budget used for \indsep.
Note that \ouralgorithm\  \revisioncol{  and \ouralgorithmplus } are usually much slower than \indsep 
in choosing the shortcut potentials to materialize.   
However, the offline overhead of our approach is 
counterbalanced by its superior performance during the online query-processing phase, 
which is the focus of our work.

\begin{figure*}[t]
	\centering
	 \includegraphics[width=1\columnwidth, height = 0.8\textheight, keepaspectratio]{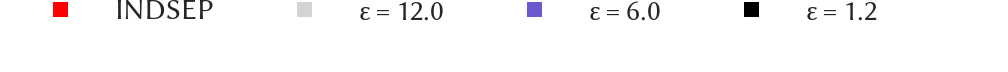}\\
		\begin{tabular}{cccc}
		\hspace{-5mm} \includegraphics[width=0.53\columnwidth, height = 0.63\textheight, keepaspectratio]{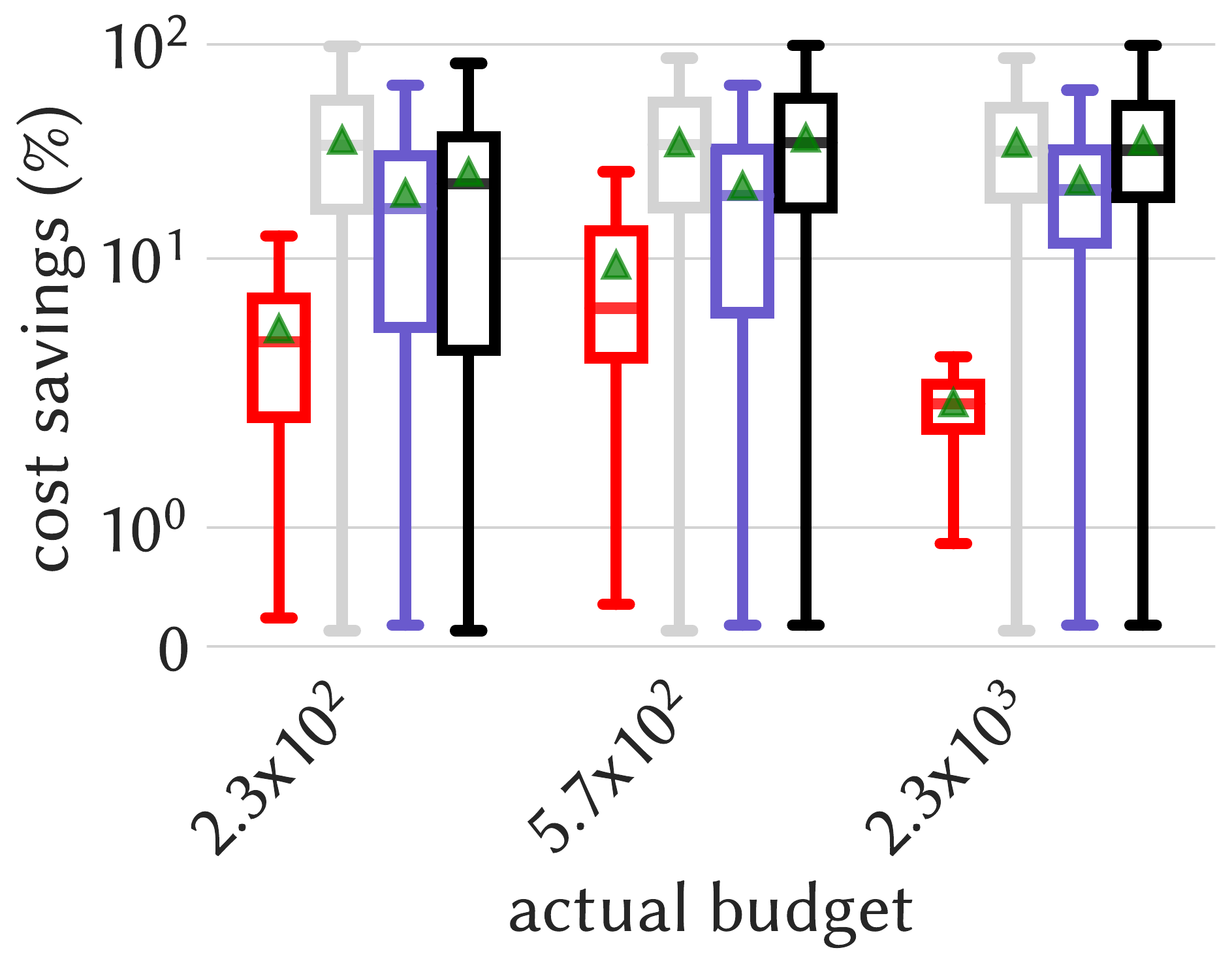}&
		\hspace{-5mm} \includegraphics[width=0.53\columnwidth, height = 0.63\textheight, keepaspectratio]{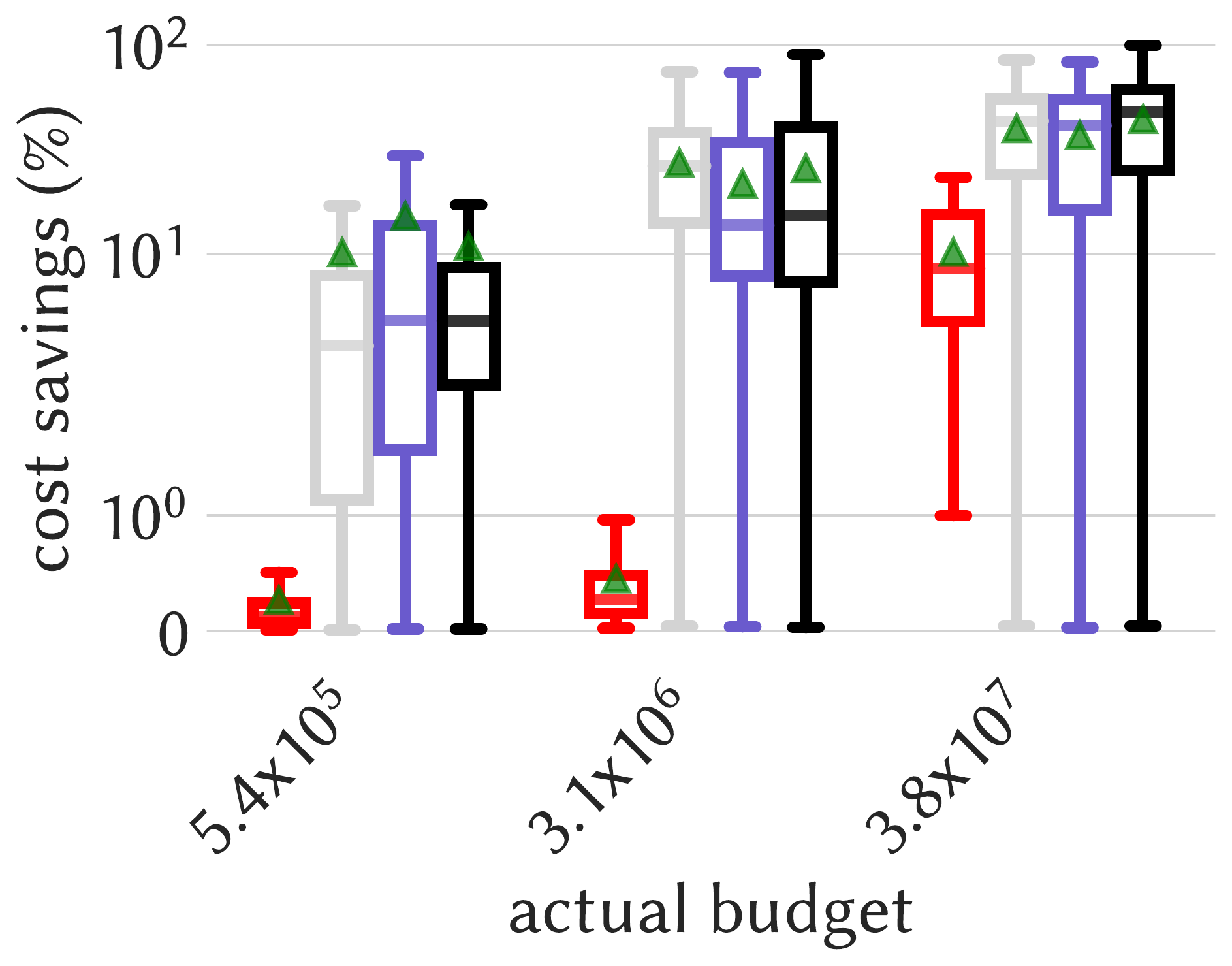}&
		\hspace{-5mm} \includegraphics[width=0.53\columnwidth, height = 0.63\textheight, 		keepaspectratio]{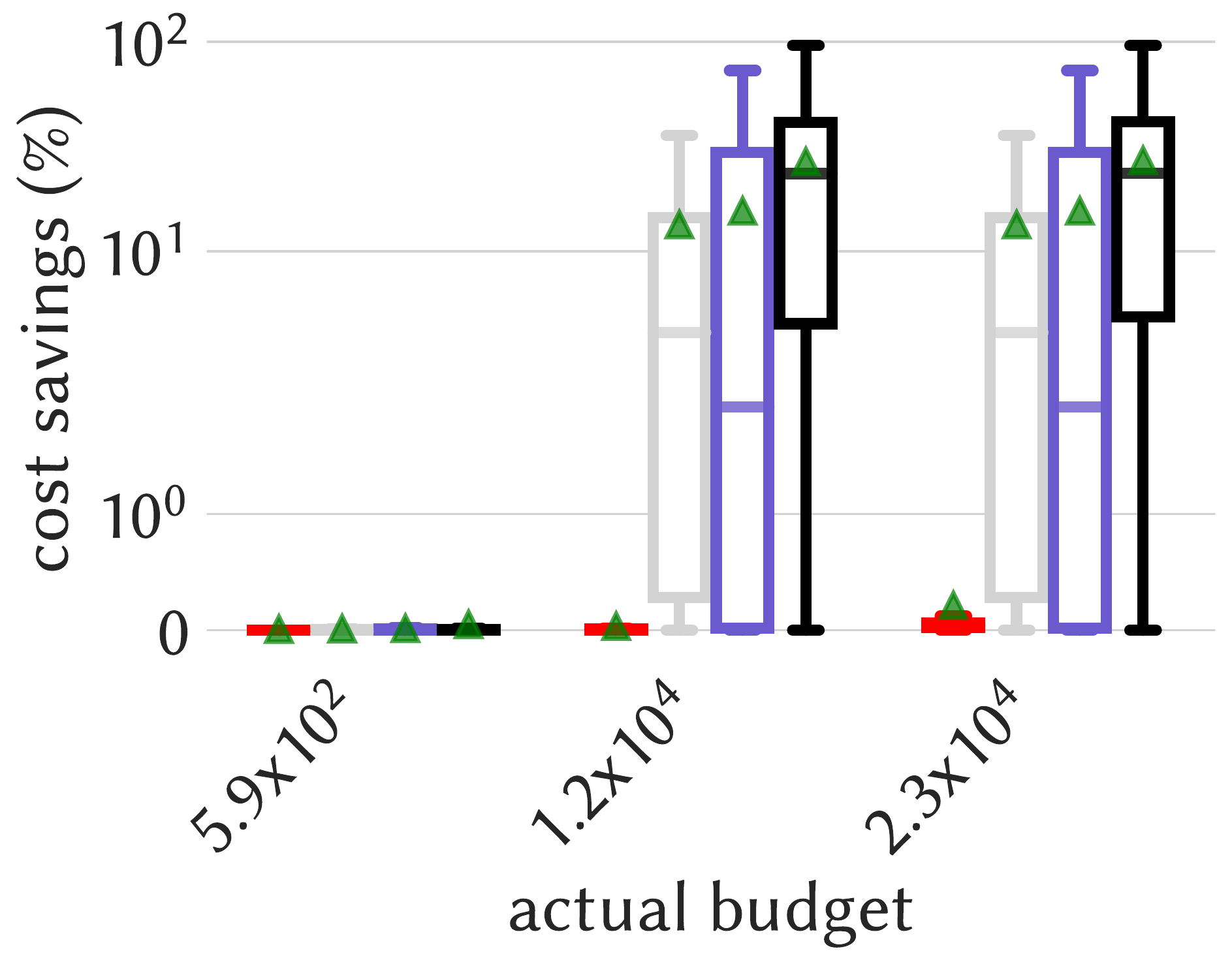}&
		\hspace{-5mm} \includegraphics[width=0.53\columnwidth, height = 0.63\textheight, keepaspectratio]{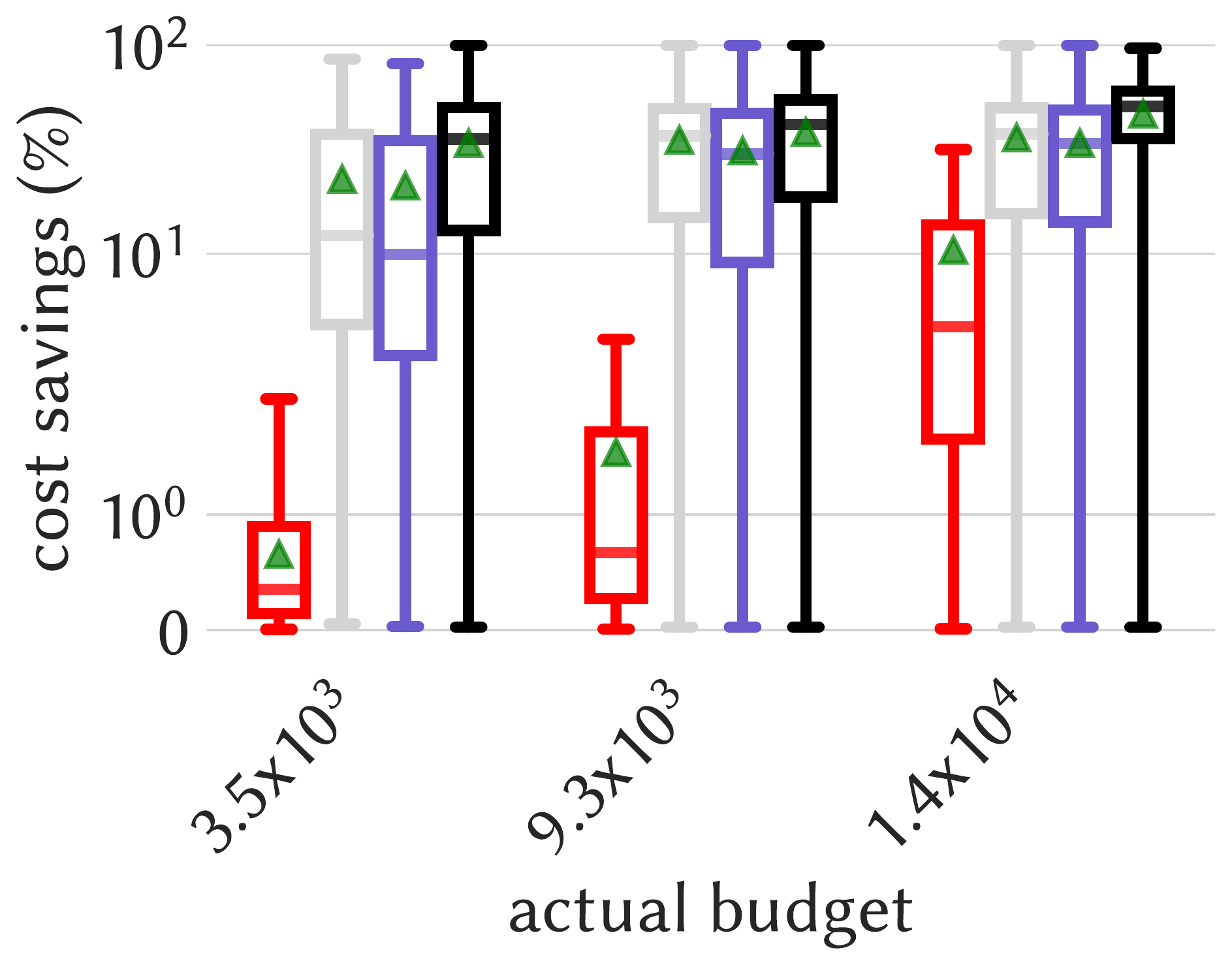}\\
		\textsc{Child} &  \textsc{Hepar II} & \textsc{Andes} &  \textsc{Hailfinder} \vspace{0mm}\\
		\hspace{-5mm} \includegraphics[width=0.53\columnwidth, height = 0.63\textheight, keepaspectratio]{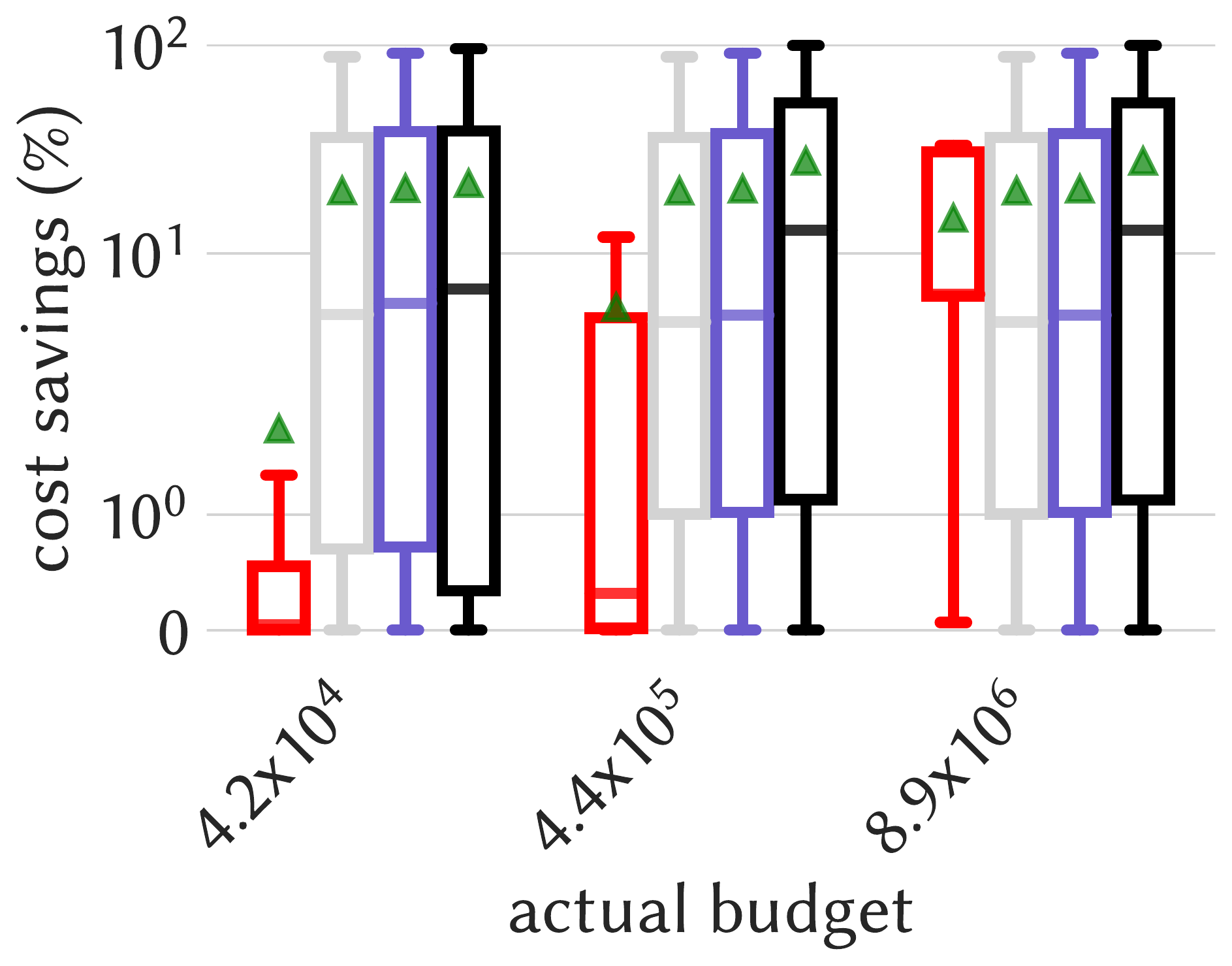}&
		\hspace{-5mm} \includegraphics[width=0.53\columnwidth, height = 0.63\textheight, keepaspectratio]{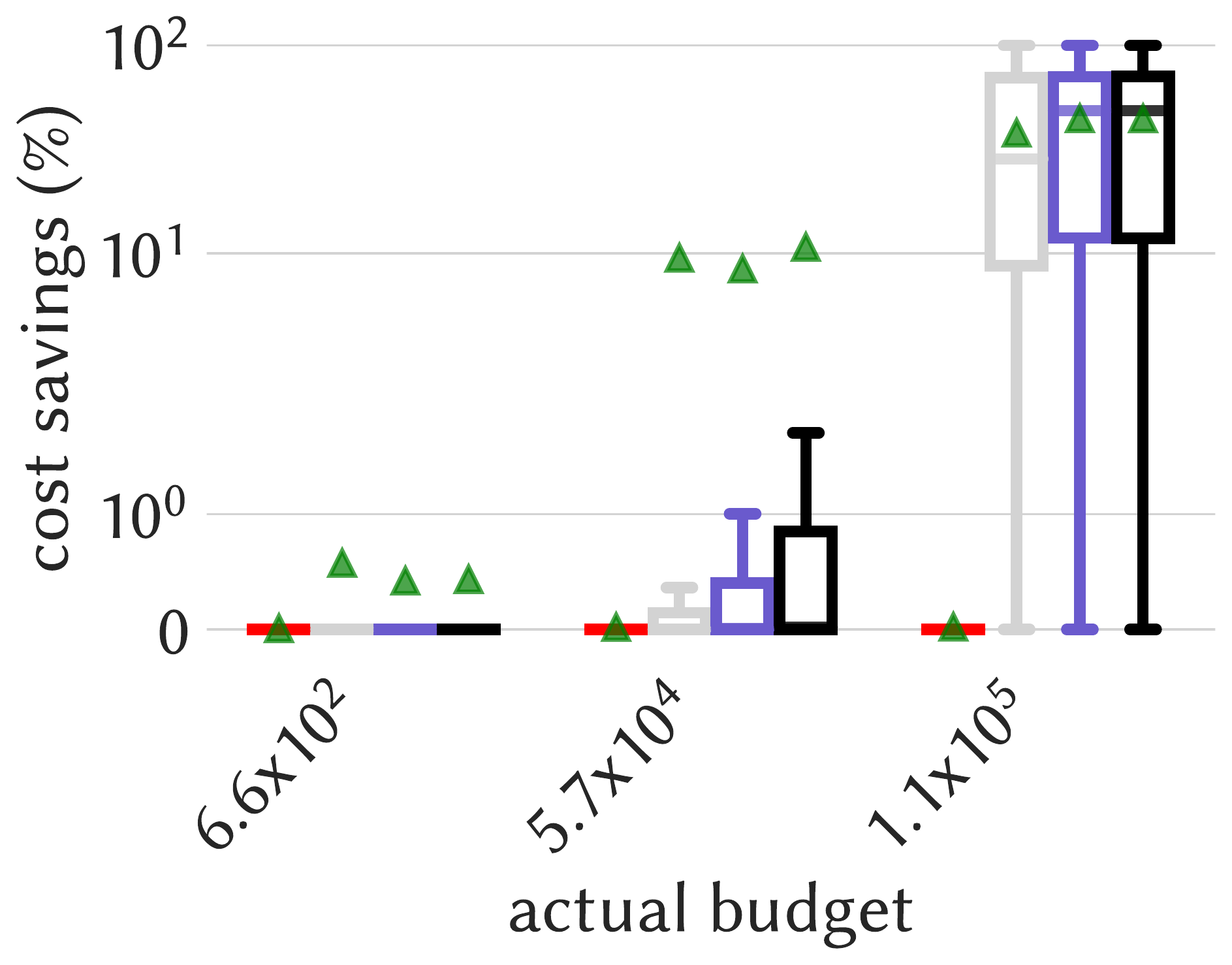}&
		\hspace{-5mm} \includegraphics[width=0.53\columnwidth, height = 0.63\textheight, keepaspectratio]{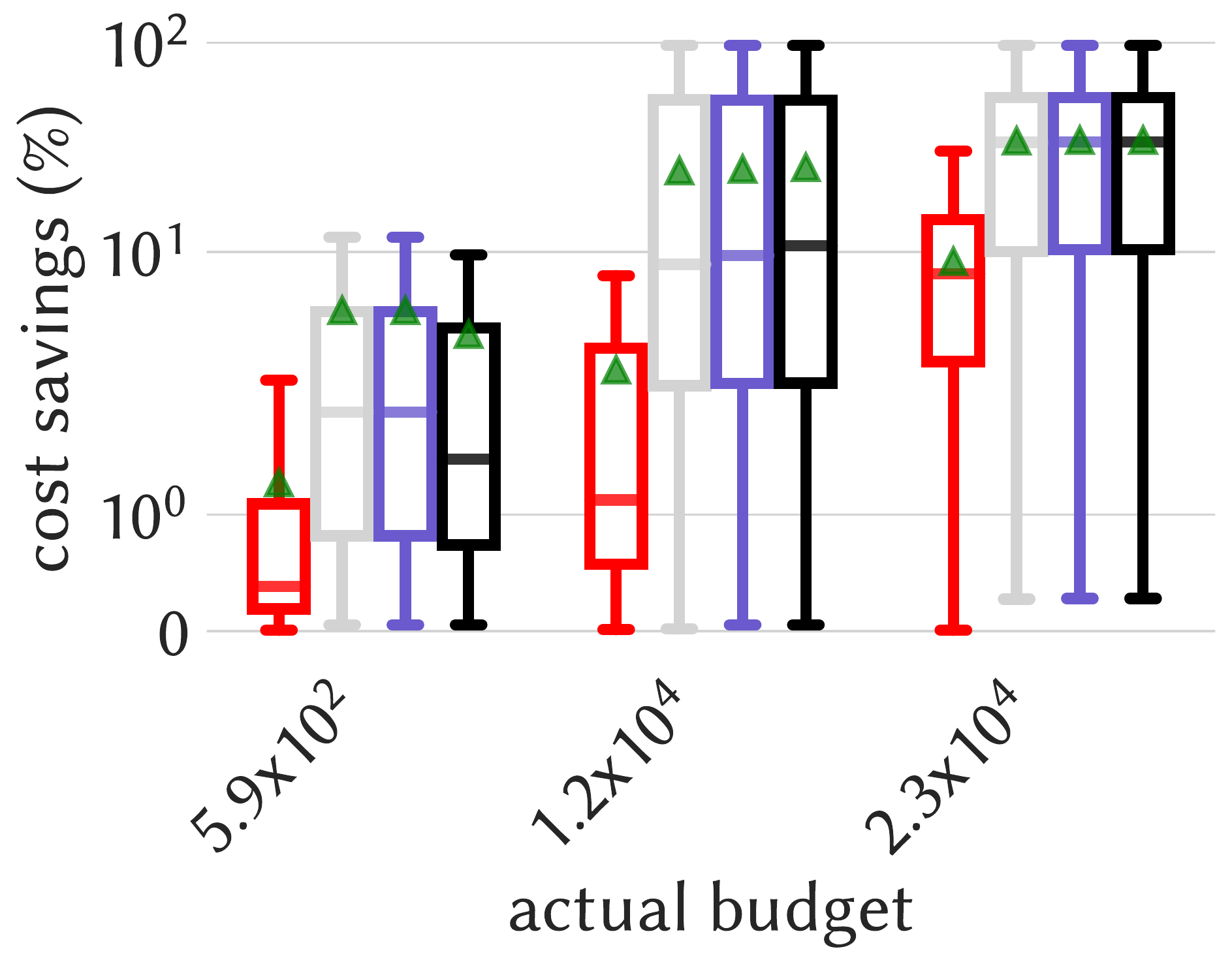}&
		\hspace{-5mm} \includegraphics[width=0.53\columnwidth, height = 0.63\textheight, keepaspectratio]{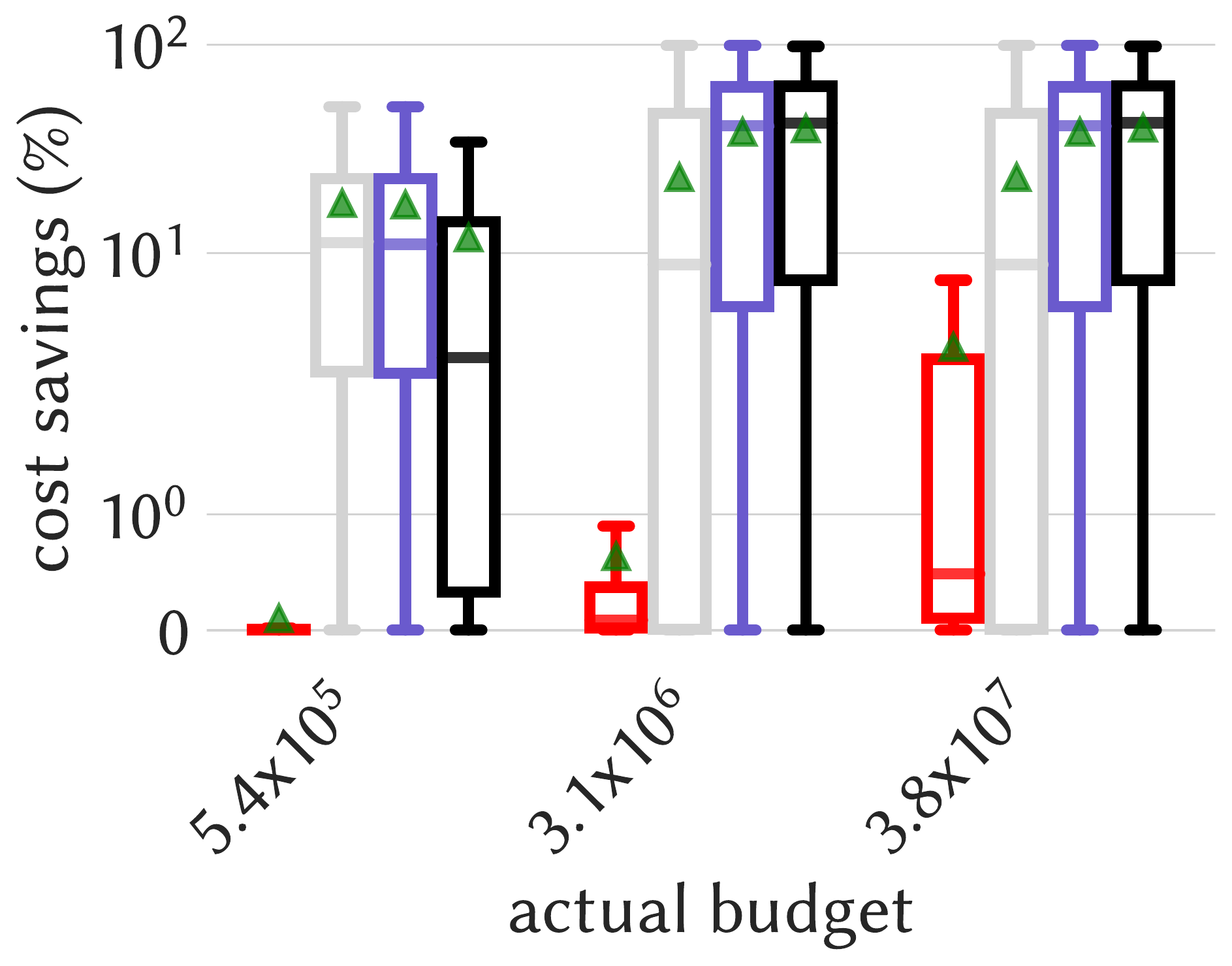}\\
		\textsc{TPC-H} &  \textsc{Munin} & \textsc{Pathfinder} &  \textsc{Barley} \vspace{0mm}\\
	\end{tabular}
	\hspace{-5mm}	\caption{\label{fig:distributions_cost_savings} Distribution of cost savings percentage against materialized budget for \indsep and	\ouralgorithmplus with different approximation levels. The average of the distribution is displayed in green.  The $y$-axis is on a logarithmic scale.    }	
\end{figure*}

\begin{figure*}[t]
		\centering
		\includegraphics[width=1.4\columnwidth, height = 1.2\textheight, keepaspectratio]{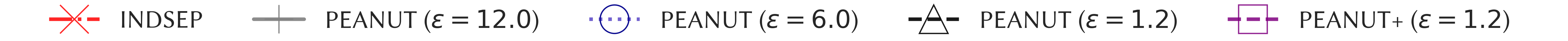}\\
		\begin{tabular}{cccc}
		\hspace{-5mm} \includegraphics[width=0.53\columnwidth, height = 0.63\textheight, keepaspectratio]{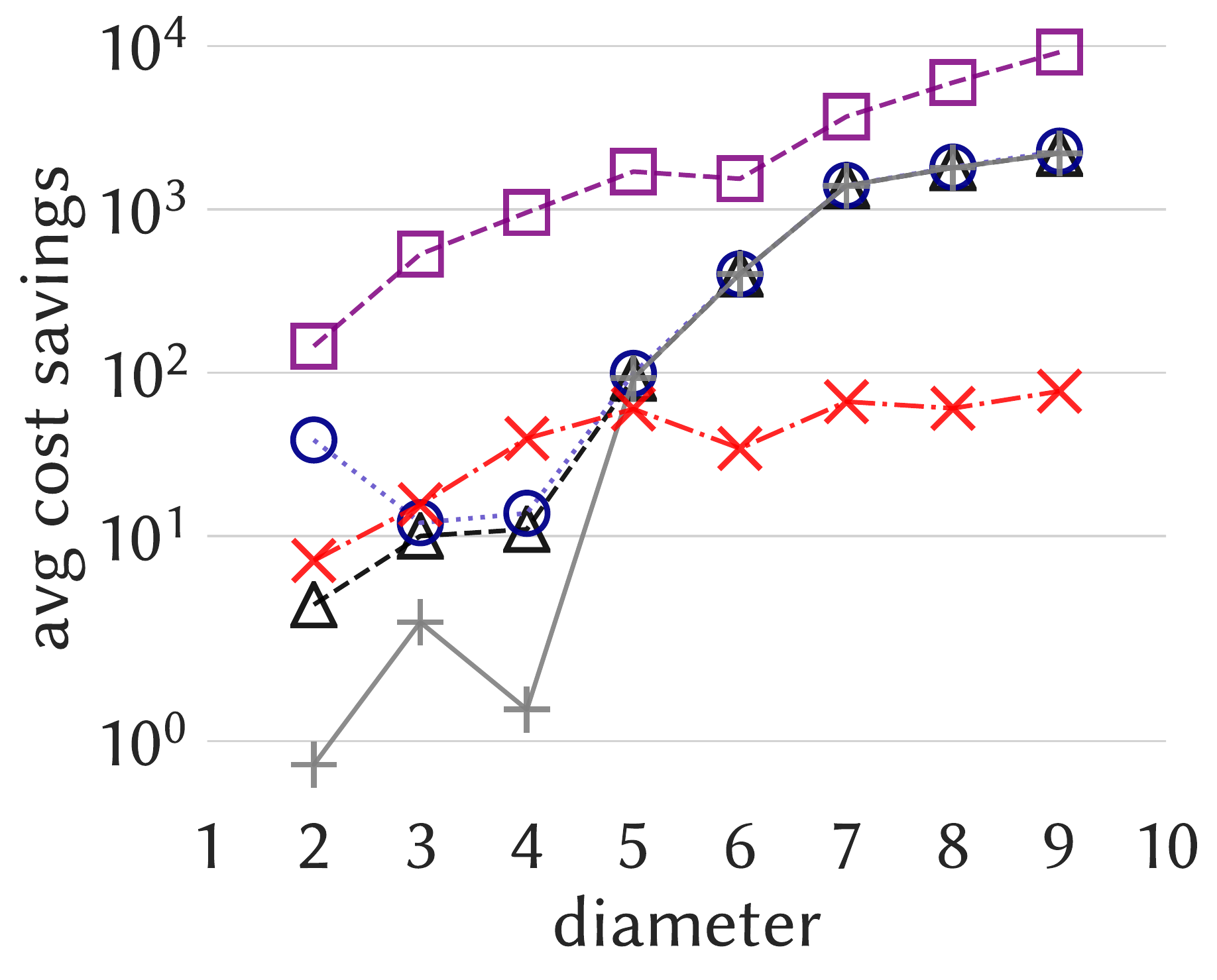}&
	\hspace{-5mm} \includegraphics[width=0.53\columnwidth, height = 0.63\textheight, keepaspectratio]{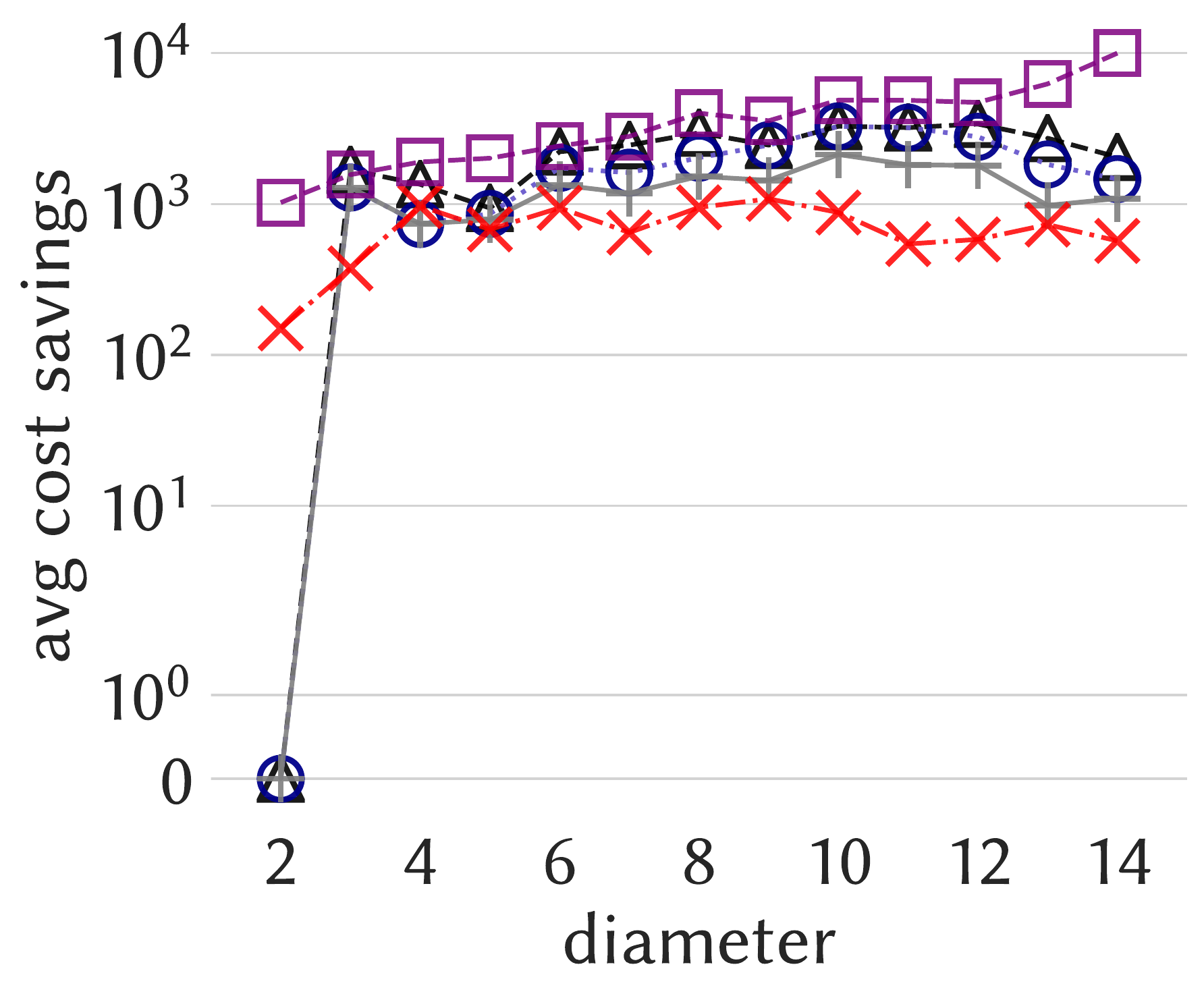} &
		\hspace{-5mm} \includegraphics[width=0.53\columnwidth, height = 0.63\textheight, keepaspectratio]{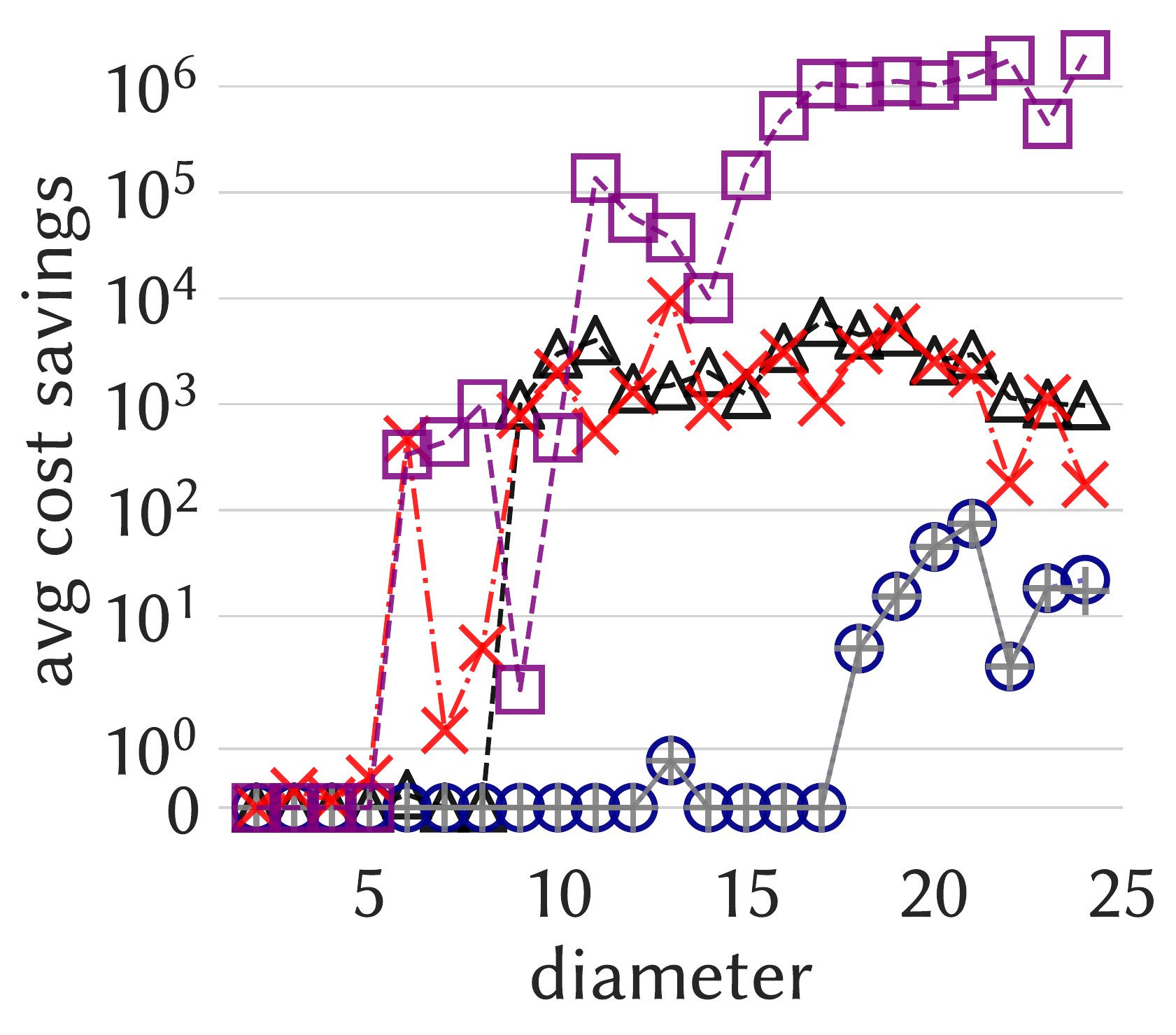}& 
		\includegraphics[width=0.52\columnwidth, height = 0.5\textheight, keepaspectratio]{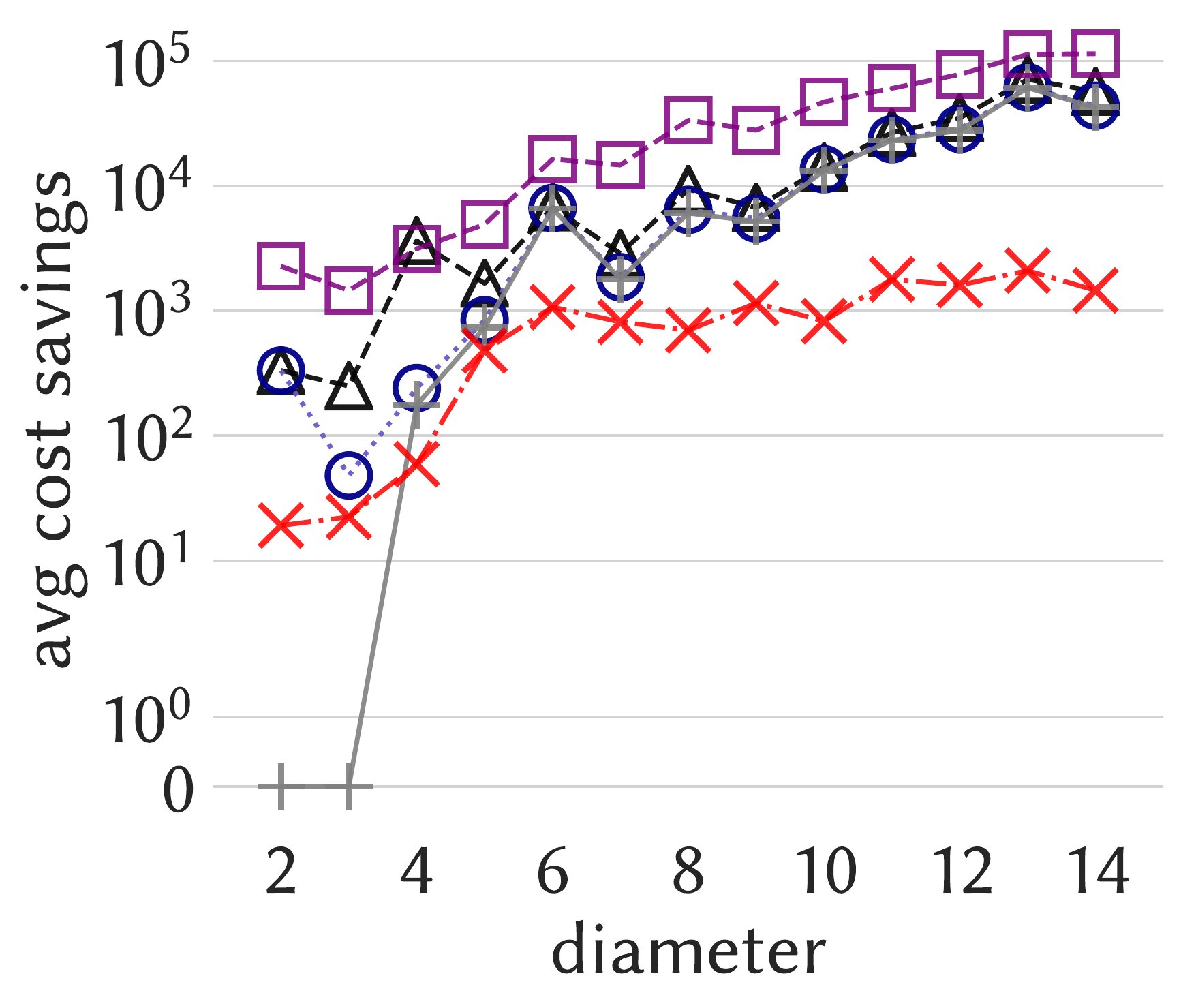} \\
		\textsc{Child} & \textsc{Hepar II} & \textsc{Andes} & \textsc{HailFinder} \vspace{-1mm}\\
		\hspace{-5mm} \includegraphics[width=0.53\columnwidth, height = 0.63\textheight, keepaspectratio]{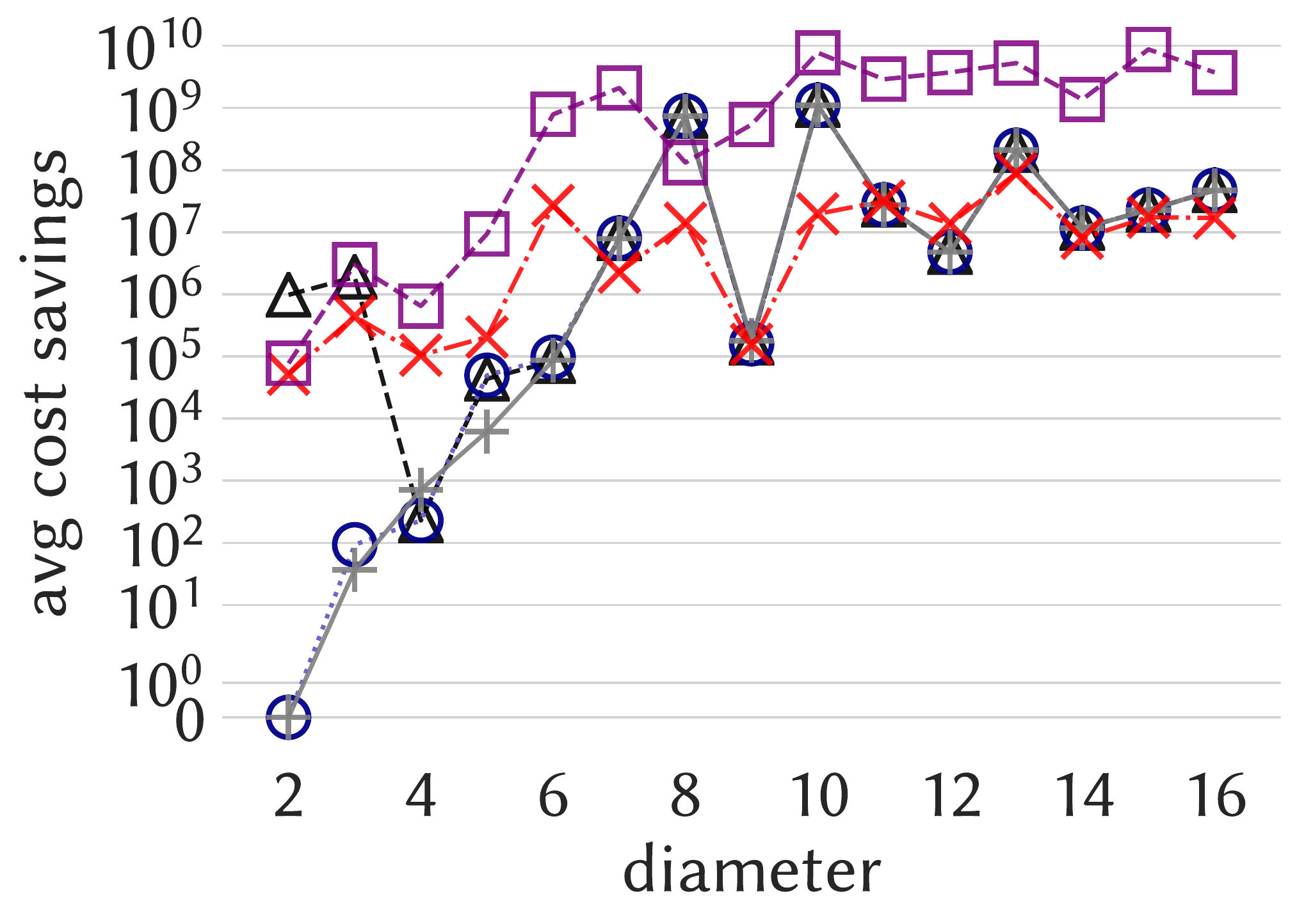}&
		\hspace{-5mm} \includegraphics[width=0.53\columnwidth, height = 0.53\textheight, keepaspectratio]{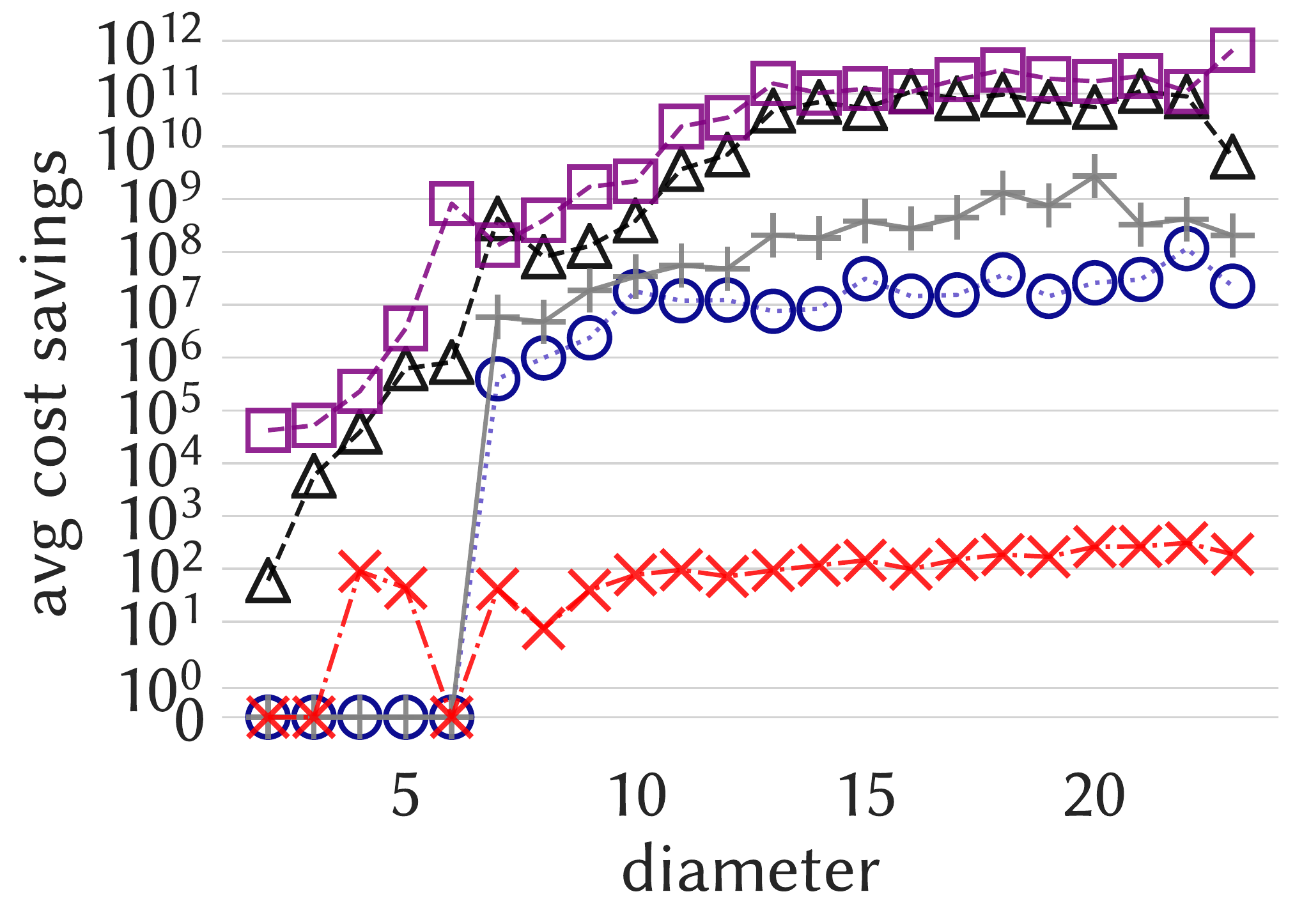} &
		\hspace{-5mm} \includegraphics[width=0.53\columnwidth, height = 0.63\textheight, keepaspectratio]{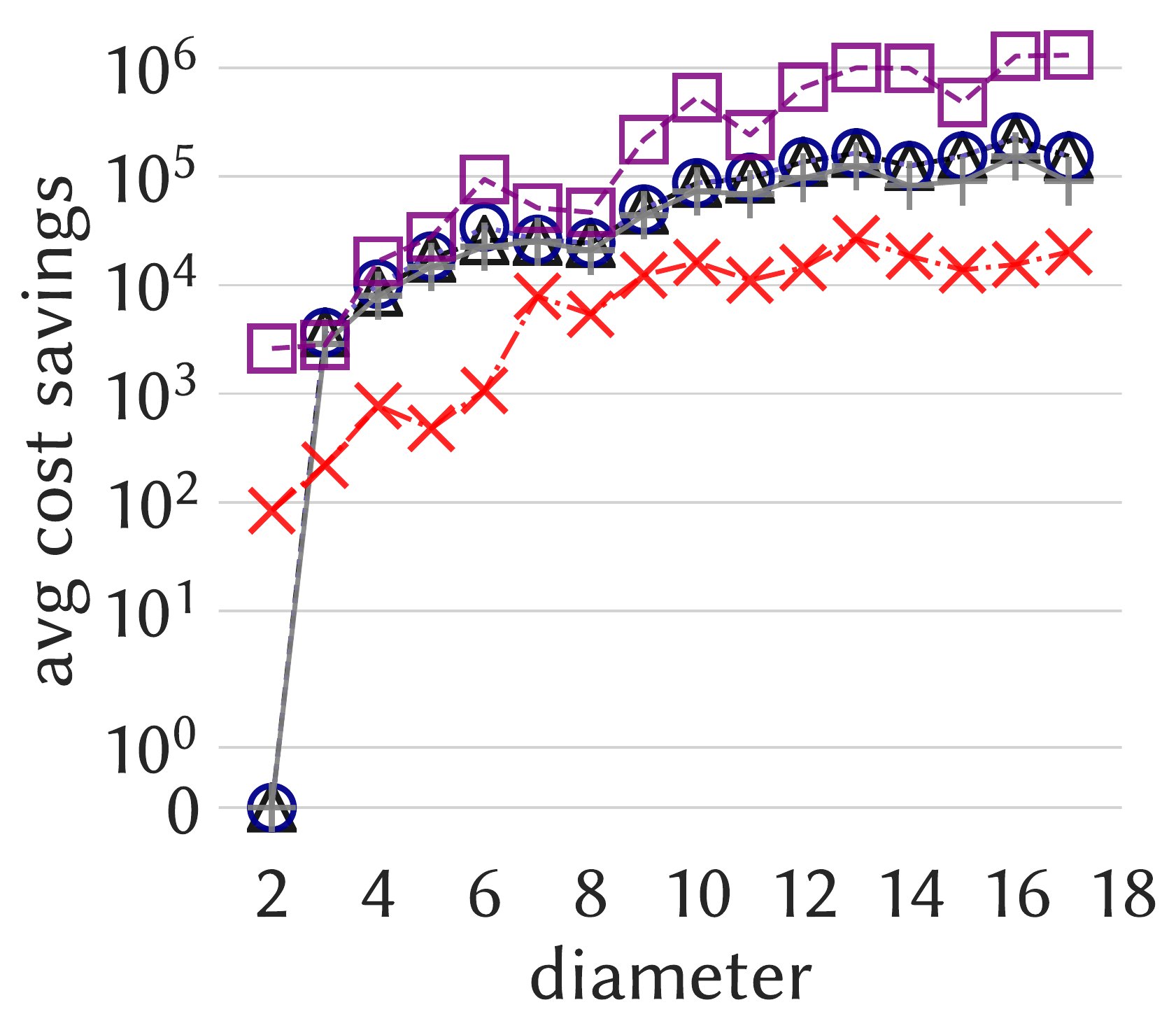} & 
		\hspace{-5mm} \includegraphics[width=0.53\columnwidth, height = 0.53\textheight, keepaspectratio]{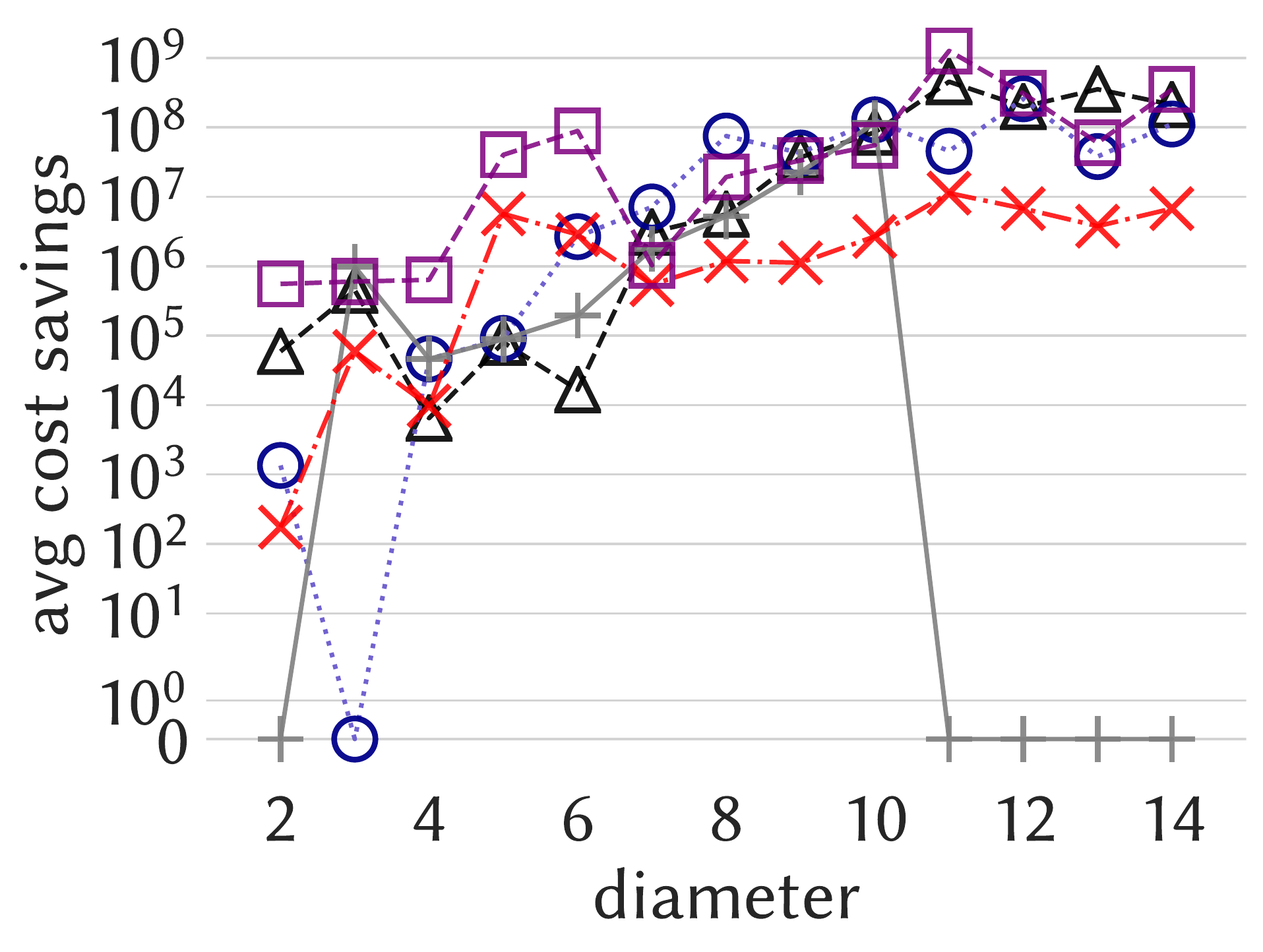} \\
		\textsc{TPC-H} &  \textsc{Munin} & \textsc{Pathfinder} &  \textsc{Barley} \vspace{0mm}\\ 
	\end{tabular}
	\hspace{-5mm} \caption{\label{fig:diameter} 
		\revisioncol{Average cost savings against Steiner-tree diameter for \indsep, \ouralgorithmplus and \ouralgorithm with different approximation levels. 
		The $y$-axis is on a logarithmic scale.}  
	}
\end{figure*}

\begin{figure*}[t]
	\centering
	\includegraphics[width=1\columnwidth, height = 0.8\textheight, keepaspectratio]{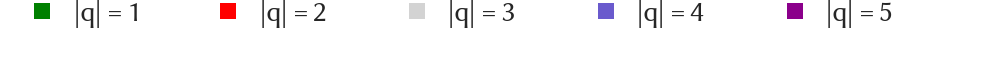}\\
	\begin{tabular}{cccc}
		\hspace{-5mm} \includegraphics[width=0.53\columnwidth, height = 0.63\textheight, keepaspectratio]{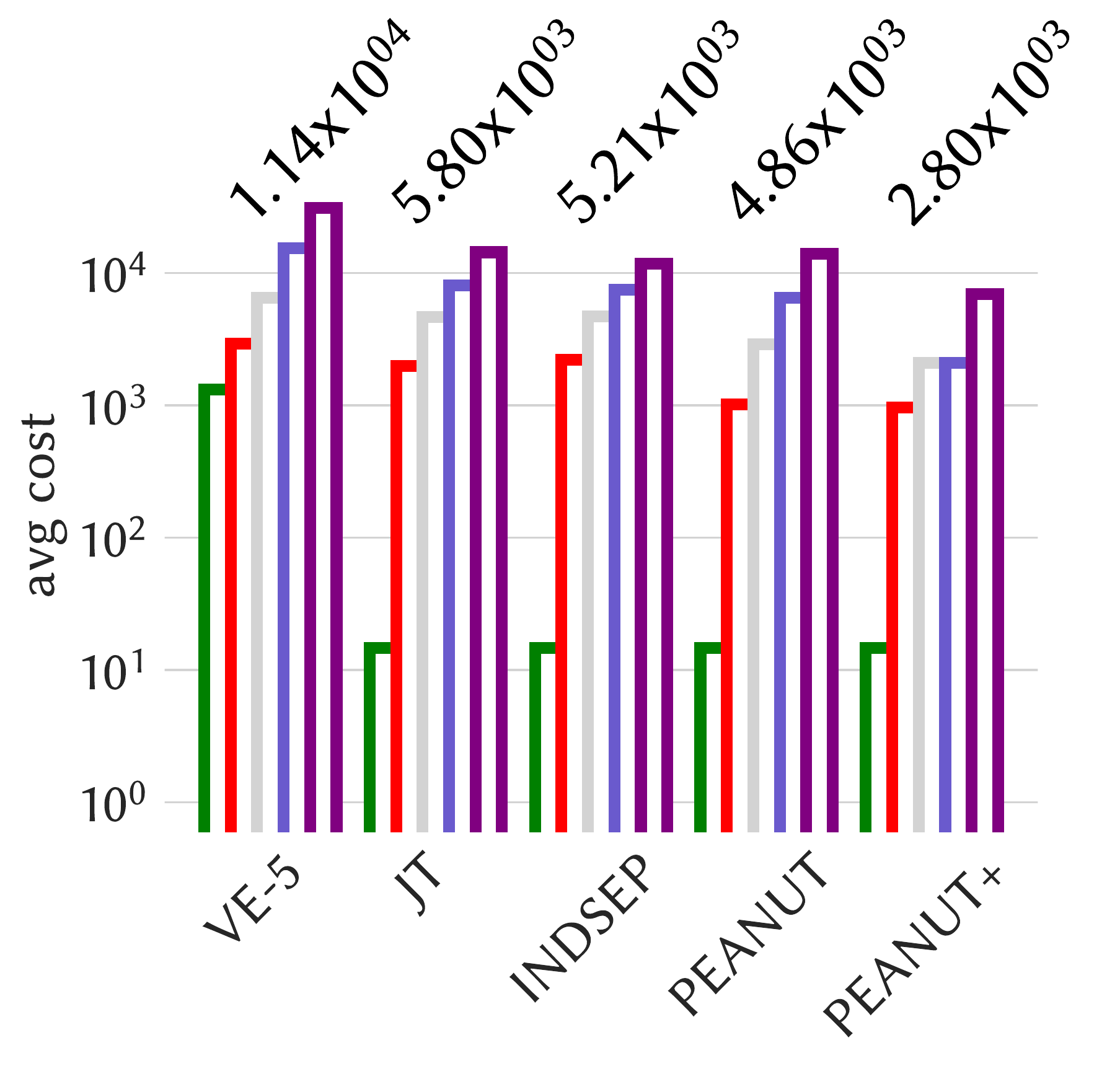}&
	   \hspace{-5mm} \includegraphics[width=0.53\columnwidth, height = 0.63\textheight, keepaspectratio]{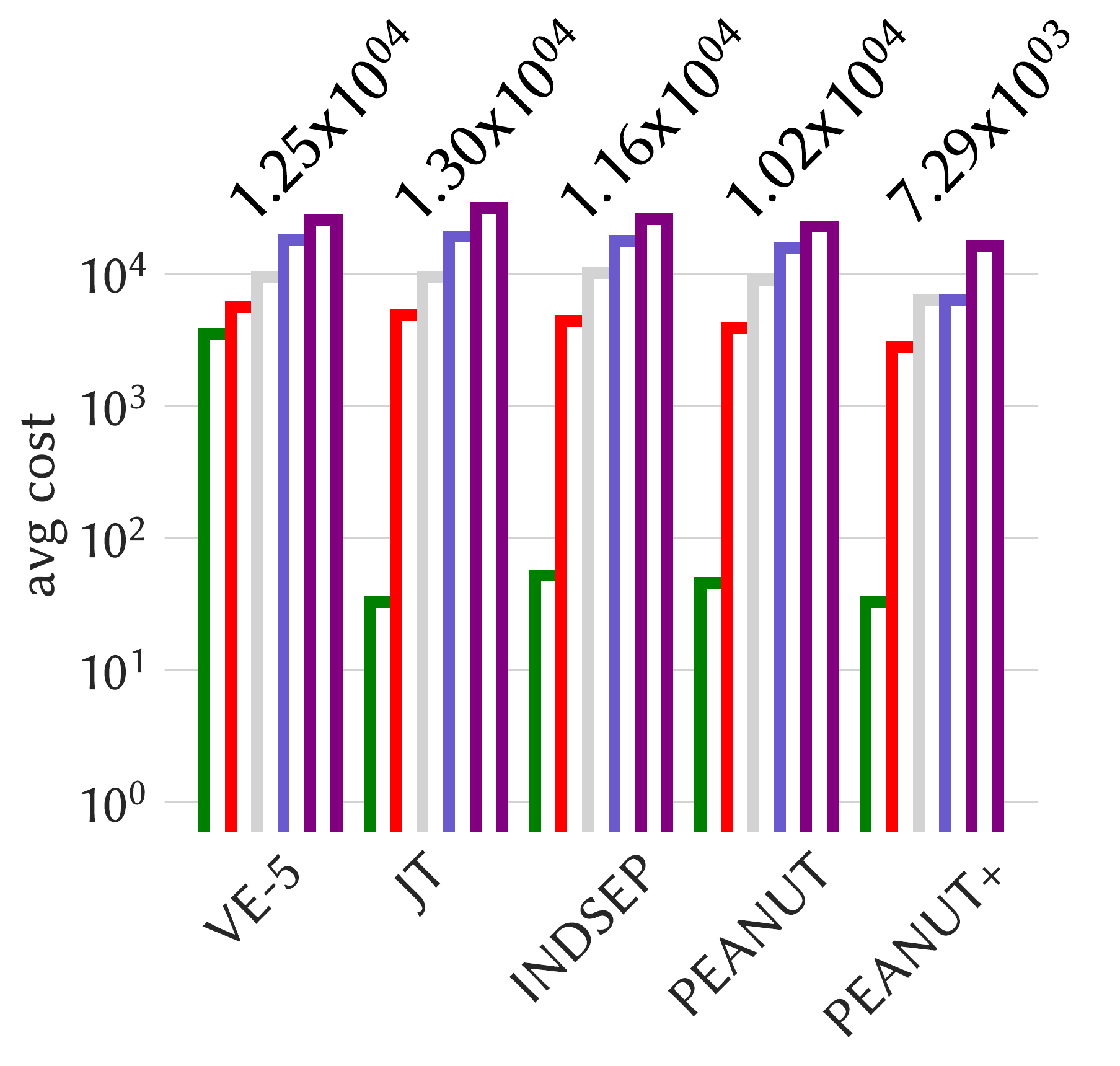}&
		\hspace{-5mm} \includegraphics[width=0.53\columnwidth, height = 0.63\textheight, keepaspectratio]{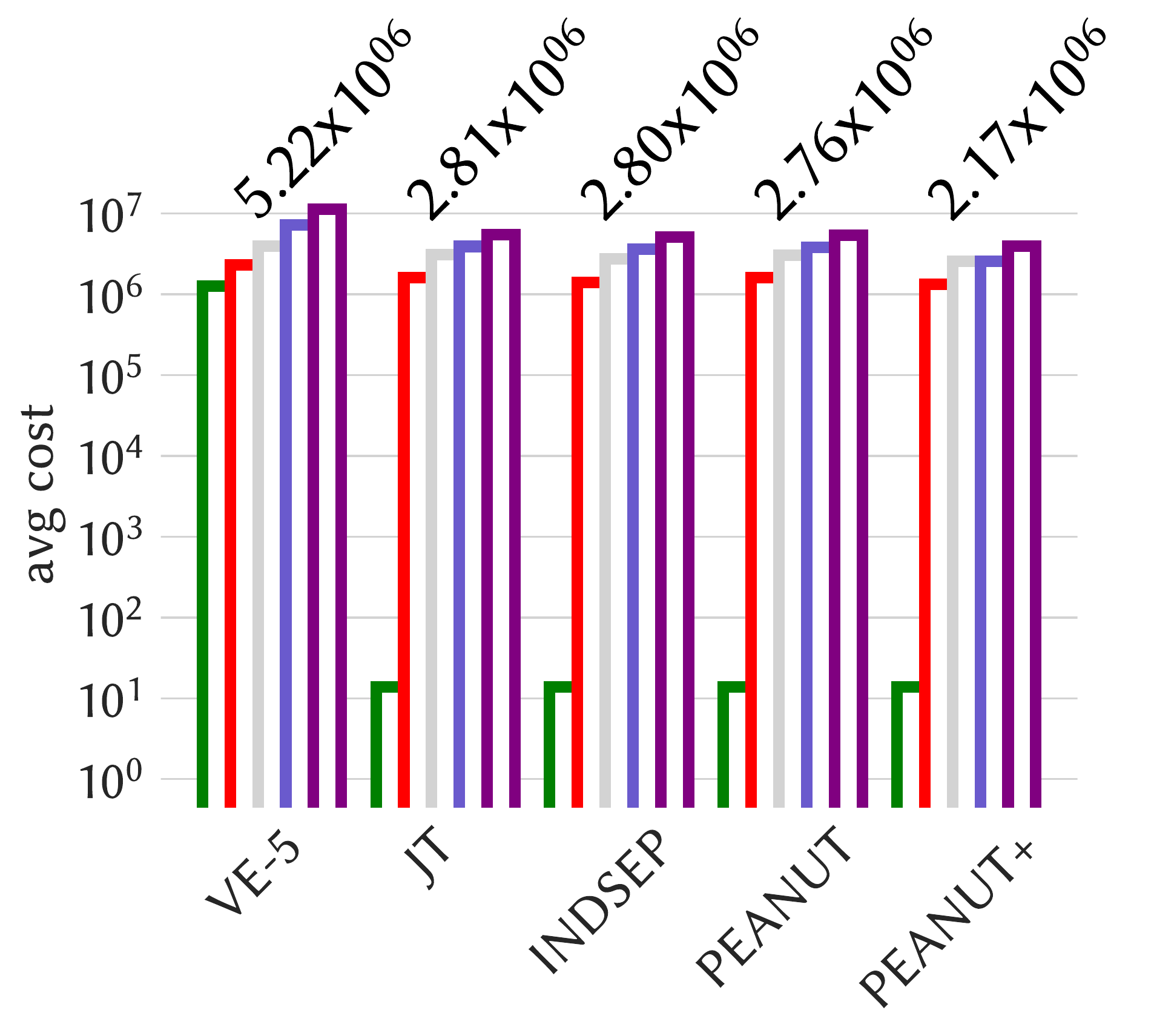}&
		\hspace{-5mm} \includegraphics[width=0.53\columnwidth, height = 0.63\textheight, keepaspectratio]{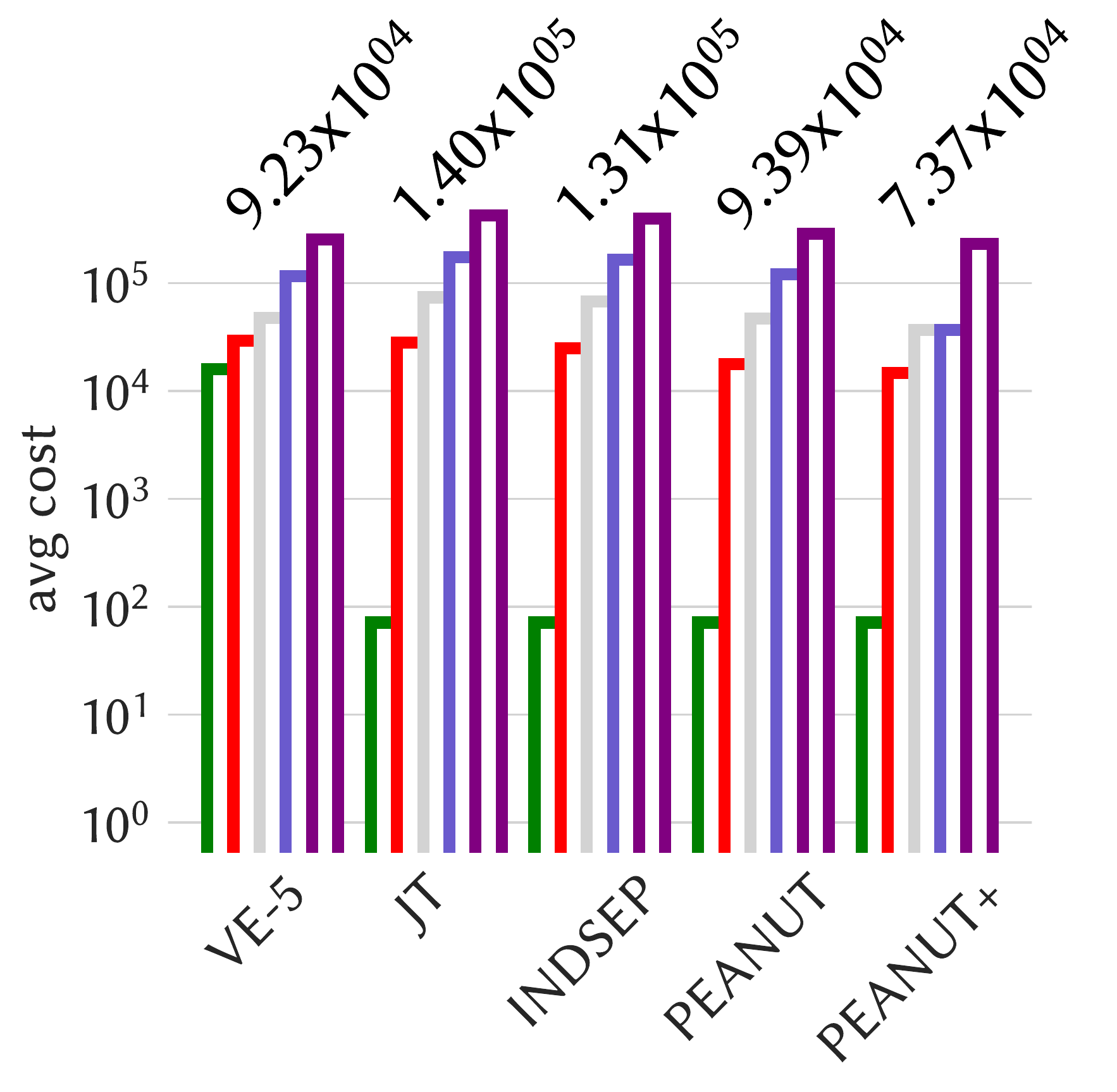}\\
	    \textsc{Child} &  \textsc{Hepar II} & \textsc{Andes} & \textsc{HailFinder} \vspace{-1mm}\\
		\hspace{-5mm} \includegraphics[width=0.53\columnwidth, height = 0.63\textheight, keepaspectratio]{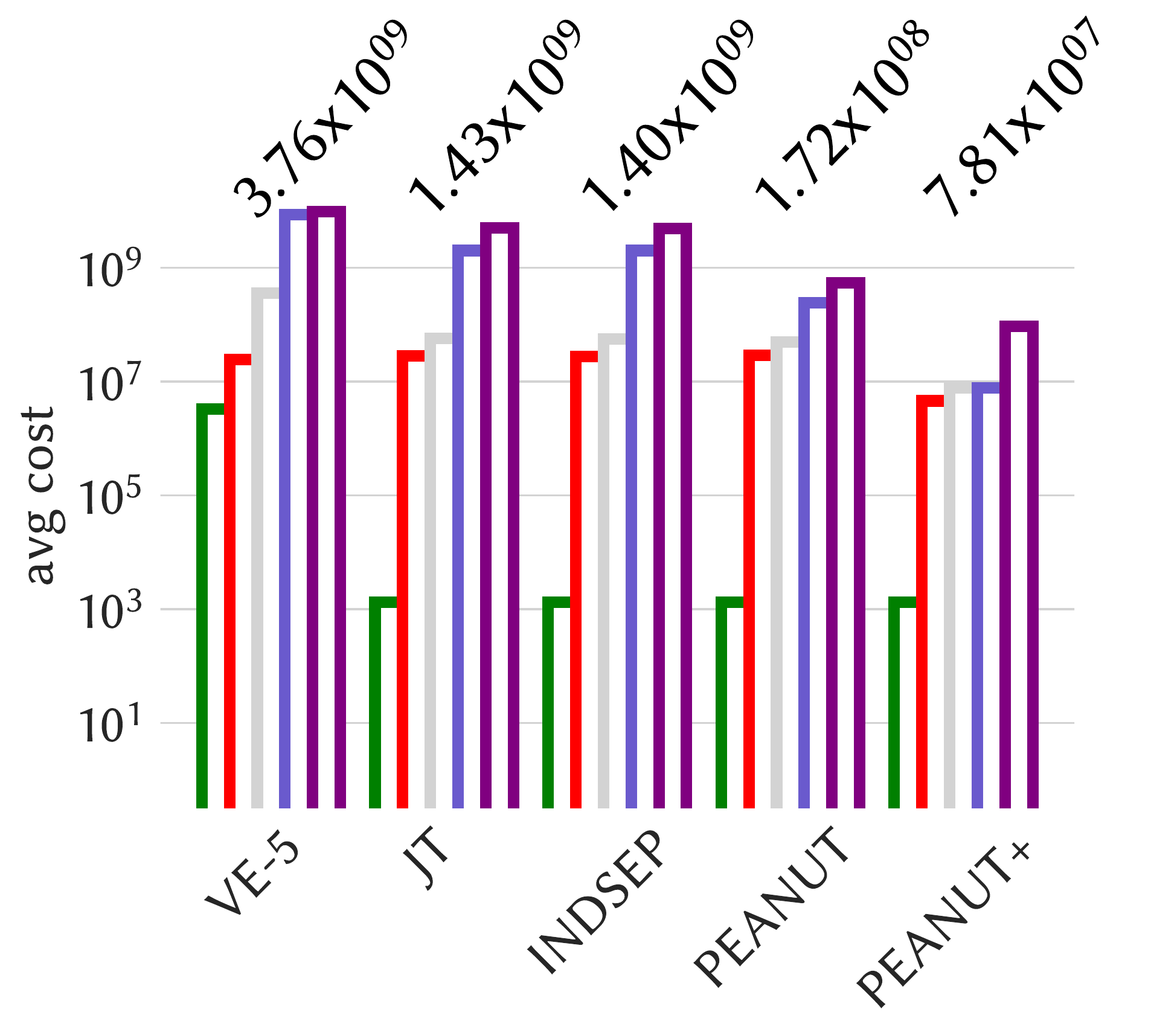}&
 		\hspace{-5mm} \includegraphics[width=0.53\columnwidth, height = 0.63\textheight, keepaspectratio]{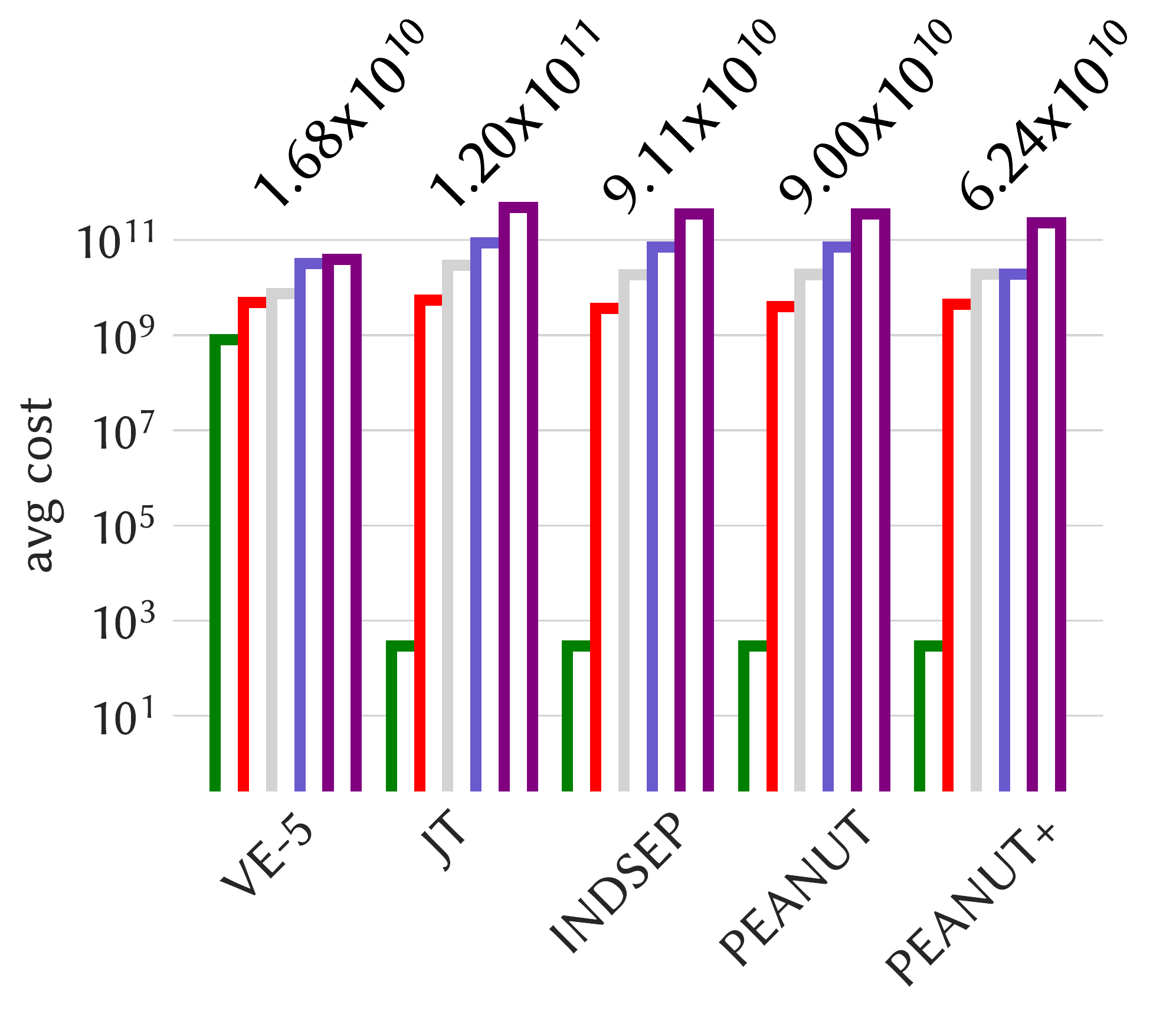}&
		\hspace{-5mm} \includegraphics[width=0.53\columnwidth, height = 0.63\textheight, keepaspectratio]{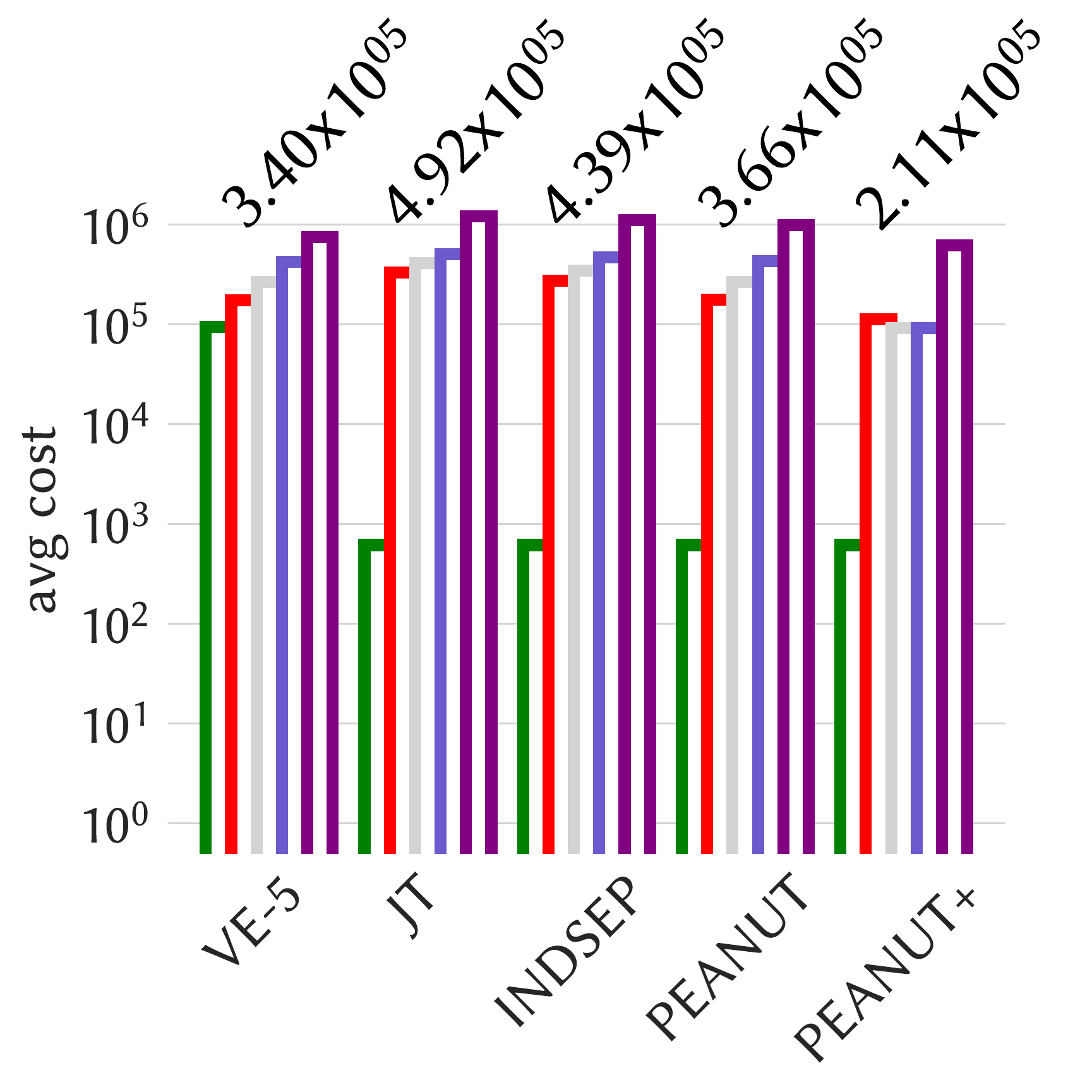} & 
		\hspace{-5mm} \includegraphics[width=0.53\columnwidth, height = 0.63\textheight, keepaspectratio]{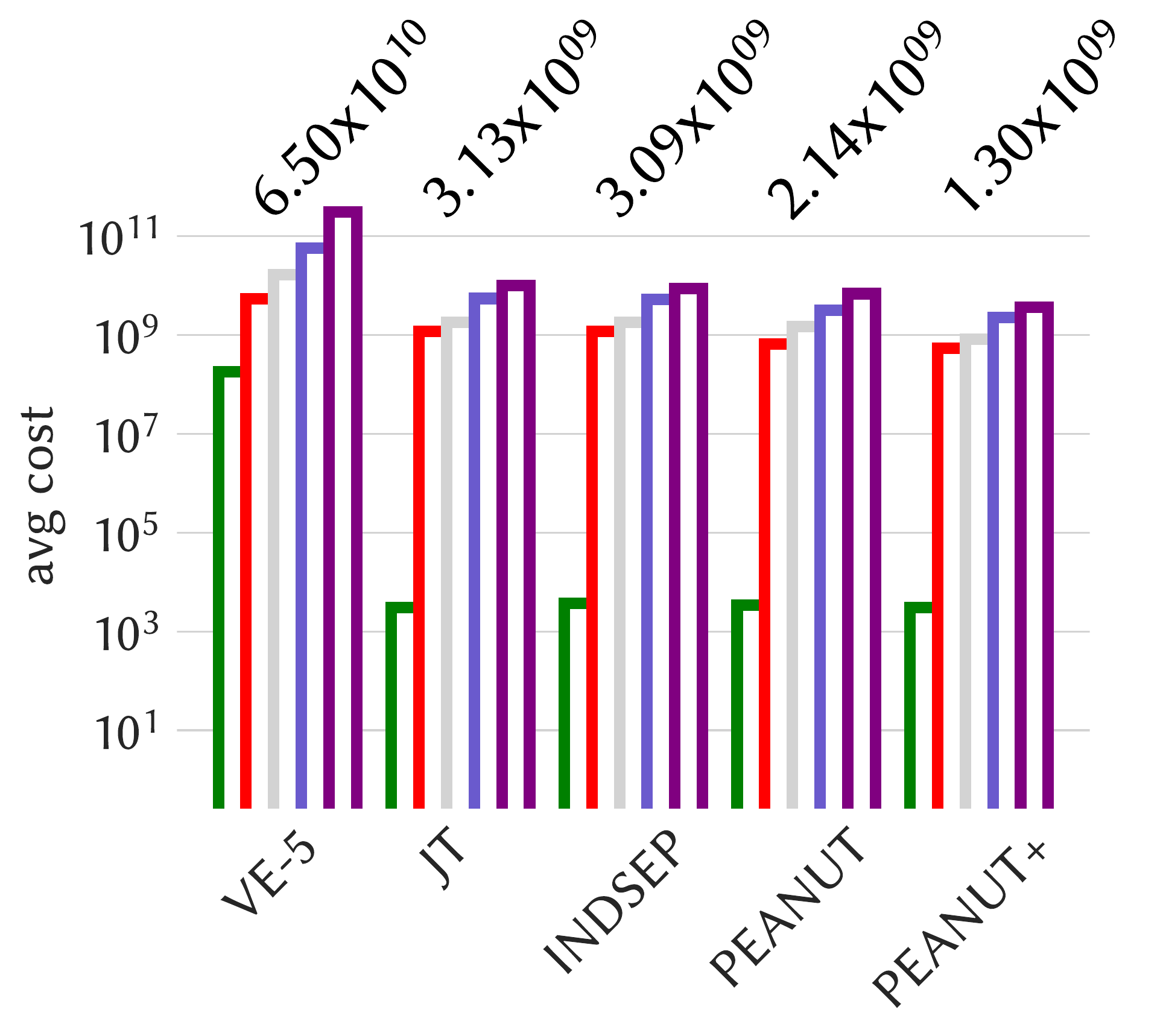}\\
		\textsc{TPC-H} &  \textsc{Munin} & \textsc{Pathfinder} &  \textsc{Barley} \vspace{0mm}\\  
	\end{tabular}
	\vspace*{-0.3cm}
	\hspace{-5mm}\caption{\label{fig:comparison_qtm} 
	\ReviewOnly{\revisioncol{ The $y$-axis shows on logarithmic scale the value of the average query processing cost by query size $|q|$ for different algorithms. In addition, the value of the cost ratio aggregated over all values of $|q|$ is reported as text for each method.}}
	\FullOnly{The $y$-axis shows on logarithmic scale the value of the cost ratio which is computed as the ratio of the total query processing cost for different algorithms to the query processing cost of \ouralgorithmplus. The height of the bars indicate the value of the ratio by query size $|q|$. In addition, the value of the cost ratio aggregated over all values of $|q|$ is reported as text for each method.}}
\end{figure*}

\begin{figure*}[t]
	\vspace*{-0.8cm}
	\centering
	\includegraphics[width=1\columnwidth, height = 0.8\textheight, keepaspectratio]{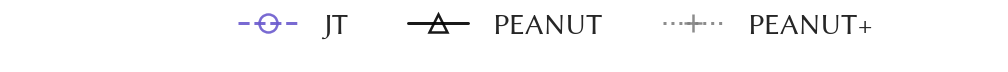}\\
	\begin{tabular}{cccc}
		\hspace{-5mm} \includegraphics[width=0.45\columnwidth, height = 0.55\textheight, keepaspectratio]{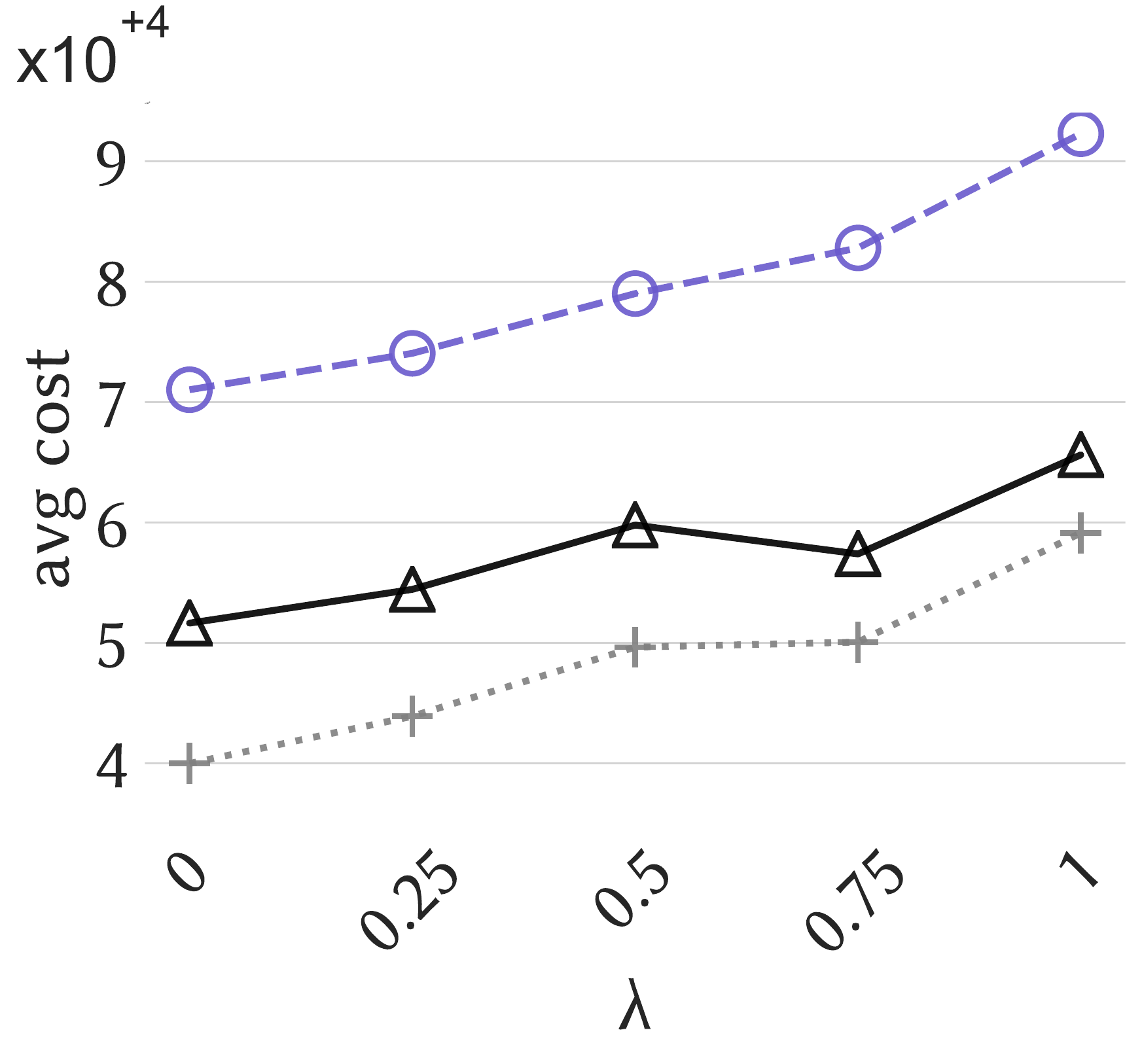}&
		\hspace{-5mm} \includegraphics[width=0.45\columnwidth, height = 0.55\textheight, keepaspectratio]{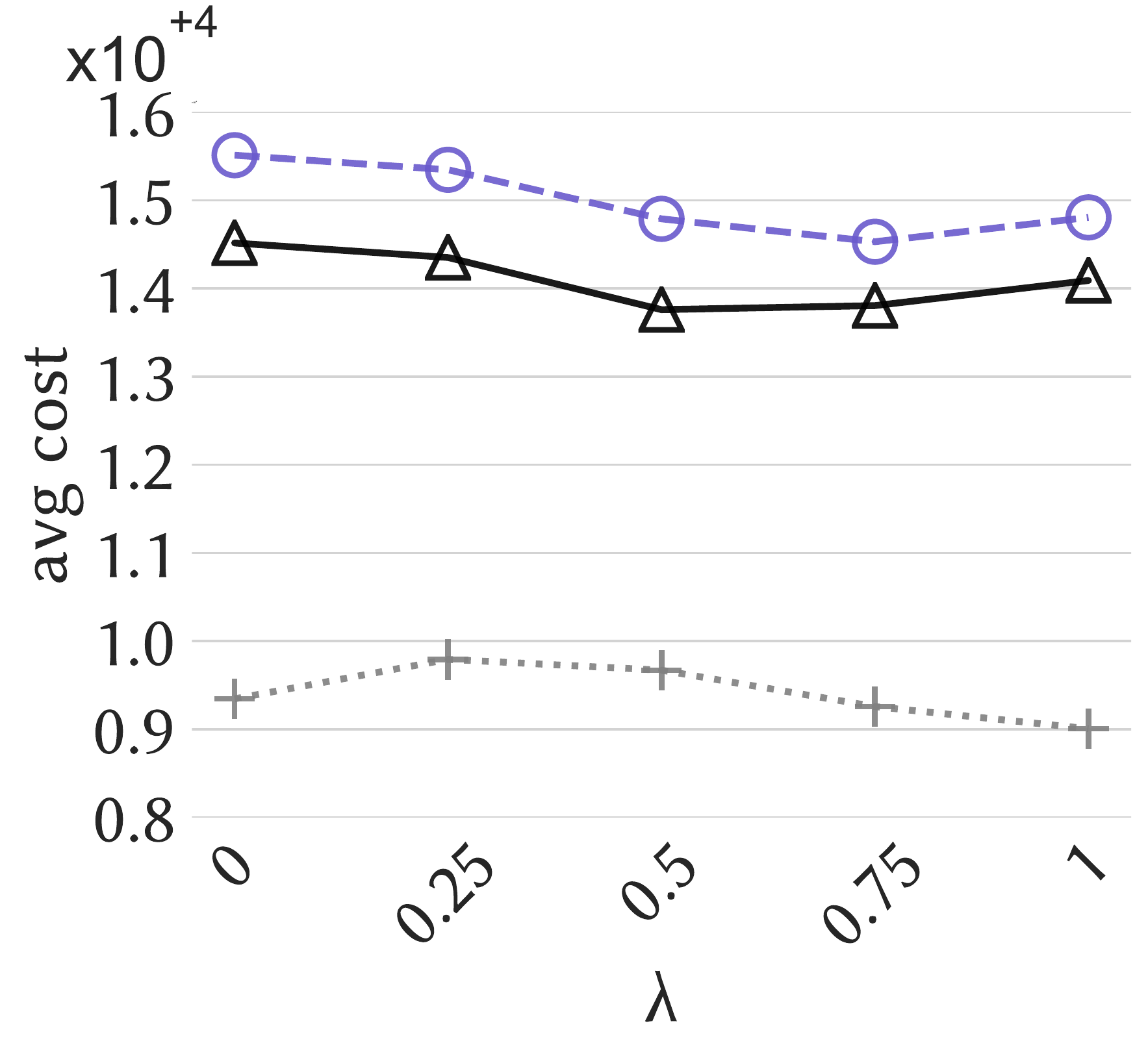}&
		\hspace{-5mm} \includegraphics[width=0.45\columnwidth, height = 0.55\textheight, keepaspectratio]{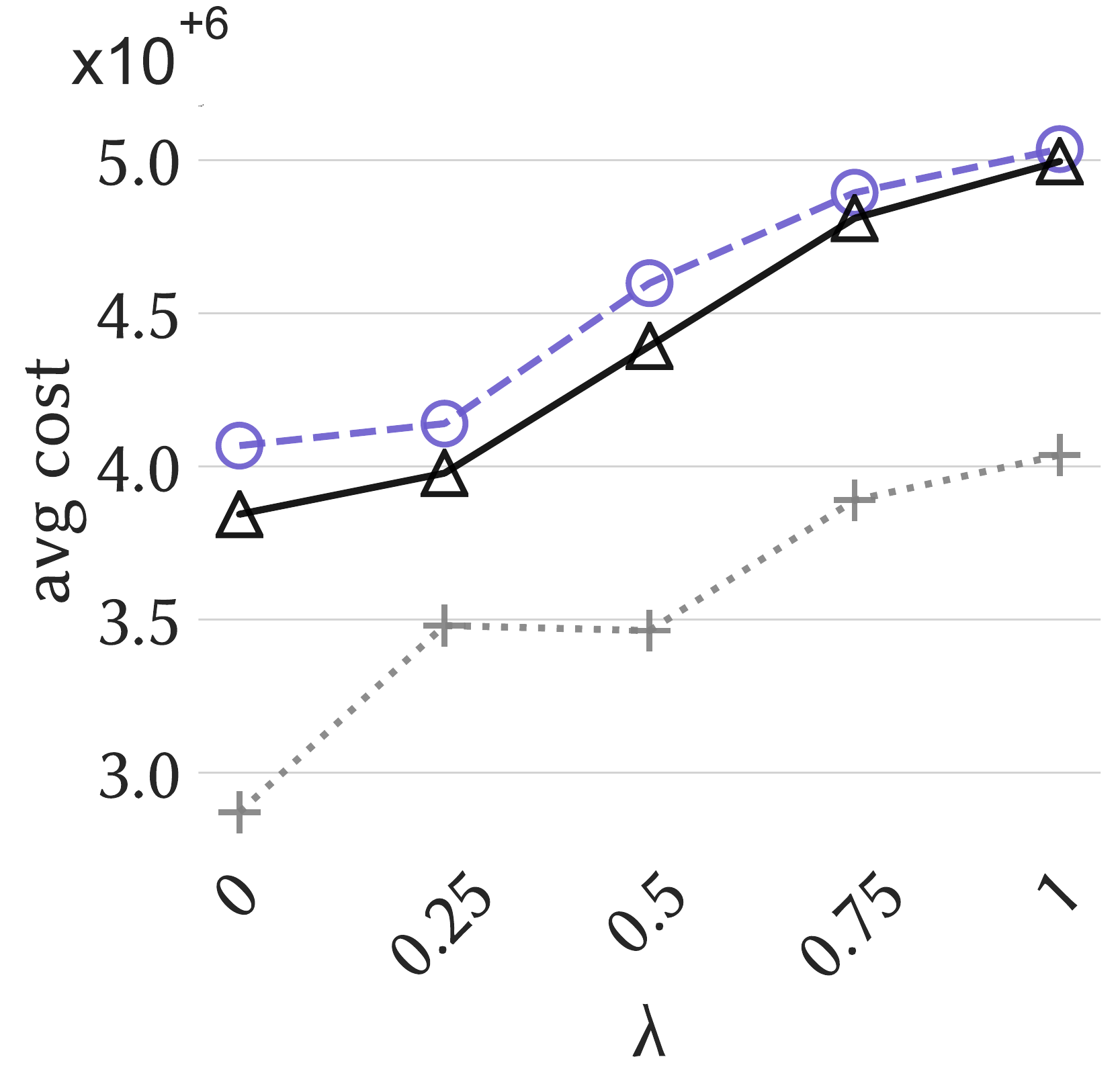}&
		\hspace{-5mm} \includegraphics[width=0.45\columnwidth, height = 0.55\textheight, keepaspectratio]{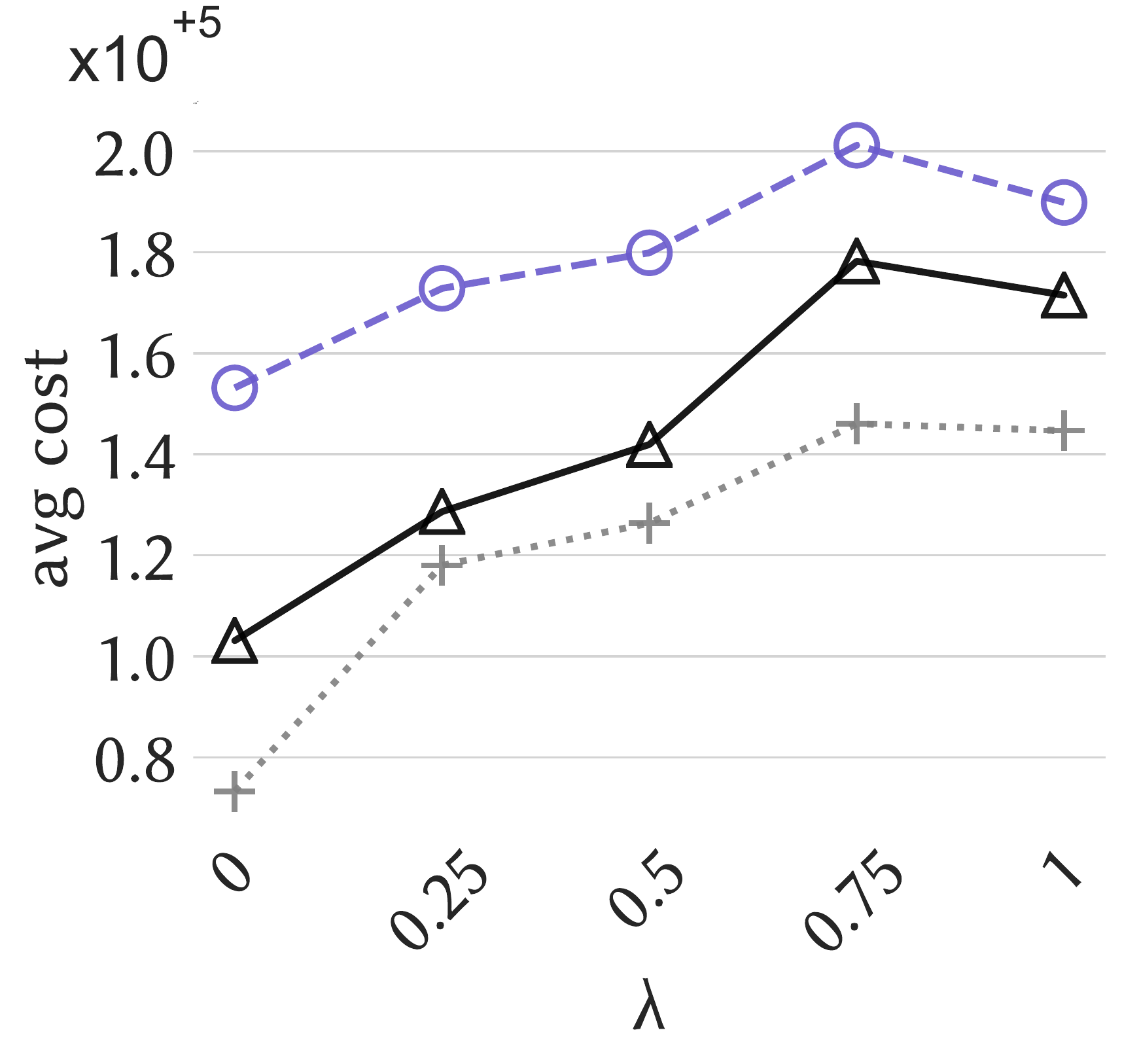}\\ 
		\textsc{Child} &  \textsc{Hepar II} & \textsc{Andes} & \textsc{HailFinder} \vspace{-1mm}\\
		\hspace{-5mm} \includegraphics[width=0.45\columnwidth, height = 0.55\textheight, keepaspectratio]{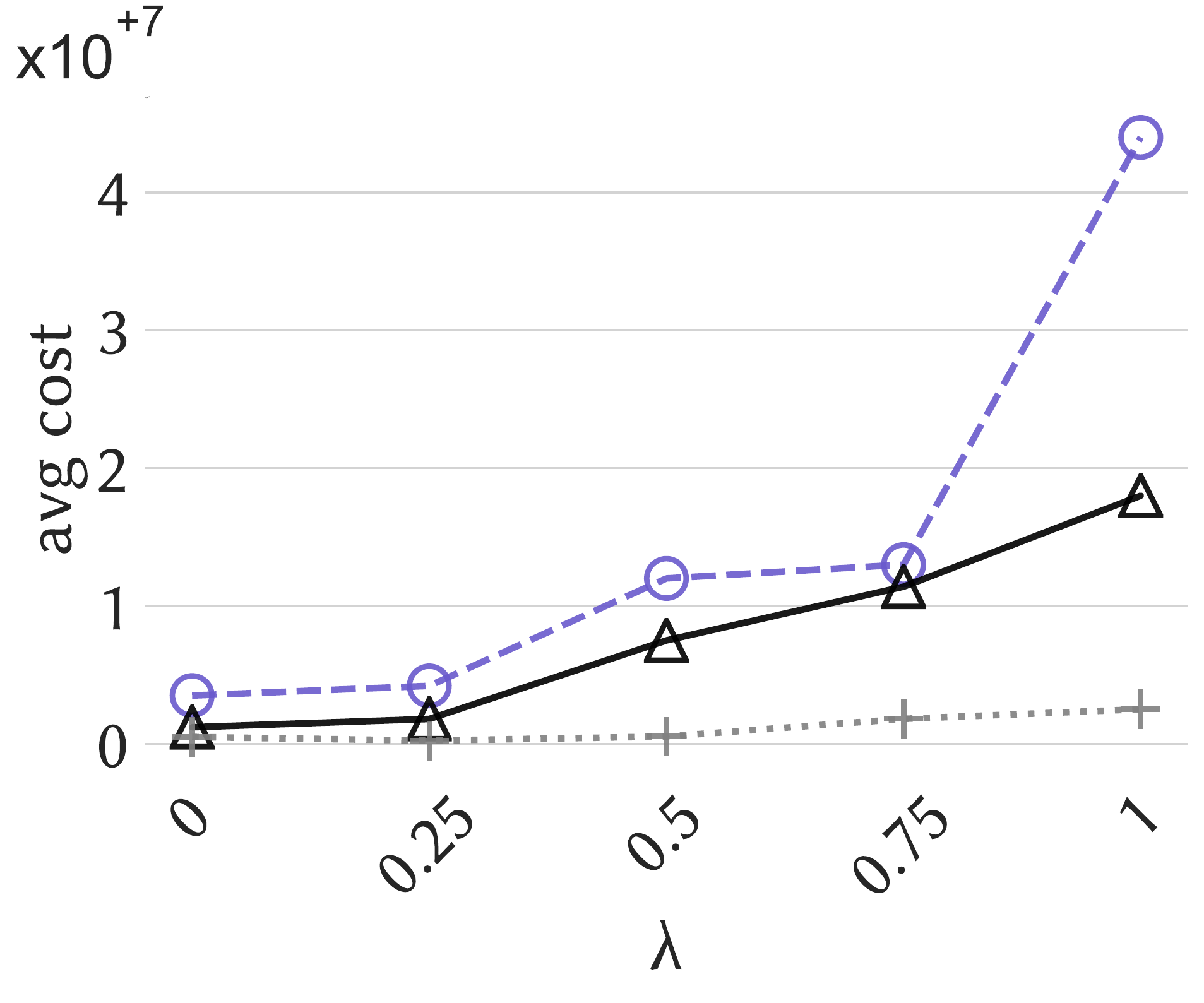}&
		\hspace{-5mm} \includegraphics[width=0.45\columnwidth, height = 0.55\textheight, keepaspectratio]{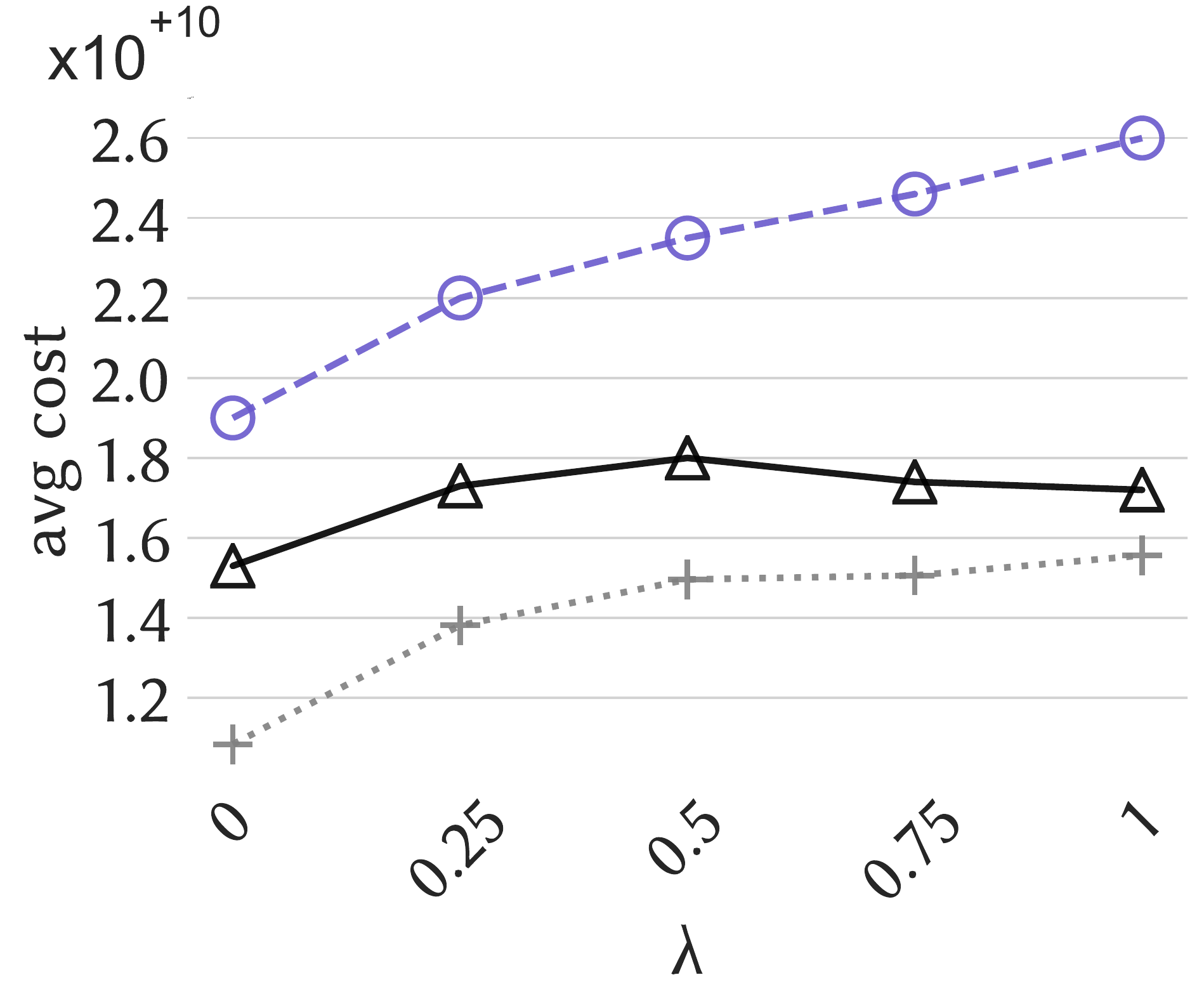}&
		\hspace{-5mm} \includegraphics[width=0.45\columnwidth, height = 0.55\textheight, keepaspectratio]{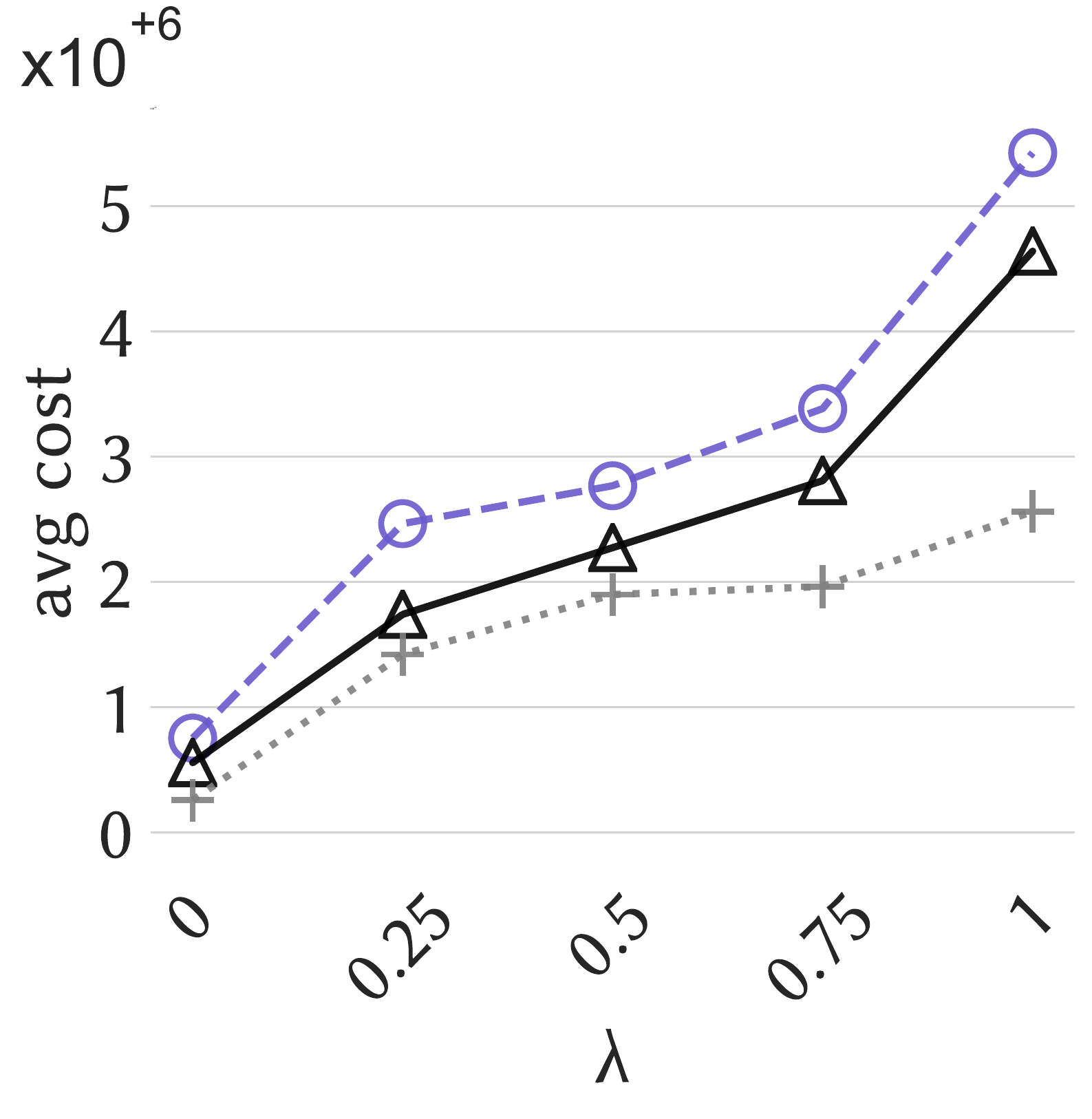}&
		\hspace{-5mm} \includegraphics[width=0.45\columnwidth, height = 0.55\textheight, keepaspectratio]{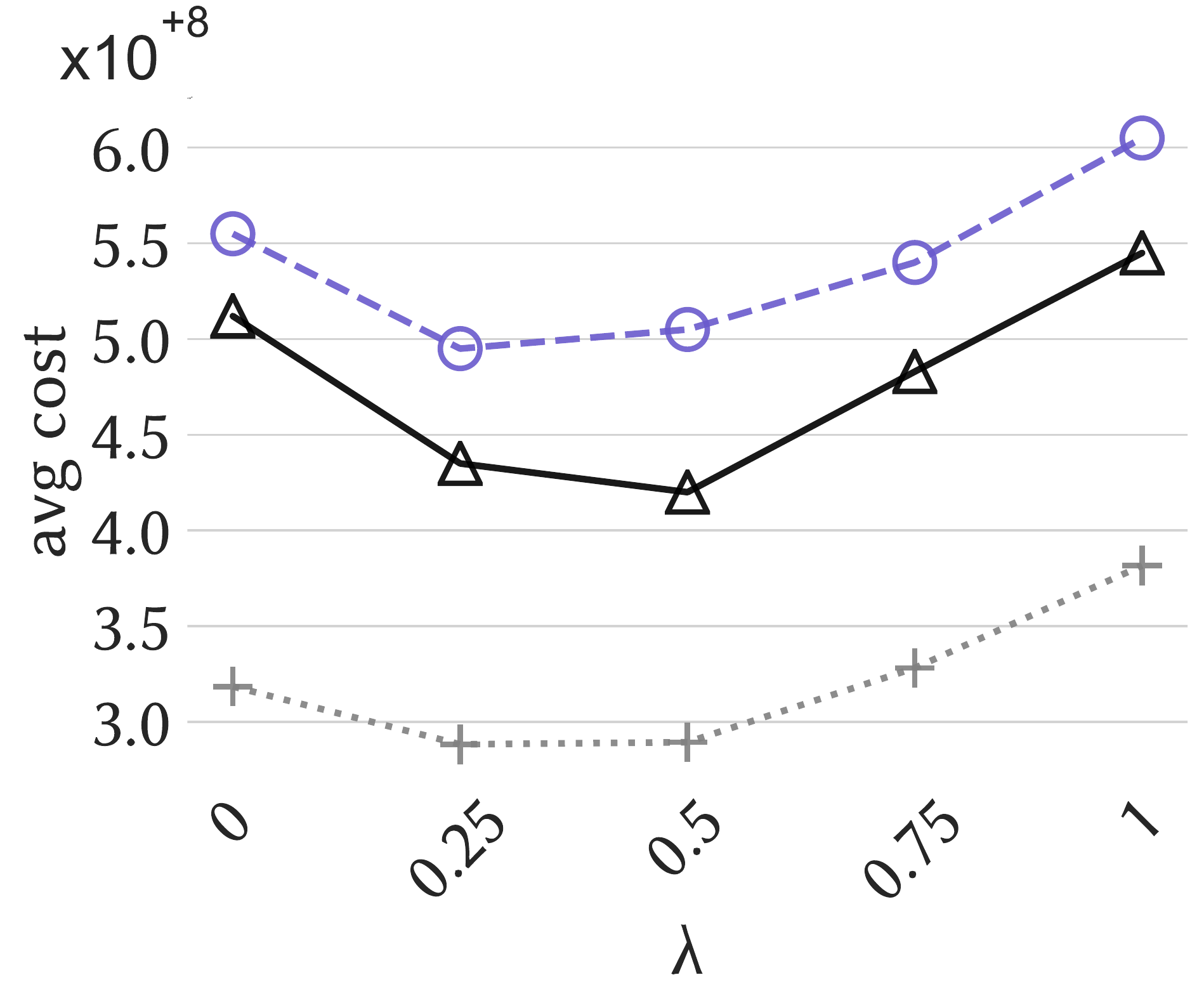}\\ 
		\textsc{TPC-H} &  \textsc{Munin} & \textsc{Pathfinder} &  \textsc{Barley} \vspace{0mm}\\ 
		\vspace*{-0.6cm}\\
	\end{tabular}
	\hspace{-5mm} \caption{\label{fig:robustness_skewed} 
			\revisioncol{Average cost of processing \querylog' against proportion $\lambda$ of queries from \querylog for the standard junction-tree algorithm (JT), \ouralgorithm and \ouralgorithmplus. Here, $\querylog$ is the skewed workload.}
	}
\end{figure*}

\begin{figure*}[t]
	\centering
	\includegraphics[width=1\columnwidth, height = 0.8\textheight, keepaspectratio]{plots_experiments/legend_robustnessPlus.png}\\
	\begin{tabular}{cccc}
		\hspace{-5mm} \includegraphics[width=0.45\columnwidth, height = 0.55\textheight, keepaspectratio]{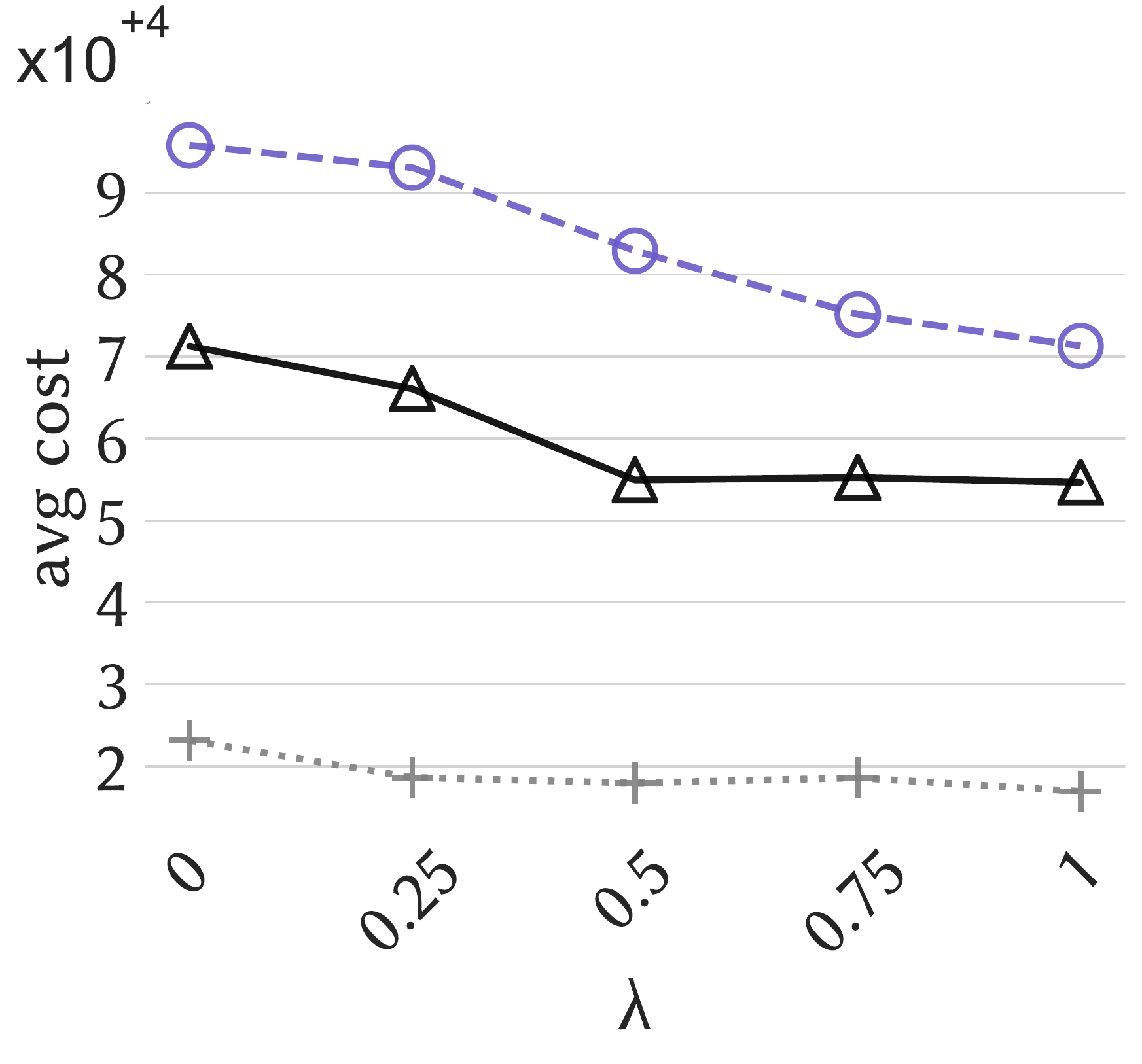}&
		\hspace{-5mm} \includegraphics[width=0.45\columnwidth, height = 0.55\textheight, keepaspectratio]{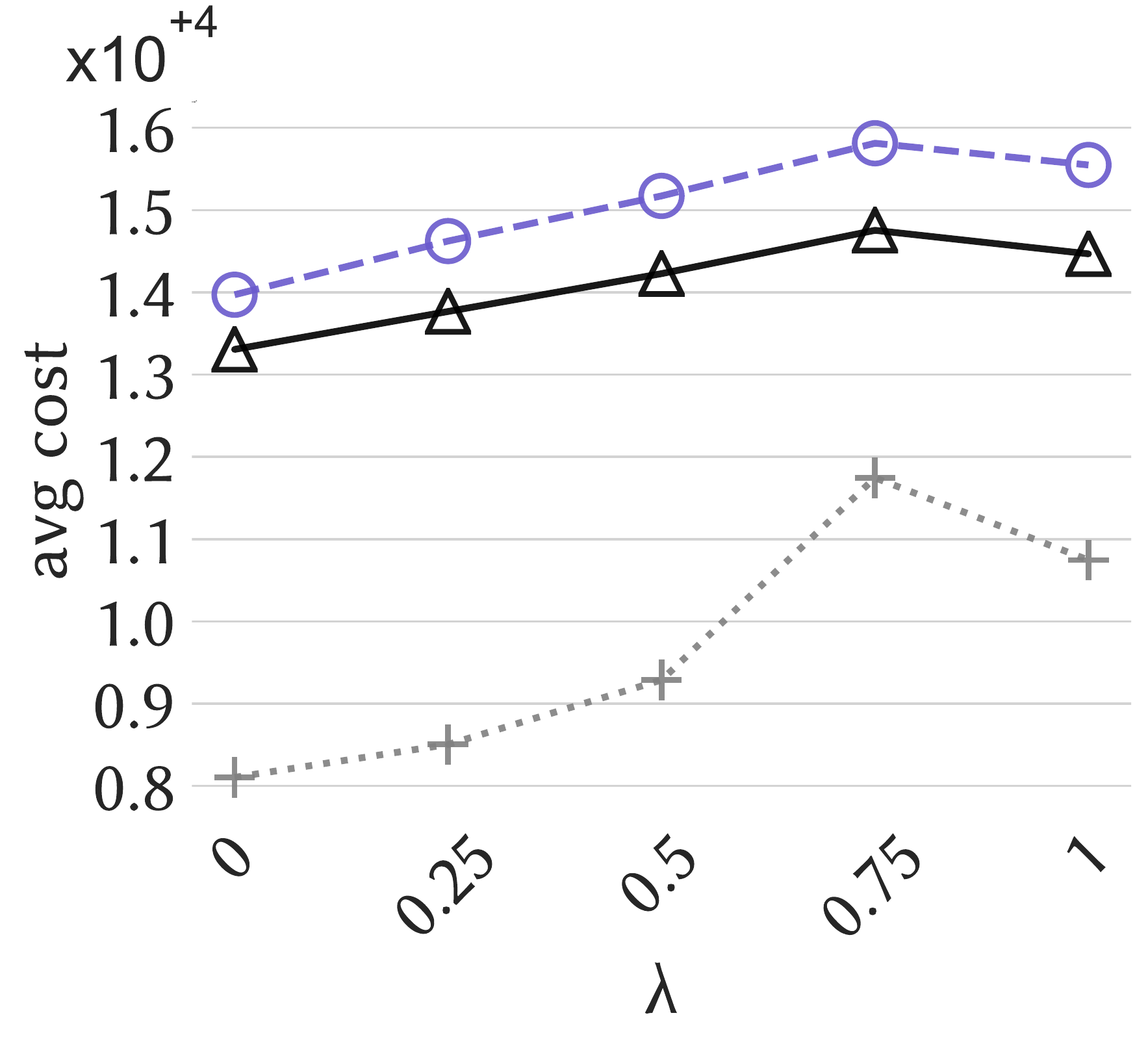}&
		\hspace{-5mm} \includegraphics[width=0.45\columnwidth, height = 0.55\textheight, keepaspectratio]{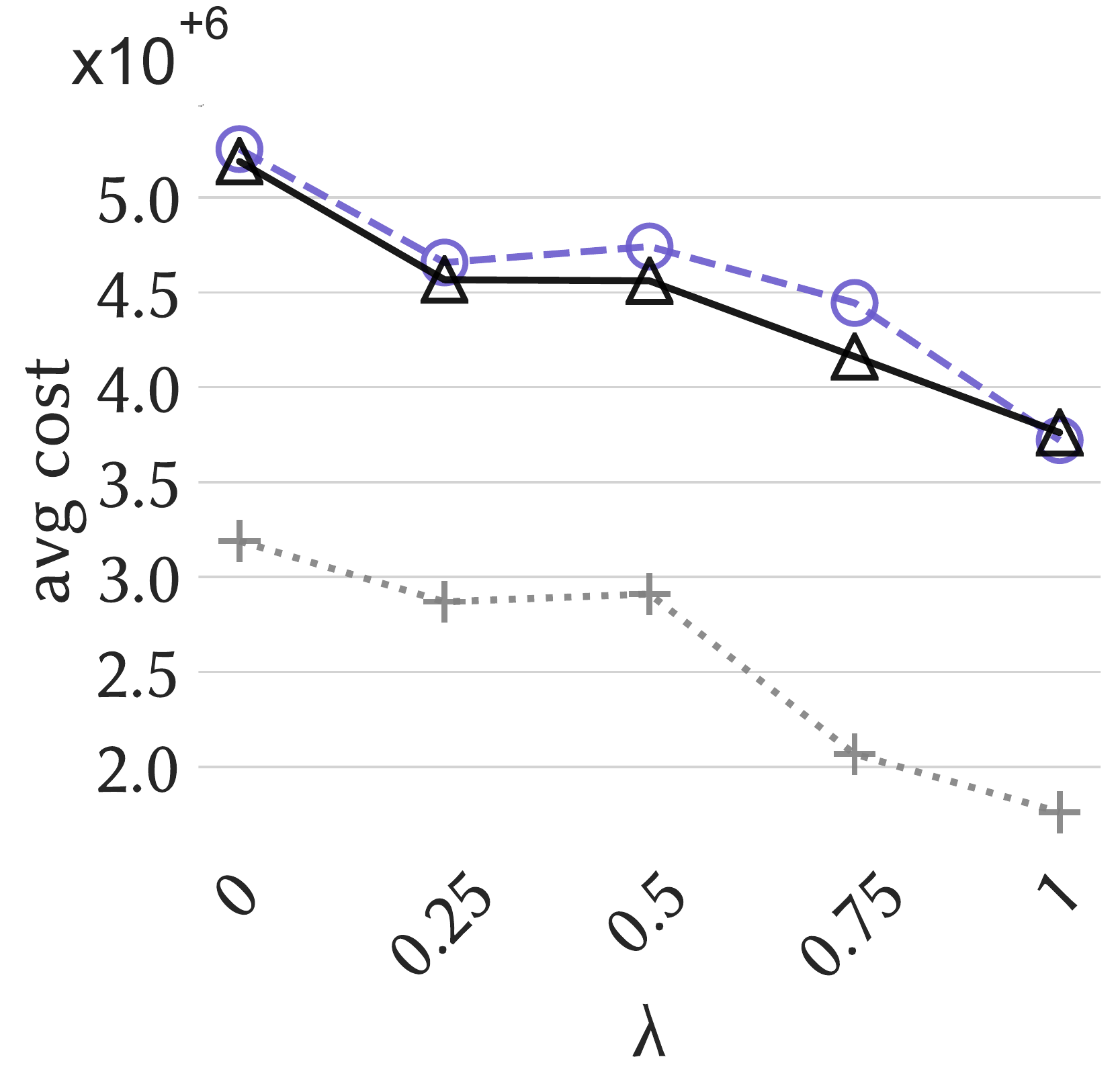}&
		\hspace{-5mm} \includegraphics[width=0.45\columnwidth, height = 0.55\textheight, keepaspectratio]{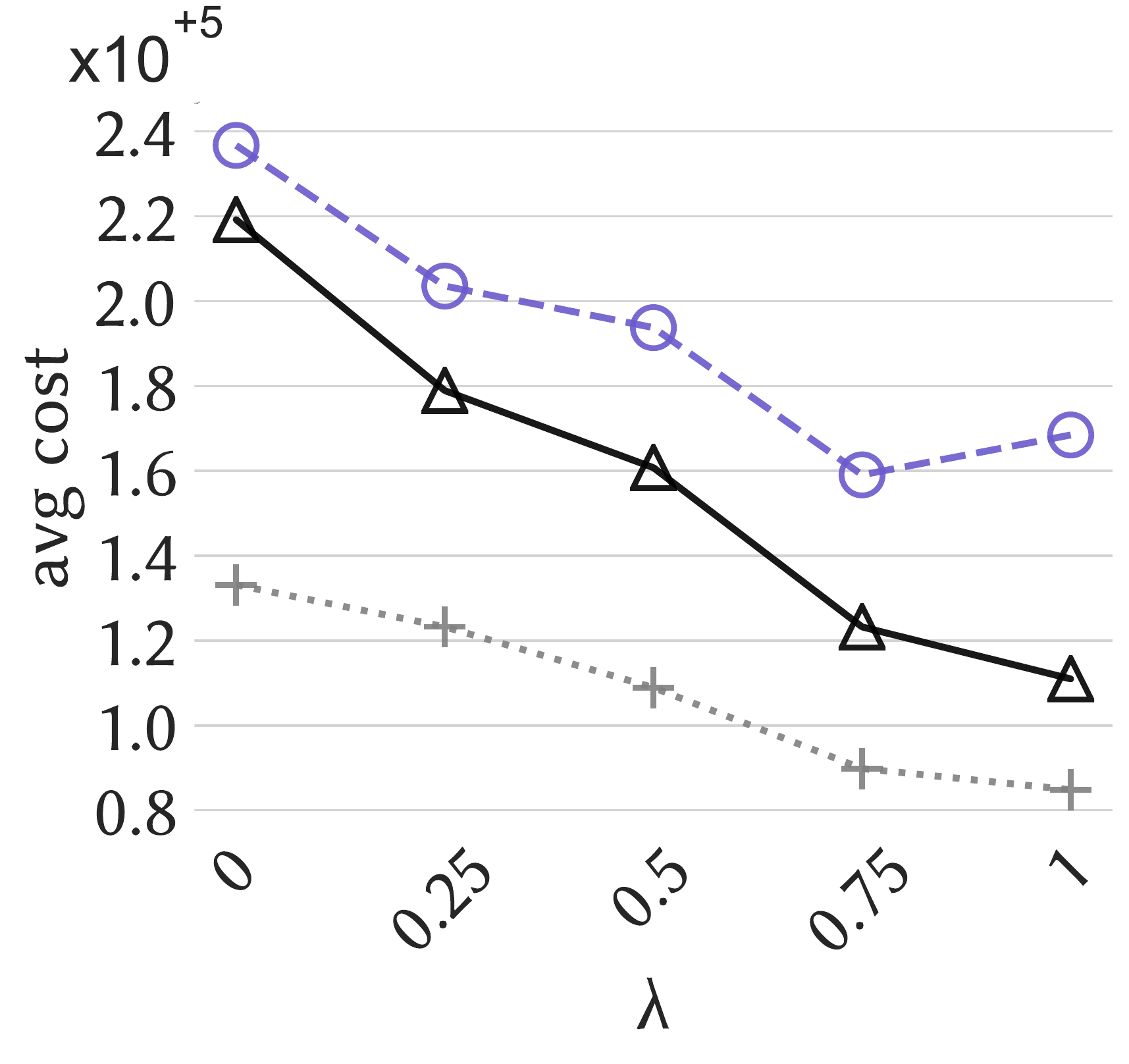} \\ 
		\textsc{Child} &  \textsc{Hepar II} & \textsc{Andes} & \textsc{HailFinder} \vspace{-1mm}\\
		\hspace{-5mm} \includegraphics[width=0.45\columnwidth, height = 0.55\textheight, keepaspectratio]{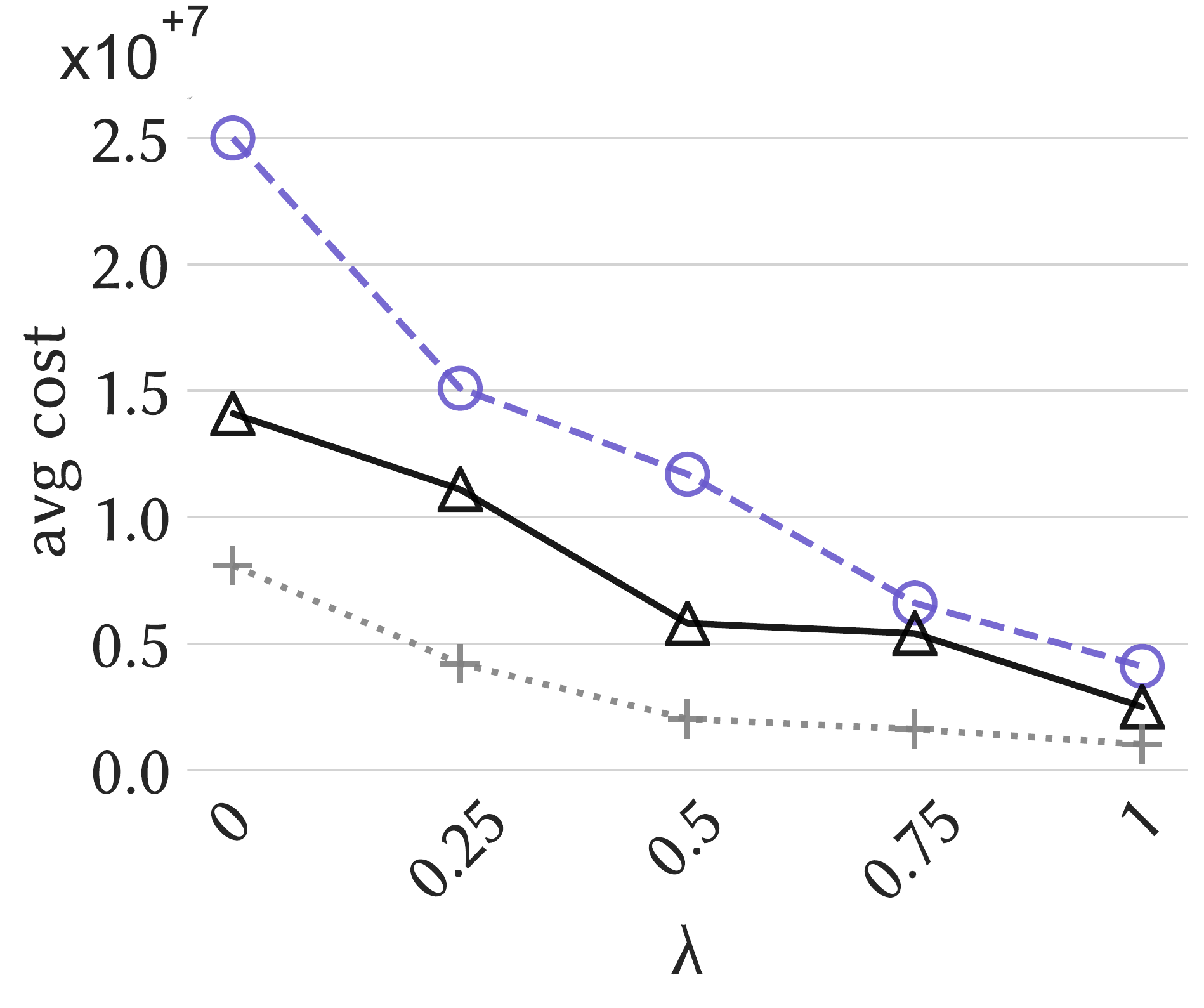}&
		\hspace{-5mm} \includegraphics[width=0.45\columnwidth, height = 0.55\textheight, keepaspectratio]{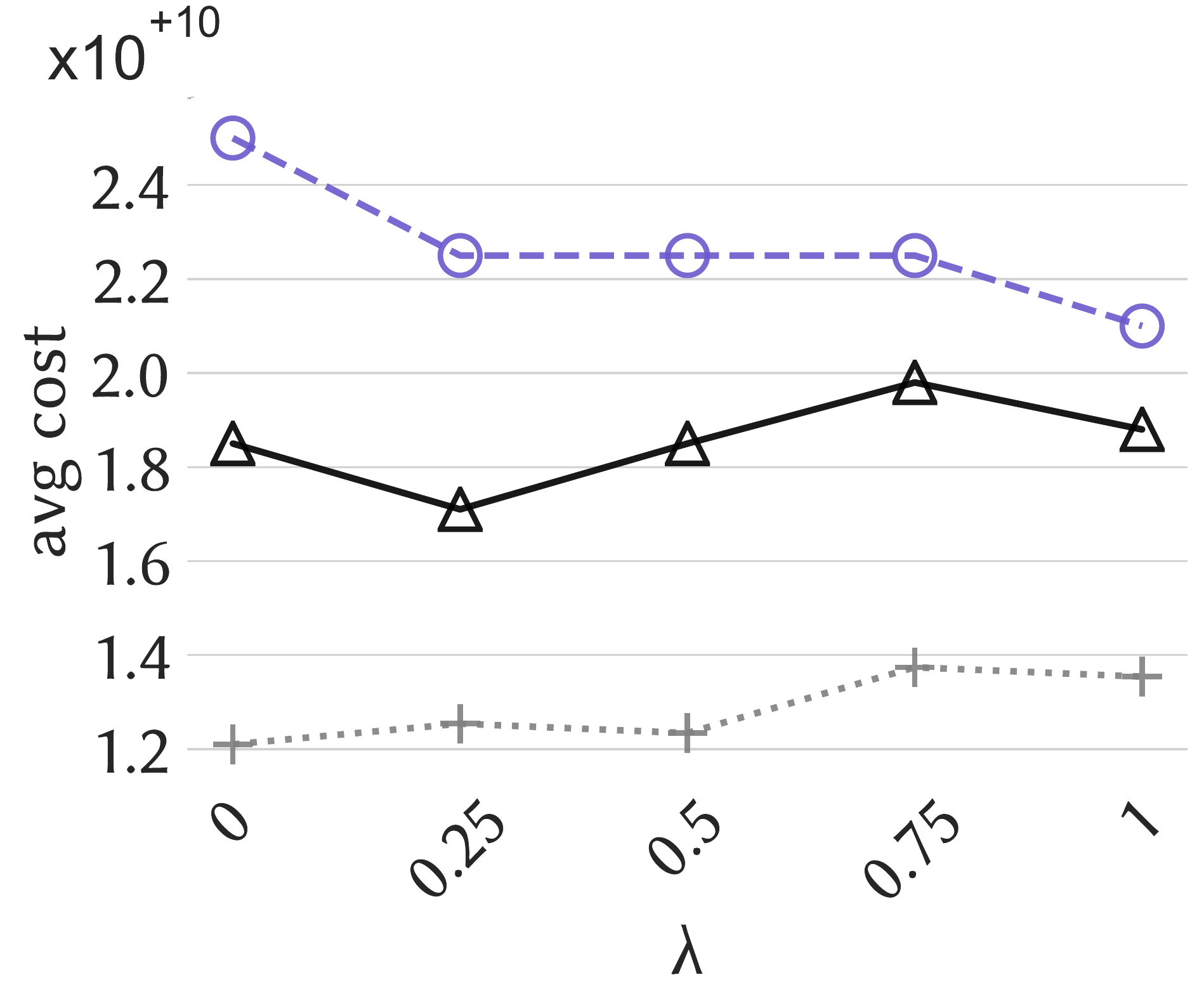}&
		\hspace{-5mm} \includegraphics[width=0.45\columnwidth, height = 0.55\textheight, keepaspectratio]{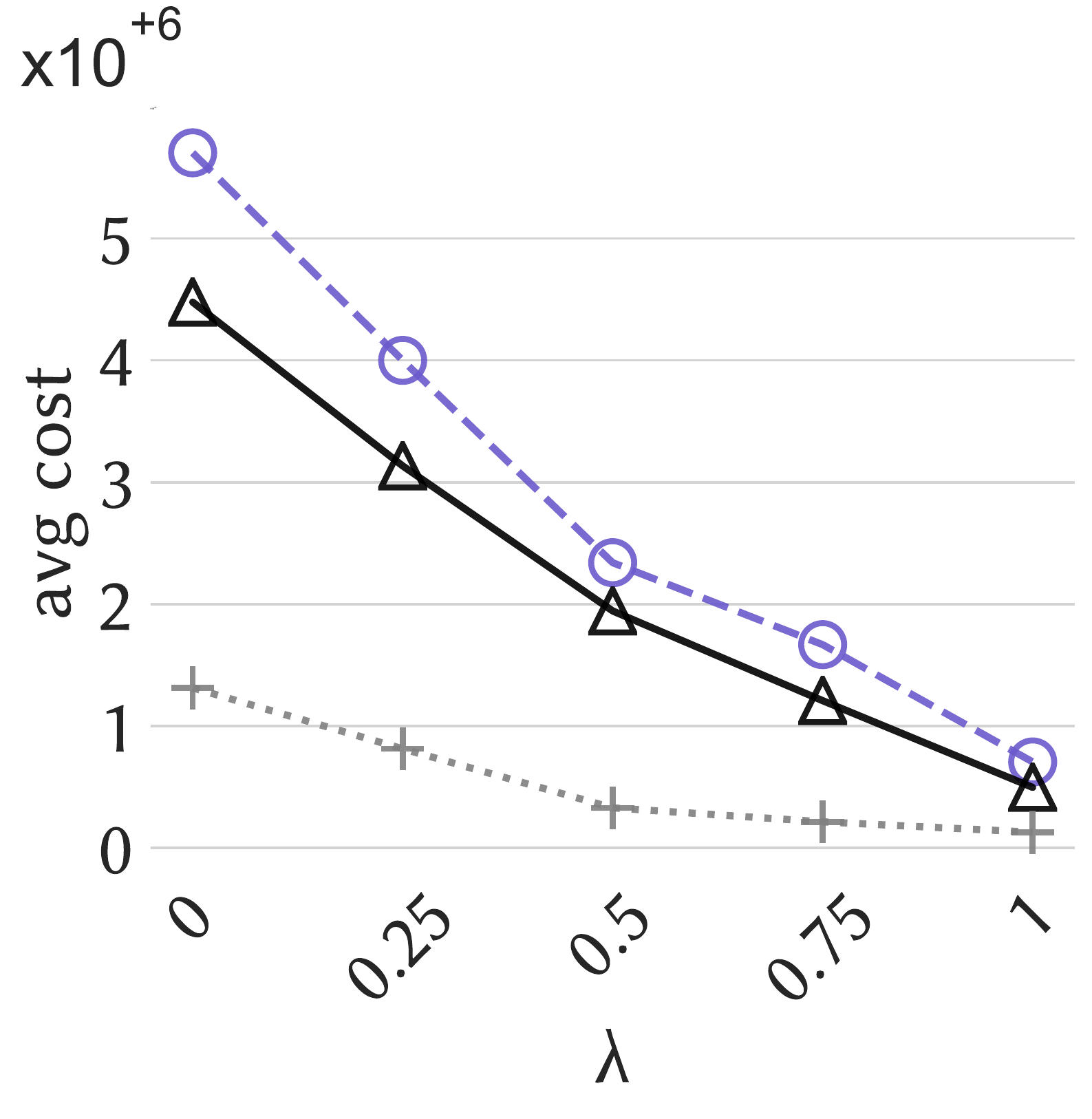}&
		\hspace{-5mm} \includegraphics[width=0.45\columnwidth, height = 0.55\textheight, keepaspectratio]{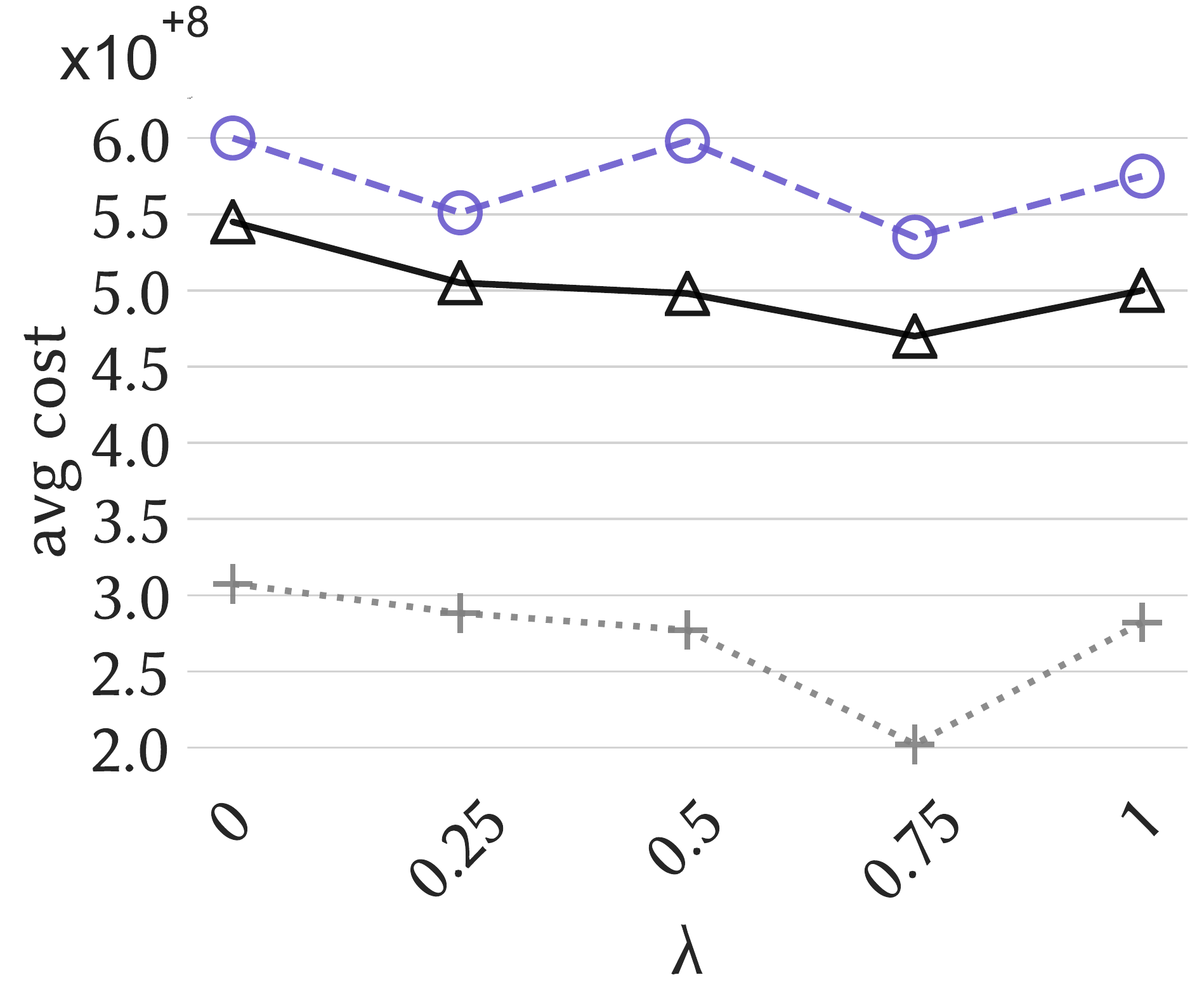}\\ 
		\textsc{TPC-H} &  \textsc{Munin} & \textsc{Pathfinder} &  \textsc{Barley} \vspace{0mm}\\ 
		\vspace*{-0.2cm}\\
	\end{tabular}
	\hspace{-5mm} \caption{\label{fig:robustness_uniform} 
		\revisioncol{Average cost of processing \querylog' against proportion $\lambda$ of queries from \querylog for the standard junction-tree algorithm (JT), \ouralgorithm and \ouralgorithmplus. Here, $\querylog$ is the uniform workload.} 
	}
\end{figure*}

\begin{figure*}[t]
	\centering
	\includegraphics[width=1.2\columnwidth, height = 0.6\textheight, keepaspectratio]{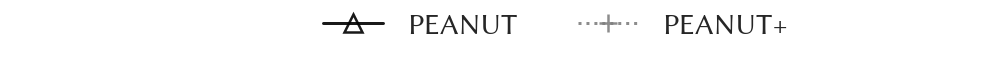}\\
	\vspace*{-0.4cm}
	\begin{tabular}{cccc}
		\hspace{-5mm} \includegraphics[width=0.45\columnwidth, height = 0.55\textheight, keepaspectratio, keepaspectratio]{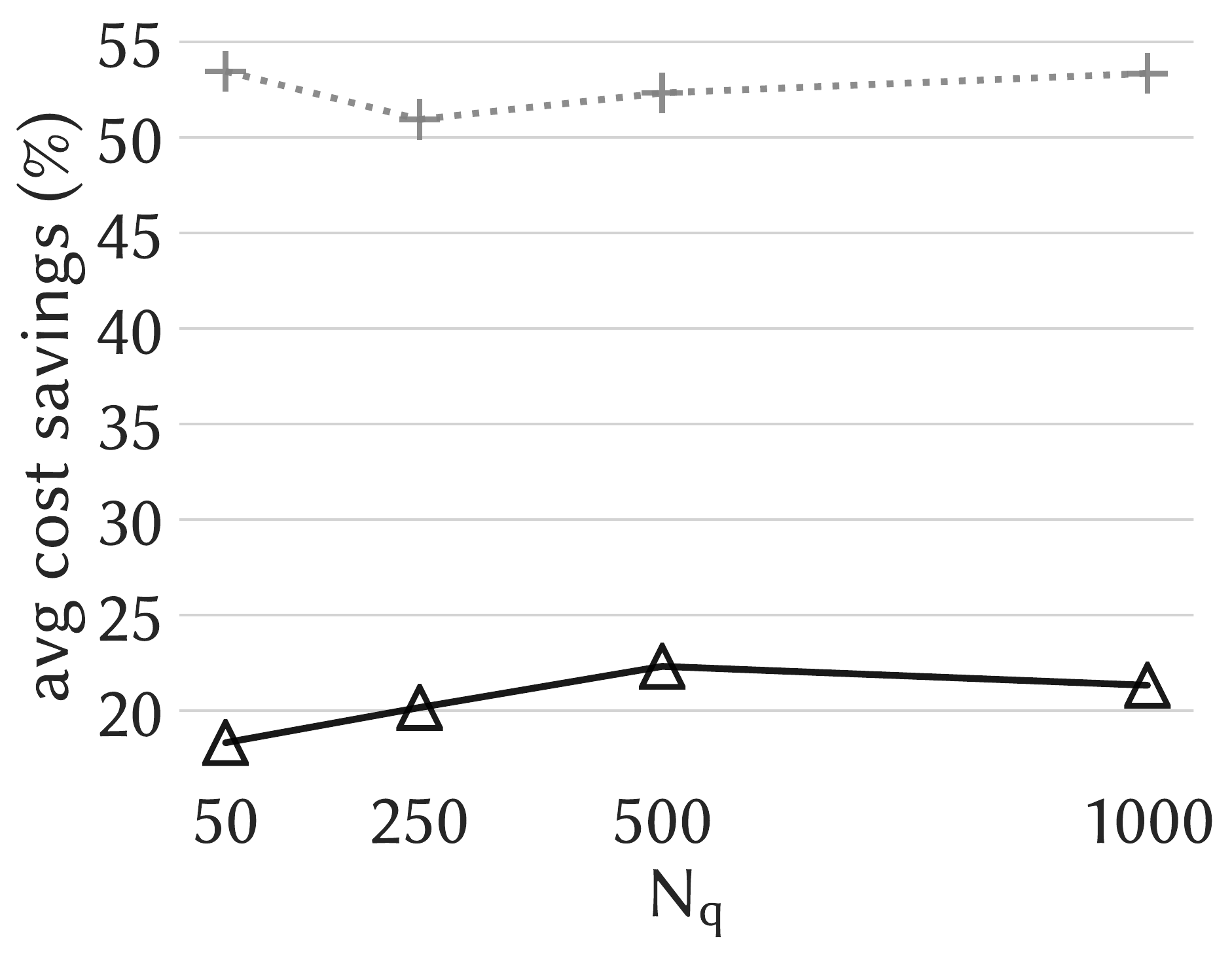}&
		\hspace{-5mm} \includegraphics[width=0.45\columnwidth, height = 0.55\textheight, keepaspectratio, keepaspectratio]{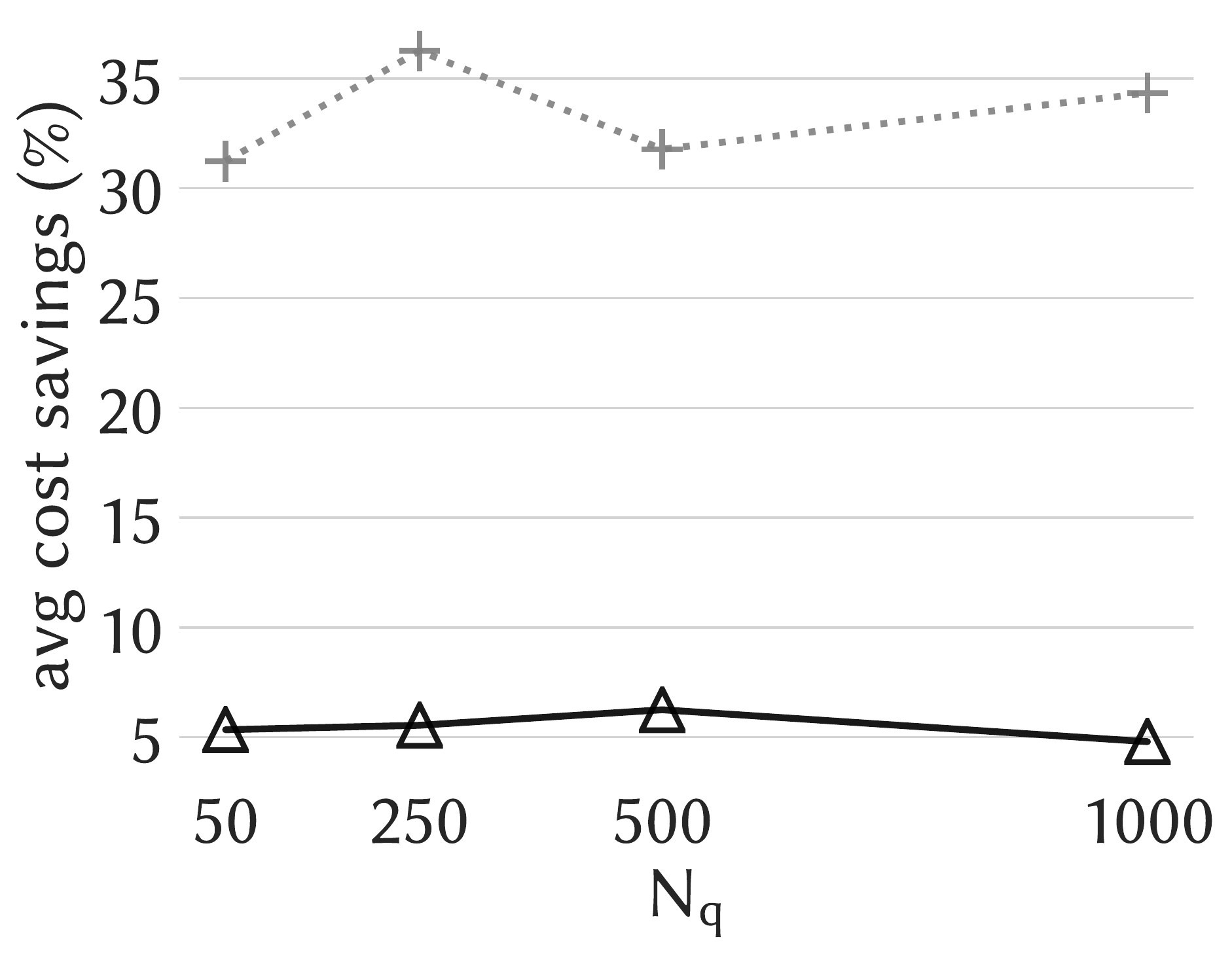}&
		\hspace{-5mm} \includegraphics[width=0.45\columnwidth, height = 0.55\textheight, keepaspectratio]{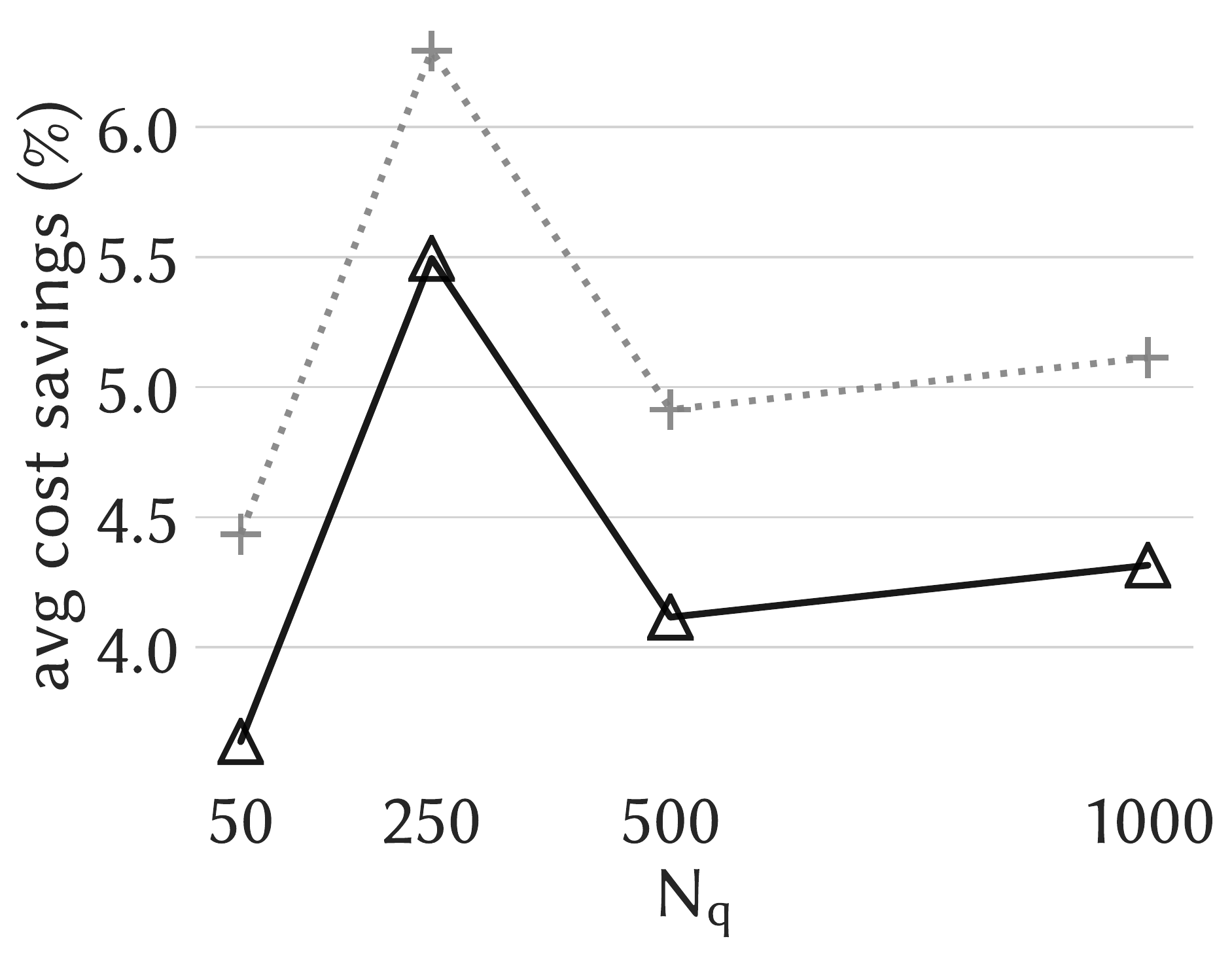}&
		\hspace{-5mm} \includegraphics[width=0.45\columnwidth, height = 0.55\textheight, keepaspectratio]{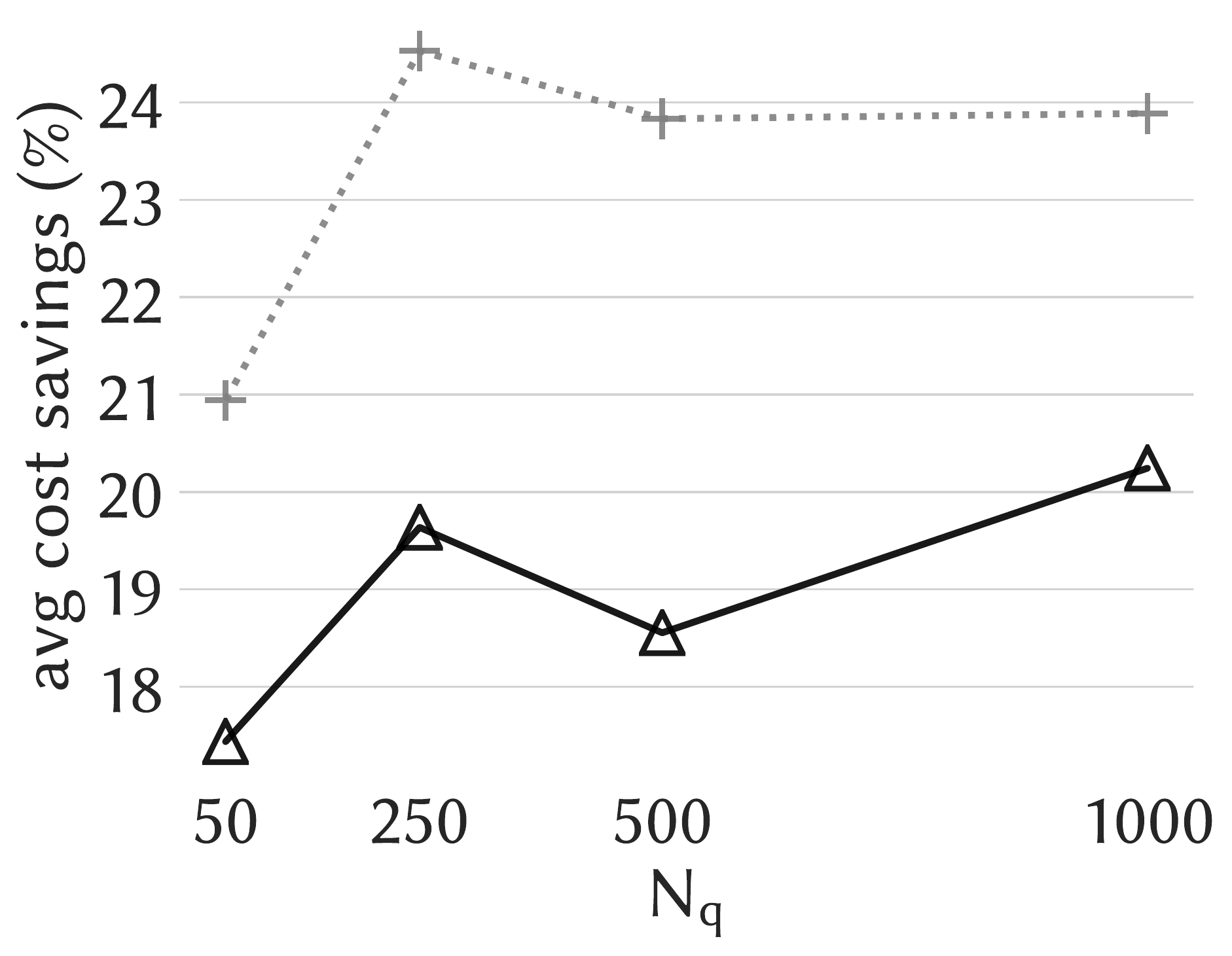}\\ 
		\textsc{Child} &  \textsc{Hepar II} & \textsc{Andes} & \textsc{Hailfinder} \hspace{-5mm}\\
		\hspace{-5mm} \includegraphics[width=0.45\columnwidth, height = 0.55\textheight, keepaspectratio]{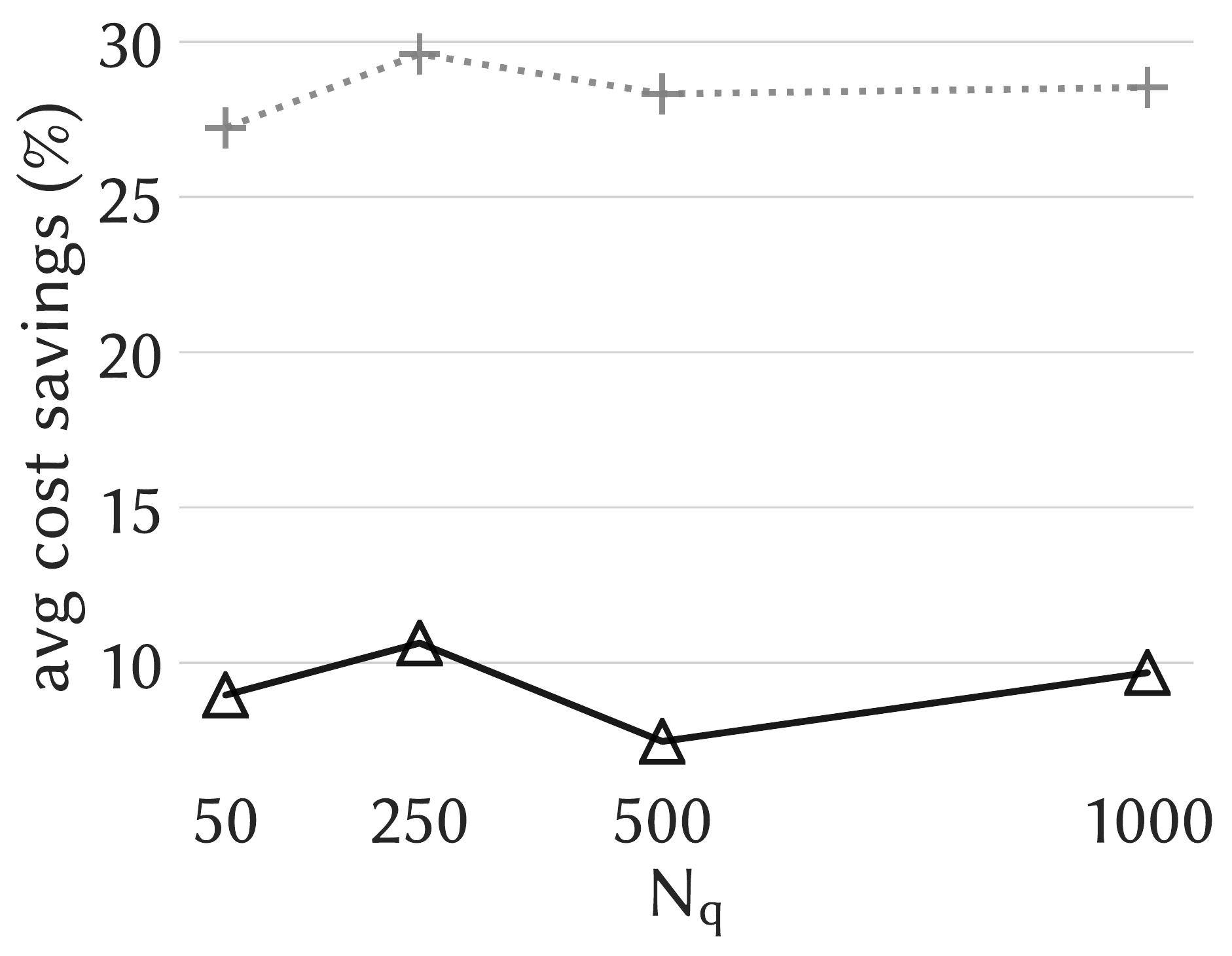} &
		\hspace{-5mm} \includegraphics[width=0.45\columnwidth, height = 0.55\textheight, keepaspectratio]{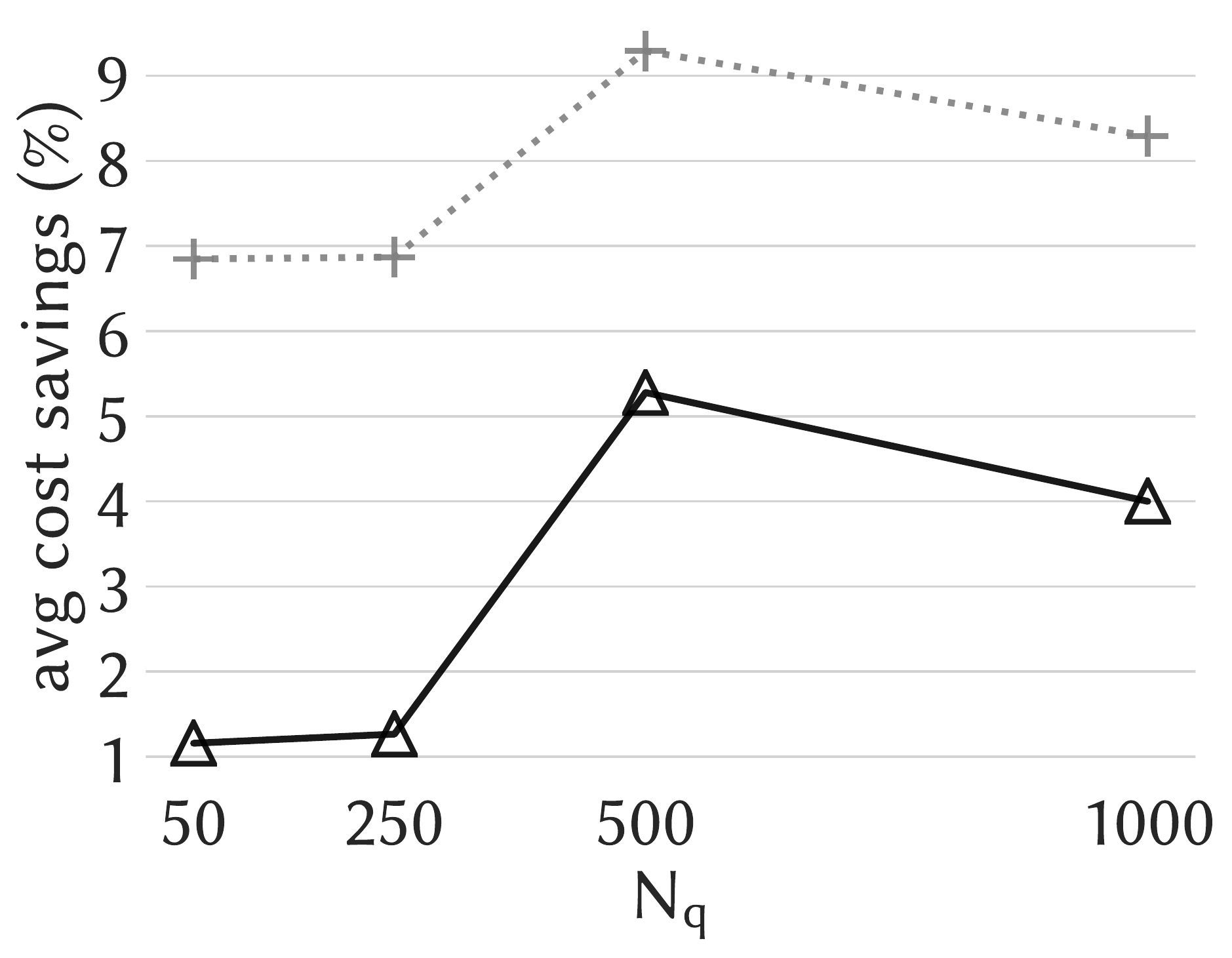} &
		\hspace{-5mm} \includegraphics[width=0.45\columnwidth, height = 0.55\textheight, keepaspectratio]{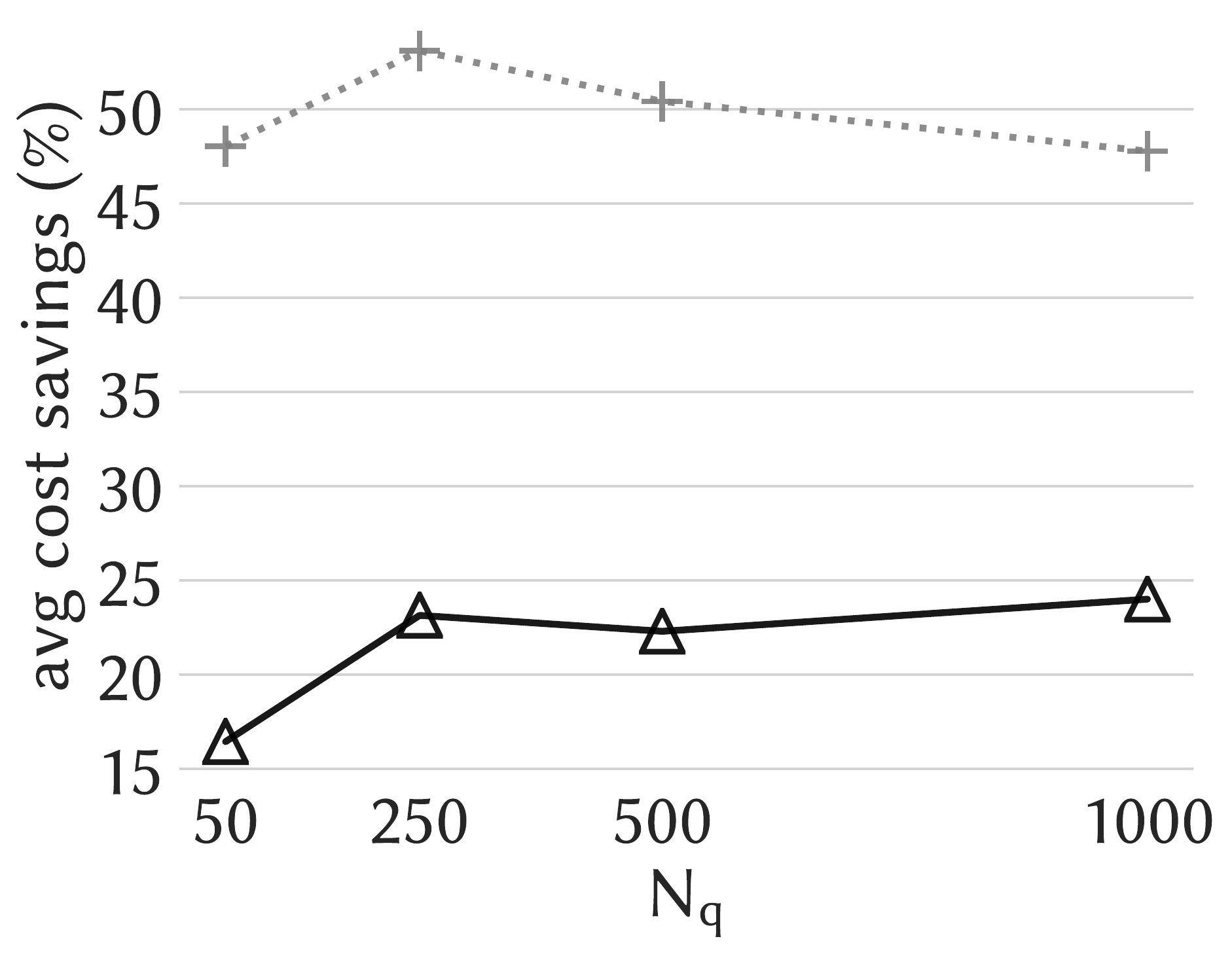} &
		\hspace{-5mm} \includegraphics[width=0.45\columnwidth, height = 0.55\textheight, keepaspectratio]{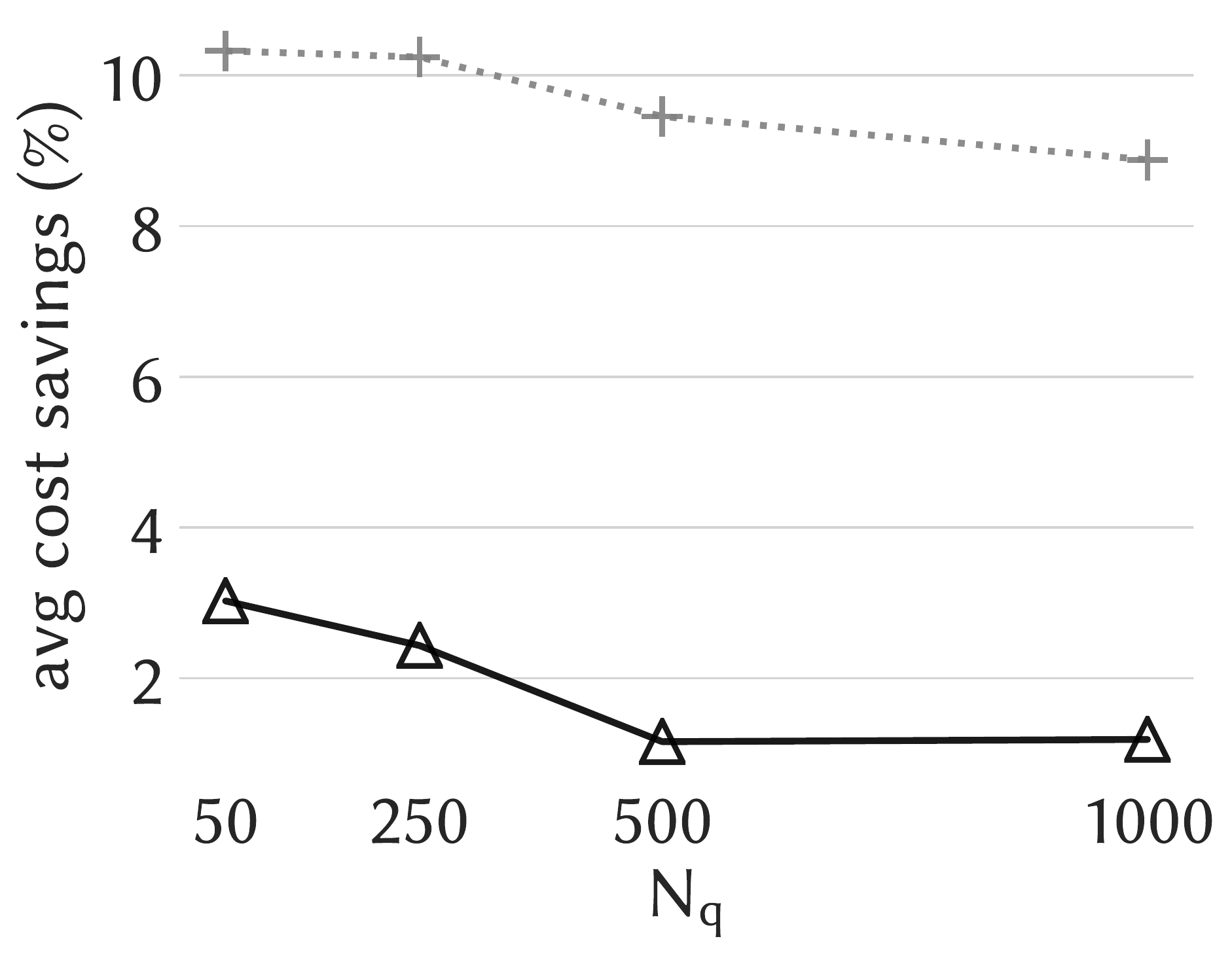}\\ 
		\textsc{TPC-H}  & \textsc{Munin} & \textsc{Pathfinder}   & \textsc{Barley}  \vspace*{-0.3cm}\\
	\end{tabular}
	\caption{\label{fig:impact_nq} 
		\revisioncol{Average cost savings percentage for \ouralgorithm and \ouralgorithmplus against query log size $N_q$.} 
	}
\end{figure*}

\spara{Implementation}. 
Experiments are executed on a machine with 
\textsc{2$\times$10 core Xeon E5 2680 v2 2.80\,GHz} processor and \textsc{256\,GB} memory.
All methods have been implemented in Python.
Our implementation is available online.\footnote{\url{https://github.com/martinociaperoni/PEANUT}}



\begin{table}[t]
	\setlength\tabcolsep{1.0pt}
	\fontsize{6.0}{7.0}\selectfont
	\centering
	\caption{\label{table:offlineStats} \revisioncol{Materialization phase statistics for the experiments comparing \ouralgorithm, \ouralgorithmplus, \indsep and \qtm{n}.}}
	\begin{tabular}
		{lrrrrrrrrrr}
		\toprule
		& \multicolumn{5}{c}{Disk Space (MB)} & \multicolumn{5}{c}{Time (seconds)}  \\
		\cmidrule(lr{0.35em}){2-6}
		\cmidrule(lr{0.35em}){7-11}
		dataset & VE-5 & JT & INDSEP & PEANUT & PEANUT+ & VE-5 & JT & INDSEP  & PEANUT & PEANUT+ \\
		\textsc{Child} & $0.025$ & $0.005$ & $0.002$ & $0.002$ & $ 0.044$  & $0.020$ & $0.035$  & $0.016$ & $34.12$  & $23.013$ \\
		\textsc{Hepar II} & $0.025$ & $0.013$ & $0.011$ & $0.001$ & $1.55$ & $0.040$ & $0.37$  & $0.12$ & $493.86$ & $313.21$  \\ 
		\textsc{Andes} & $0.62$& $70$ & $78$ & $55.16$ & $17.7$\,K & $0.84$ & $3.7$\,K & $2.55$ & $31$\,K & $20$\,K\\ 
		\textsc{Hailfinder}& $0.025$ & $0.071$ & $0.11$ & $0.13$ & $8.28$ & $0.041$ & $0.93$ & $0.056$ & $199.32$ & $78.96$ \\
		\textsc{TPC-H}* & $0.041$ & $23.77$ & $43.19$ & $1.23$ & $483.78$ & $2.76$ & NA & $0.027$ & $320.71$ & $57.83$ \\
		\textsc{Munin}*  &  $265$ & $2.8$\,K & $361.95$ & $358.16$ & $1$\,K & $253.67$  & NA & $2.57$ & $25.6$\,K & $16.9$\,K \\
		\textsc{Pathfinder}  & $0.06$ & $1.37$ & $1.40$ & $0.14$ & $220.79$ & $0.31$ & $302$ & $2.80$ & $1$\,K & $420.32$ \\ 
		\textsc{Barley}* & $24$ & $186.99$ & $3$\,K & $2.2$\,K & $13.2$\,K  & $11.57$ & NA & $0.034$ & $635.051$ & $263.57$  \\ 
	\end{tabular}
\end{table}

\vspace*{1cm}
\subsection{Query-processing results}
\label{sec:experiments_results}

We first report results on the cost of processing inference queries. 

\spara{Skewed workload}. In Figure~\ref{fig:distributions_cost_savings}\revisioncol{,}  we show the performance gains 
due to materialization of shortcut potentials achieved by \indsep and \ouralgorithmplus with different approximation levels when processing the skewed workload of queries.
The performance gains are given by the net cost savings 
obtained by resorting to the shortcut potentials, 
compared to the case when no shortcut potential is used. 
The results confirm that 
\ouralgorithmplus typically outperforms \indsep, even for large values of the approximation parameter $\epsilon$. 
As expected, for \ouralgorithmplus, 
the lowest value of $\epsilon = 1.2$ typically yields the greatest cost savings, 
while $\epsilon = 12$ tends to result in the worst performance. 
\revisioncol{In general, although it is not guaranteed, for both \ouralgorithm and \ouralgorithmplus, increasing $\epsilon$ tend to decrease savings. }
Next, we explore the relation between the net cost savings and the diameter of the Steiner tree extracted for query answering. 
Figure~\ref{fig:diameter} displays the average net cost savings in the skewed workload as a function of the Steiner-tree diameter for \indsep\ReviewOnly{\revisioncol{, \ouralgorithmplus}} and \ouralgorithm with different values of the approximation parameter $\epsilon$. Here, for each query we take the maximum savings over the considered budgets for \indsep, \ouralgorithm and \ouralgorithmplus.  
\FullOnly{The figure reveals that savings grow quickly with the Steiner-tree diameter, because queries which span a larger part of the tree are more likely to hit materialized shortcut potentials. }
\ReviewOnly{\revisioncol{While the figure does not offer a fair comparison between candidate strategies since savings are obtained with largely varying budgets, it reveals that savings grow quickly with the Steiner-tree diameter simply because larger Steiner trees lead to a higher shortcut potential hit rate. The increasing trend is common to all datasets, although the growth rate depends heavily on the characteristics of the materialized shortcut potentials. }}

\vspace{1mm}
\noindent
{\bf Uniform workload}.  In Figure~\ref{fig:comparison_qtm} we compare, in the uniform workload, the total cost for inference based on the junction tree (JT) without additional materialization, \indsep, \ReviewOnly{\revisioncol{ \qtm{n} with $n=5$ materialized factors,  \ouralgorithm and \ouralgorithmplus for fixed $\epsilon = 1.2$ }}. 
\FullOnly{, \ouralgorithmplus with $\epsilon=1.2$, and \qtm{n} with $n=5$ materialized factors} 
In addition, 
we analyse how the cost varies with the query size~$|q|$.  \revisioncol{As expected, larger queries are more computationally demanding because they yield larger Steiner trees and larger messages. }
Clearly, for $|q|=1$, \qtm{n} performs remarkably worse than the other methods. \ouralgorithmplus always leads to the best query processing except for the \textsc{Munin} dataset, in which junction-tree-based approaches exhibit very poor performance. Similarly, \ouralgorithm outperforms \indsep in all datasets and  \qtm{5} in five datasets.

\revisioncol{Finally, it is important to observe that the scale of savings across different datasets is highly variable. 
In particular, the effectiveness of materialization depends largely on the tree structure.
As it can be confirmed with a simple regression analysis of the average cost on characteristics of the datasets, for a fixed number of cliques, the benefits of our methods are larger in sparser Bayesian networks leading to junction trees with limited treewidth. Similarly, the diameter of the junction tree appears to have a relevant positive effect on the performance of \ouralgorithm. }

\revisioncol{\subsection{Robustness analysis}
We carry out additional experiments to explore the robustness of our system \ouralgorithm with respect to drifts in the query workload distribution. In particular, we explore the setting in which materialization is optimized with respect to a query workload $\querylog$ and then the results are evaluated on a different query workload $\querylog'$. Both $\querylog$ and $\querylog'$ are of size $500$. 
When $\querylog$ is the skewed (uniform) query workload considered in our experiments, 
$\querylog'$ consists of queries from the same skewed (uniform) workload in
proportion $\lambda$ and of queries from the uniform (skewed) workload 
in proportion
$1 - \lambda$, with varying $\lambda \in [0, 1]$. 
The average cost of processing $\querylog'$ with and without materialization for different values of $\lambda$ 
in the case that $\querylog$ is the skewed and the uniform workload are shown in Figures~\ref{fig:robustness_skewed} and~\ref{fig:robustness_uniform}, respectively.
Here, \ouralgorithm and \ouralgorithmplus parameters are fixed to $K = 10 b_{T}$ and $\epsilon = 1.2$. 
The cost savings in $\querylog'$ are generally not drastically affected by the value of $\lambda$. Thus, the results of this experiment highlight that \ouralgorithm and \ouralgorithmplus are fairly robust with respect to drifts in the query workload distribution. 
Note that the average cost tends to increase with the proportion of skewed queries. This is not surprising since skewed queries are more likely to correspond to Steiner trees of large diameter. This trend is observed in all datasets except for \textsc{HeparII}, in which, although skewed queries are as usual associated with larger Steiner tree diameters, they also favour query variables with smaller cardinality. 
Despite the robustness of the proposed approach, it is appropriate to recompute the materialization either periodically over time or 
whenever there is sufficient evidence that the query distribution has changed significantly. 


}

\revisioncol{
	\subsection{Impact of query log size}\label{sec:query_size}

	We examine how the number $N_q$ of queries \query in the query log \querylog, used to drive the offline optimization step, impact the performance of \ouralgorithm and \ouralgorithmplus. 
	More specifically, we vary $N_q$  in $\{ 50, 250, 500, 1000 \}$ to examine its effect and we fix \ouralgorithm and \ouralgorithmplus parameters to $\epsilon = 6.0$ and $K = 10 b_{T}$. While the size of the query log $N_q$ used in the optimization is varied, the size of the query log $\querylog'$ used for testing the performance of the methods is held constant at $N'_q = 1000$. 
	We expect the parameter $N_q$ to have a not significant impact on the performance of \ouralgorithm and \ouralgorithmplus. This follows since such methods do not use \querylog to learn a complex function, but only to learn variable probabilities via empirical frequencies. 
	Provided that the size of \querylog is large enough, the empirical frequencies deliver accurate estimates of the variable probabilities and small differences in probability estimates do not have a significant effect on \ouralgorithm and \ouralgorithmplus. 
	Moreover, our robustness analysis suggests that the performance of our methods does not remarkably declines even in the in the scenario in which the distribution of the empirical frequencies deviates consistently from the target distribution of variable probabilities. 
	Figure~\ref{fig:impact_nq} shows results of the experiments aimed at assessing the effect of $N_q$. 
	The average cost savings for \ouralgorithm and \ouralgorithmplus exhibit little variability over values of $N_q$, which confirms our expectations concerning the impact of $N_q$ on the performance of our methods.
}

\section{Conclusions}

We presented a novel technique to accelerate inference queries over Bayesian networks
using the junction-tree algorithm.
Our approach builds on the idea introduced by \citet{kanagal2009indexing} 
to materialize shortcut potentials over the junction tree. 
However, unlike their work, we have framed the problem of choosing shortcut potentials to precompute and 
materialize as a workload-aware optimization problem. 
In particular, we have formulated the problems of selecting a single optimal shortcut potential and 
an optimal set of shortcut potentials in view of evidence provided by historical query logs. 
We have proven hardness results for these problems, 
and we have developed a method, called \ouralgorithm,  
consisting of a two-level dynamic-programming framework that tackles the problems in pseudo-polynomial time. 
We have additionally introduced a strongly-polynomial algorithm
that trades solution quality with speed, as well as \ouralgorithmplus, a simple modification of \ouralgorithm, which ensures a better utilization of the available space budget. 
Extensive experimental evaluation confirms the effectiveness of our materialization methods 
in reducing the computational burden associated with a given query workload and the
superiority with respect to previous work. 
Interesting avenues to explore for future work include extending our techniques to allow for node-overlapping shortcut-potential subtrees in a more principled way, 
or studying a similar approach for optimal materialization of intermediate computations 
to speed up inference based on a different, possibly approximate, algorithm, 
such as \emph{loopy belief propagation}.

\balance
\bibliographystyle{ACM-Reference-Format}
\bibliography{main} 
\balance

\end{document}